\documentclass[12pt,british,letter]{amsart}
\usepackage[utf8]{inputenc}
\usepackage{geometry}
\usepackage{color}
\usepackage{babel}
\usepackage{verbatim}
\usepackage{prettyref}
\usepackage{mathtools}
\usepackage{enumitem}
\usepackage{amsbsy}
\usepackage{amstext}
\usepackage{amsthm}
\usepackage{amssymb}
\usepackage{graphicx,caption}
\usepackage{setspace}
\usepackage[authoryear]{natbib}
\usepackage[unicode=true,
 bookmarks=false,
 breaklinks=false,pdfborder={0 0 0},pdfborderstyle={},backref=false,colorlinks=false]
 {hyperref}
\hypersetup{
 unicode}

\makeatletter

\providecommand{\tabularnewline}{\\}

\numberwithin{equation}{section}
\numberwithin{figure}{section}
\numberwithin{table}{section}
\theoremstyle{remark}
\newtheorem*{notation*}{\protect\notationname}
\theoremstyle{plain}
\newtheorem{assumption}{\protect\assumptionname}
\theoremstyle{plain}
\newtheorem{thm}{\protect\theoremname}[section]
\theoremstyle{definition}
\newtheorem*{defn*}{\protect\definitionname}
\theoremstyle{remark}
\newtheorem{rem}{\protect\remarkname}[section]
\theoremstyle{plain}
\newtheorem{cor}{\protect\corollaryname}[section]
\theoremstyle{plain}
\newtheorem{lem}{\protect\lemmaname}[section]
\theoremstyle{definition}
\newtheorem*{urtest}{\protect\urtestname}

\usepackage{babel}
\usepackage{caption}
\usepackage{subcaption}
\usepackage{diagbox}

\def\@secnumfont{\bfseries}
\def\section{\@startsection{section}{1}%
\z@{1.0\linespacing\@plus\linespacing}{0.5\linespacing}%
{\normalfont\bfseries\centering}}

\newtheorem{example}{Example}

\usepackage{scalefnt}
\newcommand{\ass}[1]{{\upshape{\scalefont{0.76}{#1}}}}

\allowdisplaybreaks[1]

\setlist[enumerate,1]{label=\upshape{(\roman*)}, ref=(\roman*)}
\setlist[enumerate,2]{label=\upshape{(\alph*)}, ref=(\alph*)}
\setlist[enumerate,3]{label=\upshape{\roman*.}, ref=\roman*}

\AtBeginDocument{\providecommand\assref[1]{\ref{ass:#1}}}
\AtBeginDocument{\providecommand\enuref[1]{\ref{enu:#1}}}

\providecommand{\assumptionname}{Assumption}
\providecommand{\corollaryname}{Corollary}
\providecommand{\definitionname}{Definition}
\providecommand{\lemmaname}{Lemma}
\providecommand{\notationname}{Notation}
\providecommand{\remarkname}{Remark}
\providecommand{\theoremname}{Theorem}

\makeatother

\providecommand{\assumptionname}{Assumption}
\providecommand{\corollaryname}{Corollary}
\providecommand{\definitionname}{Definition}
\providecommand{\lemmaname}{Lemma}
\providecommand{\notationname}{Notation}
\providecommand{\remarkname}{Remark}
\providecommand{\theoremname}{Theorem}

\providecommand{\urtestname}{Unit root test procedure}

\begin{document}


\global\long\def\uwrite#1#2{\underset{#2}{\underbrace{#1}} }%

\global\long\def\blw#1{\ensuremath{\underline{#1}}}%

\global\long\def\abv#1{\ensuremath{\overline{#1}}}%

\global\long\def\vect#1{\mathbf{#1}}%


\global\long\def\smlseq#1{\{#1\} }%

\global\long\def\seq#1{\left\{  #1\right\}  }%

\global\long\def\smlsetof#1#2{\{#1\mid#2\} }%

\global\long\def\setof#1#2{\left\{  #1\mid#2\right\}  }%


\global\long\def\goesto{\ensuremath{\rightarrow}}%

\global\long\def\ngoesto{\ensuremath{\nrightarrow}}%

\global\long\def\uto{\ensuremath{\uparrow}}%

\global\long\def\dto{\ensuremath{\downarrow}}%

\global\long\def\uuto{\ensuremath{\upuparrows}}%

\global\long\def\ddto{\ensuremath{\downdownarrows}}%

\global\long\def\ulrto{\ensuremath{\nearrow}}%

\global\long\def\dlrto{\ensuremath{\searrow}}%


\global\long\def\setmap{\ensuremath{\rightarrow}}%

\global\long\def\elmap{\ensuremath{\mapsto}}%

\global\long\def\compose{\ensuremath{\circ}}%

\global\long\def\cont{C}%

\global\long\def\cadlag{D}%

\global\long\def\Ellp#1{\ensuremath{\mathcal{L}^{#1}}}%


\global\long\def\naturals{\ensuremath{\mathbb{N}}}%

\global\long\def\reals{\mathbb{R}}%

\global\long\def\complex{\mathbb{C}}%

\global\long\def\rationals{\mathbb{Q}}%

\global\long\def\integers{\mathbb{Z}}%


\global\long\def\abs#1{\ensuremath{\left|#1\right|}}%

\global\long\def\smlabs#1{\ensuremath{\lvert#1\rvert}}%
 
\global\long\def\bigabs#1{\ensuremath{\bigl|#1\bigr|}}%
 
\global\long\def\Bigabs#1{\ensuremath{\Bigl|#1\Bigr|}}%
 
\global\long\def\biggabs#1{\ensuremath{\biggl|#1\biggr|}}%

\global\long\def\norm#1{\ensuremath{\left\Vert #1\right\Vert }}%

\global\long\def\smlnorm#1{\ensuremath{\lVert#1\rVert}}%
 
\global\long\def\bignorm#1{\ensuremath{\bigl\|#1\bigr\|}}%
 
\global\long\def\Bignorm#1{\ensuremath{\Bigl\|#1\Bigr\|}}%
 
\global\long\def\biggnorm#1{\ensuremath{\biggl\|#1\biggr\|}}%

\global\long\def\floor#1{\left\lfloor #1\right\rfloor }%
\global\long\def\smlfloor#1{\lfloor#1\rfloor}%

\global\long\def\ceil#1{\left\lceil #1\right\rceil }%
\global\long\def\smlceil#1{\lceil#1\rceil}%


\global\long\def\Union{\ensuremath{\bigcup}}%

\global\long\def\Intsect{\ensuremath{\bigcap}}%

\global\long\def\union{\ensuremath{\cup}}%

\global\long\def\intsect{\ensuremath{\cap}}%

\global\long\def\pset{\ensuremath{\mathcal{P}}}%

\global\long\def\clsr#1{\ensuremath{\overline{#1}}}%

\global\long\def\symd{\ensuremath{\Delta}}%

\global\long\def\intr{\operatorname{int}}%

\global\long\def\cprod{\otimes}%

\global\long\def\Cprod{\bigotimes}%


\global\long\def\smlinprd#1#2{\ensuremath{\langle#1,#2\rangle}}%

\global\long\def\inprd#1#2{\ensuremath{\left\langle #1,#2\right\rangle }}%

\global\long\def\orthog{\ensuremath{\perp}}%

\global\long\def\dirsum{\ensuremath{\oplus}}%


\global\long\def\spn{\operatorname{sp}}%

\global\long\def\rank{\operatorname{rk}}%

\global\long\def\proj{\operatorname{proj}}%

\global\long\def\tr{\operatorname{tr}}%

\global\long\def\vek{\operatorname{vec}}%

\global\long\def\diag{\operatorname{diag}}%

\global\long\def\col{\operatorname{col}}%


\global\long\def\smpl{\ensuremath{\Omega}}%

\global\long\def\elsmp{\ensuremath{\omega}}%

\global\long\def\sigf#1{\mathcal{#1}}%

\global\long\def\sigfield{\ensuremath{\mathcal{F}}}%
\global\long\def\sigfieldg{\ensuremath{\mathcal{G}}}%

\global\long\def\flt#1{\mathcal{#1}}%

\global\long\def\filt{\mathcal{F}}%
\global\long\def\filtg{\mathcal{G}}%

\global\long\def\Borel{\ensuremath{\mathcal{B}}}%

\global\long\def\cyl{\ensuremath{\mathcal{C}}}%

\global\long\def\nulls{\ensuremath{\mathcal{N}}}%

\global\long\def\gauss{\mathfrak{g}}%

\global\long\def\leb{\mathfrak{m}}%


\global\long\def\prob{\ensuremath{\mathbb{P}}}%

\global\long\def\Prob{\ensuremath{\mathbb{P}}}%

\global\long\def\Probs{\mathcal{P}}%

\global\long\def\PROBS{\mathcal{M}}%

\global\long\def\expect{\ensuremath{\mathbb{E}}}%

\global\long\def\Expect{\ensuremath{\mathbb{E}}}%

\global\long\def\probspc{\ensuremath{(\smpl,\filt,\Prob)}}%


\global\long\def\iid{\ensuremath{\textnormal{i.i.d.}}}%

\global\long\def\as{\ensuremath{\textnormal{a.s.}}}%

\global\long\def\asp{\ensuremath{\textnormal{a.s.p.}}}%

\global\long\def\io{\ensuremath{\ensuremath{\textnormal{i.o.}}}}%

\newcommand\independent{\protect\mathpalette{\protect\independenT}{\perp}}
\def\independenT#1#2{\mathrel{\rlap{$#1#2$}\mkern2mu{#1#2}}}

\global\long\def\indep{\independent}%

\global\long\def\distrib{\ensuremath{\sim}}%

\global\long\def\distiid{\ensuremath{\sim_{\iid}}}%

\global\long\def\asydist{\ensuremath{\overset{a}{\distrib}}}%

\global\long\def\inprob{\ensuremath{\overset{p}{\goesto}}}%

\global\long\def\inprobu#1{\ensuremath{\overset{#1}{\goesto}}}%

\global\long\def\inas{\ensuremath{\overset{\as}{\goesto}}}%

\global\long\def\eqas{=_{\as}}%

\global\long\def\inLp#1{\ensuremath{\overset{\Ellp{#1}}{\goesto}}}%

\global\long\def\indist{\ensuremath{\overset{d}{\goesto}}}%

\global\long\def\eqdist{=_{d}}%

\global\long\def\wkc{\ensuremath{\rightsquigarrow}}%

\global\long\def\wkcu#1{\overset{#1}{\ensuremath{\rightsquigarrow}}}%

\global\long\def\plim{\operatorname*{plim}}%


\global\long\def\var{\operatorname{var}}%

\global\long\def\lrvar{\operatorname{lrvar}}%

\global\long\def\cov{\operatorname{cov}}%

\global\long\def\corr{\operatorname{corr}}%

\global\long\def\bias{\operatorname{bias}}%

\global\long\def\MSE{\operatorname{MSE}}%

\global\long\def\med{\operatorname{med}}%

\global\long\def\avar{\operatorname{avar}}%

\global\long\def\se{\operatorname{se}}%

\global\long\def\sd{\operatorname{sd}}%


\global\long\def\nullhyp{H_{0}}%

\global\long\def\althyp{H_{1}}%

\global\long\def\ci{\mathcal{C}}%


\global\long\def\simple{\mathcal{R}}%

\global\long\def\sring{\mathcal{A}}%

\global\long\def\sproc{\mathcal{H}}%

\global\long\def\Wiener{\ensuremath{\mathbb{W}}}%

\global\long\def\sint{\bullet}%

\global\long\def\cv#1{\left\langle #1\right\rangle }%

\global\long\def\smlcv#1{\langle#1\rangle}%

\global\long\def\qv#1{\left[#1\right]}%

\global\long\def\smlqv#1{[#1]}%


\global\long\def\trans{\mathsf{T}}%

\global\long\def\indic{\ensuremath{\mathbf{1}}}%

\global\long\def\Lagr{\mathcal{L}}%

\global\long\def\grad{\nabla}%

\global\long\def\pmin{\ensuremath{\wedge}}%
\global\long\def\Pmin{\ensuremath{\bigwedge}}%

\global\long\def\pmax{\ensuremath{\vee}}%
\global\long\def\Pmax{\ensuremath{\bigvee}}%

\global\long\def\sgn{\operatorname{sgn}}%

\global\long\def\argmin{\operatorname*{argmin}}%

\global\long\def\argmax{\operatorname*{argmax}}%

\global\long\def\Rp{\operatorname{Re}}%

\global\long\def\Ip{\operatorname{Im}}%

\global\long\def\deriv{\ensuremath{\mathrm{d}}}%

\global\long\def\diffnspc{\ensuremath{\deriv}}%

\global\long\def\diff{\ensuremath{\,\deriv}}%

\global\long\def\i{\ensuremath{\mathrm{i}}}%

\global\long\def\e{\mathrm{e}}%

\global\long\def\sep{,\ }%

\global\long\def\defeq{\coloneqq}%

\global\long\def\eqdef{\eqqcolon}%

\global\long\def\vec#1{\boldsymbol{#1}}%

\global\long\def\jsr{{\scriptstyle \mathrm{JSR}}}%

\author{Anna Bykhovskaya}
\address[Anna Bykhovskaya]{Duke University}
\email{anna.bykhovskaya@duke.edu}

\author{James A.\ Duffy}
\address[James A.\ Duffy]{University of Oxford}
\email{james.duffy@economics.ox.ac.uk}
\title{The Local to Unity Dynamic Tobit Model}
\begin{abstract}
This paper considers highly persistent time
series that are subject to nonlinearities in the form of censoring or
an occasionally binding constraint, such as are regularly encountered
in macroeconomics.  A tractable candidate model for such series is the dynamic Tobit
with a root local to unity. We show that this model generates a process that converges weakly to a non-standard
limiting process, that is constrained (regulated) to be positive.
Surprisingly, despite the presence of censoring,
the OLS estimators of the model parameters are consistent.
We show that this allows OLS-based inferences to be drawn on the overall persistence of the process (as measured by
the sum of the autoregressive coefficients), and for the null of a
unit root to be tested in the presence of censoring. Our simulations
illustrate that the conventional ADF test substantially over-rejects
when the data is generated by a dynamic Tobit with a unit root,
whereas our proposed test is correctly sized.
We provide an application of our methods to testing for a unit root in the Swiss franc / euro exchange rate,
during a period when this was subject to an occasionally binding lower bound.

\medskip{}
\textbf{Keywords:} non-negative time series, dynamic Tobit, local
unit root, unit root test.
\end{abstract}

\thanks{We thank T.\ Bollerslev, G.\ Cavaliere, V.\ Kleptsyn, S.\ Mavroeidis, and seminar participants at Cambridge, Cyprus, Duke, NES$30$ Conference, St~Andrews, and Oxford for their helpful comments and advice on earlier drafts of this work.}

\vspace*{-0.9cm}

\maketitle

\vspace*{-0.45cm}

\section{Introduction}

\label{sec:intro}

Since the 1950s nonlinear models have played an increasingly prominent
role in the analysis and prediction of time series data. In many cases,
as was noted in early work by \citet{Moran_lynx}, linear models are
unable to adequately match the features of observed time series. The
efforts to develop models that enjoy the flexibility afforded by nonlinearities,
while retaining the tractability of linear models, have subsequently
engendered an enormous literature (see e.g.\ \citealp{Fan_Yao_nonlinear_book};
\citealp{Gao07}; \citealp{Chan09}; and \citealp{TTG10}).

An important instance of nonlinearity arises when data is bounded
by, truncated at, or censored below some threshold, since such phenomena
cannot be adequately captured -- even approximately -- by a linear
model. Many observed series are bounded below by construction, and
may spend lengthy periods at or near their lower boundary, such as
unemployment rates, prices, gross sectoral trade flows, and nominal interest
rates. The non-negativity of interest rates, and the resulting constraints
that this may impose on the efficacy of monetary policy, has received
particular attention in recent years, as central bank policy rates
have remained at or near the zero lower bound for a significant portion
of the past two decades, across many economies (see e.g.\ \citealp{Mav21Ecta},
and the works cited therein).

A tractable model for such series, which generates both their characteristic
serial dependence and censoring, is the dynamic Tobit model. In its
static formulation, the model originates with \citet{Tobin_tobit}.
In its dynamic formulation, the model typically comes in one of two
varieties, which we refer to as the \emph{latent} and \emph{censored}
models. In the latent dynamic Tobit, an unobserved process $\{y_{t}^{\ast}\}$
follows a linear autoregression, with $y_{t}=\max\{y_{t}^{\ast},0\}$
being observed; whereas in the censored dynamic Tobit, $\{y_{t}\}$
is modelled as the positive part of a linear function of its own lags,
and an additive error (see e.g.\ \citealt[p.~186]{Maddala83}, or
\citealt[p.~419]{wei1999}). In both models the right hand side may
be augmented with other explanatory variables. Relative to the latent
model, the censored model has the advantage of being Markovian, which
greatly facilitates its use in forecasting. It has also been successfully
applied to a range of censored series, in both purely time series
and panel data settings, including: the open market operations of
the Federal Reserve (\citealp{Demiralp_Jorda2002}; \citealp{stat_paper});
household commodity purchases (\citealp{DSK12AE}); loan charge-off
rates (\citealp{LMS19}); credit default and overdue loan repayments
(\citealp{BMMV21JBF}); and sectoral bilateral trade flows \citep{bykh_JBES}.
Recently, \citet{Mav21Ecta} proposed the censored and kinked structural
VAR model to describe the operation of monetary policy during periods
when the zero lower bound may occasionally bind on the policy rate.
If only the actual interest rate (rather than some `shadow rate')
affects agents' decision making, as assumed in closely related work
by \citet{AMSV21JoE}, then the univariate counterpart of this model
is exactly the censored dynamic Tobit.

The present work is concerned with the censored, rather than the latent,
dynamic Tobit model. As discussed by \citet[][p.~229]{stat_paper}, the censored model
is arguably more appropriate in settings where $y_t = 0$ results
not from limits on the observability of some underlying $y_t^{\ast}$,
but from an economic constraint on the values taken by $y_t$.
For example, \citet[Supplementary Material, Appendix A]{bykh_JBES}
justifies the application of the dynamic Tobit with reference
to a game-theoretic model of network formation, in which the zeros
correspond to corner solutions of a constrained optimisation problem,
i.e.\ where zeros are systematically observed because of a non-negativity
constraint on agents' choices. In our illustrative empirical application,
to an exchange rate that is subject to a floor engineered by a central bank
(Section~\ref{sec:empirical}), the most recent values of the exchange rate
are taken as sufficient to describe the market equilibrium, and thus
the conditional distribution of future rates. Our focus on the censored
model is also motivated by relatively greater need for the
development of relevant econometric theory in this area. In the latent model, the dynamics are simply
those of the latent autoregression, and so are readily understood by standard methods;
whereas in the censored model, the censoring affects
the dynamics of $\{y_{t}\}$ in a non-trivial manner, making the analysis
rather more challenging. Indeed, establishing the stationarity or weak dependence of the censored
dynamic Tobit is far from trivial, as can be seen from
\citet{HK10EcLett}, \citet{stat_paper}, \citet{MdJ18EcLett}, and
\citet{bykh_JBES}. Henceforth, all references to the `dynamic Tobit'
are to the censored version of the model.\footnote{\citet{stat_paper} alternatively refer to this model as a `dynamic censored regression', but the term `dynamic Tobit' appears more commonly in the literature (see e.g.\ \citealp{HK10EcLett,MdJ18EcLett,bykh_JBES}).}

Motivated in part by recent work on modelling nominal interest rates
near the zero lower bound, our concern is with the application of
this model to series that are highly persistent, so that above the
censoring point they exhibit the random wandering that is characteristic
of integrated processes. The appropriate configuration of the dynamic
Tobit model for such series, in which the autoregressive polynomial
has a root local to unity, has not been considered in the literature
to date -- apart from the special case of a first-order model with
an exact unit root, as in \citet{cav_bounded} and \citet{bykh_JBES}.
Our results are thus entirely new to the literature.

Our principal technical contribution, within this setting, is to derive
the limiting distributions of both the standardised regressor process,
and the ordinary least squares (OLS) estimates of the parameters of
the dynamic Tobit, when that model has an autoregressive root local
to unity. The reader may find our focus on OLS surprising, as this method would
usually provide inconsistent estimates in the presence of censoring. However,
it turns out that in our setting consistency (for all model parameters)
is restored, and we obtain a usable limit theory for the estimated sum of the autoregressive coefficients,
which conventionally provides a measure of the overall persistence
of a process (cf.\ \citealp{AC94JBES}; \citealp{Mik07Ecta}).
While one may contemplate alternatively using maximum likelihood (ML) or
least absolute deviations (LAD) to estimate the model, OLS enjoys the advantages
of maintaining only weak distributional assumptions on the innovations (unlike ML),
and avoiding the numerical minimisation of a non-convex criterion function (unlike LAD).

Our asymptotics provide the basis for practical unit root tests for highly
persistent, censored time series. Motivated by our finding that OLS is consistent,
we consider a test based on the (constant only) augmented Dickey--Fuller (ADF) $t$ statistic,
but which employs critical values modified to reflect the censoring present
in the data generating process. We show, via Monte Carlo simulations,
that as our critical values
are larger than the conventional ADF critical values, their use eliminates the significant over-rejection
that may result from the naive application of the ADF test to censored
data. (This tendency to over-reject the null of a unit root appears typical
of models that incorporate unit roots and nonlinearities,
having been also found by e.g.\ \citealp{HT97EcLett}; \citealp{KLN02JoE};
and \citealp{WDJ13SASA}.) Strikingly, the distribution of $t$ statistic, under censoring, is stochastically dominated
by that obtained from the linear autoregressive model.\footnote{This holds only with respect to the \emph{distributions}:
it is not true that if one simulates a linear and a Tobit model with the same underlying innovations, then the former $t$ statistic will always be larger than the latter.}

Our work may be construed, more broadly, as extending the analysis
of highly persistent time series, and the associated machinery of
unit root testing, from a linear setting to a nonlinear setting appropriate
to time series that are subject to a lower bound. In doing so, we
complement the seminal work of \citet{cavaliere2005}, which similarly
sought to extend this machinery to the setting of bounded time series.
Our contribution is to effect this extension within a class of
nonlinear autoregressive models that have been widely applied to
censored time series (as evinced by the works cited
above), and which fall outside his framework.

On a technical level, the most closely related works to our own are
those of \citet{cav_bounded,cavaliere2005} and \citet{CX14JoE},
who develop the asymptotics of what they term `limited autoregressive
processes' with a near-unit root, which are (one- or two-sided) non-Markovian censored
processes constructed by the addition of regulators to a latent linear
autoregression. While their (one-sided) model has a superficial resemblance
to the dynamic Tobit, there are important, but subtle differences
between the two (see Section~\ref{subsec:limited} for a
discussion). Perhaps the most striking similarity is that both models,
in the case of an exact unit root, give rise to processes
that converge weakly to regulated Brownian motion; but when
roots are merely local to unity, the limiting processes associated
with these two models are distinct (see the discussion following Theorem~\ref{thm:LUR_y}).

Convergence to regulated Brownian
motions has also been obtained previously in the setting of \emph{first-order} threshold
autoregressive models with an (exact) unit root regime and a stationary regime,
as considered by \citet{LLS11Bern} and \citet{GLY13JoE}. However, allowing for
both near-unit roots and higher-order autoregressive terms introduces
technical challenges that require us to take a markedly different
approach from those employed in these earlier works. With respect to higher-order models,
a major difficulty relates to the treatment
of the differences $\{\Delta y_t\}$. In a linear autoregressive model with a single unit
root, and all other roots outside the unit circle, these would follow a stationary autoregression; but in our setting
they instead follow a regime-switching autoregression, where the regime depends on
the (lagged) \emph{level} of $y_t$, and so are inherently non-stationary.
Accoridingly, standard arguments for controlling the magnitude of $\Delta y_t$, and deriving
the limits of functionals thereof, are unavailing. We provide a striking example
in which $\Delta y_t$ is explosive even though all but one of the autoregressive roots
(the root at unity) lie outside the unit circle,
due to the interactions between the two autoregressive regimes (see Appendix \ref{sec:appendix_JSR}).
To preclude such cases, we develop a condition relating to the joint spectral radius of the autoregressive
representation for $\Delta y_t$, which is sufficient to control the magnitude of
$\Delta y_t$ and plays an essential role in our arguments.
(This concept has been previously employed in the context of
\emph{stationary} autoregressive models, see e.g.\ \citet{Lieb05JTSA};
\citet{Saik08ET}.)

The remainder of this paper is organised as follows. Section~\ref{sec:setting}
discusses the model and our assumptions. Asymptotic results and corresponding
tests are derived in Section~\ref{sec:asymptotics}, while supporting
Monte Carlo simulations are shown in Section~\ref{sec:montecarlo}.
Section~\ref{sec:empirical} applies our framework to the exchange
rate between the Swiss franc and the euro during a period when this rate
was subject to a lower bound. Finally, Section \ref{sec:concl} concludes. All proofs appear
in the appendices.
\begin{notation*}
$C$, $C^{\prime}$, $C^{\prime\prime}$, etc.\ denote generic constants
that may take different values in different parts of this paper. All
limits are taken as $T\goesto\infty$ unless otherwise specified.
$\inprob$ and $\indist$ respectively denote convergence in probability
and distribution (weak convergence). We write `$X_{T}(r)\indist X(r)$
on $D[0,1]$' to denote that $\{X_{T}\}$ converges weakly to $X$,
where these are considered as random elements of $D[0,1]$, the space
of cadlag functions on $[0,1]$, equipped with the uniform topology.
For $p\geq1$ and $X$ a random variable, $\smlnorm X_{p}\defeq(\expect\smlabs X^{p})^{1/p}$.
\end{notation*}

\section{The dynamic Tobit model with a near-unit root}

\label{sec:setting}

\subsection{Model and assumptions}
Consider a time series $\{y_{t}\}$ generated by the dynamic Tobit
model of order $k\geq1$, written in augmented Dickey--Fuller (ADF)
form,\footnote{The autoregressive form is $y_{t}=[\alpha+\sum_{i=1}^{k}\beta_{i}y_{t-i}+u_{t}]_{+}$,
where $\beta_1 = \beta + \phi_1$, $\beta_k = -\phi_{k-1}$, and $\beta_i
= \phi_i - \phi_{i-1}$ for $i \in \{2,\ldots,k-1\}$. In particular $\beta=\sum_{i=1}^{k}\beta_{i}$ corresponds
to the sum of the autoregressive coefficients.}
\begin{equation}
y_{t}=\left[\alpha+\beta y_{t-1}+\sum_{i=1}^{k-1}\phi_{i}\Delta y_{t-i}+u_{t}\right]_{+},\quad t=1,\ldots,T,\label{eq:tobitark}
\end{equation}
where $\Delta y_{t}:=y_{t}-y_{t-1}$, and $[x]_{+}\defeq\max\{x,0\}$
denotes the positive part of $x\in\reals$. Let
\begin{equation}\label{eq:lag_operators}
B(z):=1-\beta z-(1-z)\sum_{i=1}^{k-1}\phi_{i}z^{i}=(1-\beta)z+(1-z)\phi(z),
\end{equation}
where $\phi(z):=1-\sum_{i=1}^{k-1}\phi_{i}z^{i}$. We impose the following on the data generating process \eqref{eq:tobitark}.
\begin{assumption}
\label{ass:INIT} $\{y_{t}\}$ is initialised by (possibly) random
variables $\{y_{-k+1},\ldots,y_{0}\}$. Moreover, $T^{-1/2}y_{0}\inprob b_{0}$
for some $b_{0}\geq0$.
\end{assumption}
\begin{assumption}
\label{ass:DGP} $\{y_{t}\}$ is generated according to \eqref{eq:tobitark},
where:
\begin{enumerate}[label=\ass{\arabic*.}, ref=\ass{.\arabic*}]
\item \label{enu:DGP:ut}$\{u_{t}\}_{t\in\integers}$ is independently
and identically distributed (i.i.d.)\ with $\expect u_{t}=0$ and
$\expect u_{t}^{2}=\sigma^{2}$.
\item \label{enu:DGP:LU}$\alpha=\alpha_{T}\defeq T^{-1/2}a$ and $\beta=\beta_{T}=\exp(c/T)$
for some $a,c\in\reals$.
\end{enumerate}
\end{assumption}
\begin{assumption}
\label{ass:MOM} There exist $\delta_{u}>0$ and $C<\infty$ such
that:
\begin{enumerate}[label=\ass{\arabic*.}, ref=\ass{.\arabic*}]
\item $\expect\smlabs{u_{t}}^{2+\delta_{u}}<C$.
\item $\expect\smlabs{T^{-1/2}y_{0}}^{2+\delta_{u}}<C$, and $\expect\smlabs{\Delta y_{i}}^{2+\delta_{u}}<C$
for $i\in\{-k+2,\ldots,0\}$.
\end{enumerate}
\end{assumption}
Figure~\ref{y_pic} displays a typical sample path for the dynamic Tobit \eqref{eq:tobitark}, under the preceding assumptions.
\begin{figure}[h]
{\scalebox{0.95}{\includegraphics{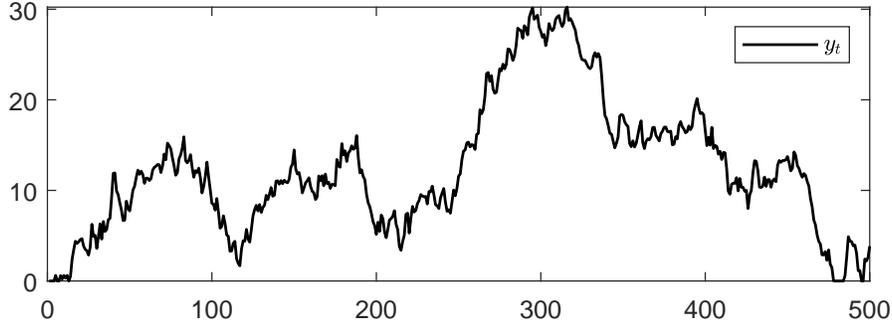}}}
\captionsetup{width=0.9\linewidth}
\caption{$y_{t}=\left[\tfrac{1}{T^{1/2}}+\left(1-\tfrac{5}{T}\right)y_{t-1}+u_{t}\right]_{+}$,
$y_{0}=0$, ${u_{t}\thicksim\text{i.i.d.}~\mathcal{N}(0,1)}$, $T=500$.}
\label{y_pic}
\end{figure}

\pagebreak[4]

The main consequences of our assumptions may be summarised as follows.
\begin{enumerate}
\item By the functional central limit theorem, \ref{ass:DGP}\ref{enu:DGP:ut}
implies $T^{-1/2}\sum_{t=1}^{\smlfloor{rT}}u_{t}\indist\sigma W(r)$
on $D[0,1]$, where $W(\cdot)$ is a standard Brownian motion. This convergence
alone is sufficient to determine the asymptotics of $T^{-1/2}y_{\smlfloor{rT}}$,
and of the OLS estimators, when $k=1$ (Theorems \ref{thm:LUR_y}
and \ref{thm:OLS_LUR} below), but extending these results to $k\geq2$
necessitates the slightly stronger conditions on $\{u_{t}\}$ provided
by \ref{ass:MOM}.
\item In the absence of censoring, \ref{ass:DGP}\ref{enu:DGP:LU} would
entail that $y_{t}$ has an autoregressive root within a $O(T^{-1})$
neighbourhood of real unity. Just as in that case, we shall show that
in the present setting $T^{-1/2}y_{\smlfloor{rT}}$ converges weakly
to a continuous process, albeit one that differs importantly from
the diffusion process limit familiar from the uncensored case.
\item We require $\alpha=O(T^{-1/2})$ in \ref{ass:DGP}\ref{enu:DGP:LU},
to ensure that the drift in $\{y_{t}\}$ is of no larger order than
the stochastic trend component. If this assumption were relaxed, so
that e.g.\ $\alpha$ were now a non-zero constant, the large-sample
behaviour of $\{y_{t}\}$ and the asymptotics of the OLS estimators
would be quite different from those developed here. A fixed positive
$\alpha$ would generate an increasing linear trend, driving $y_{t}$
ever further away from origin and making the censoring ultimately
irrelevant; whereas a fixed negative $\alpha$ can lead to $\{y_{t}\}$
being stationary (see e.g.\ \citealp[Theorem 3]{bykh_JBES}).
\item The specific parametrisations in \ref{ass:DGP}\ref{enu:DGP:LU} are
chosen merely for convenience: all of our results also hold when $\alpha_{T}$
and $\beta_{T}$ more generally satisfy $T^{1/2}\alpha_{T}\goesto a$
and $T(\beta_{T}-1)\goesto c$. For ease of notation, we shall routinely
suppress the $T$ subscripts on $\alpha_{T}$ and $\beta_{T}$ throughout
the following.
\item Assumptions \ref{ass:INIT} and \ref{ass:MOM} imply that $T^{-1/2}y_{i}\inprob b_{0}$
for all $i\in\{-k+1,\ldots,0\}$.
\end{enumerate}

\subsection{Non-zero lower bound}\label{sec:nonzero_bound}
Our machinery extends straightforwardly to the case where $y_{t}$
is censored at some $\mathbf{L}\neq0$. Suppose \eqref{eq:tobitark}
is modified to
\begin{equation}
y_{t}=\max\left\{ \mathbf{L}\sep\alpha+\beta y_{t-1}+\sum_{i=1}^{k-1}\phi_{i}\Delta y_{t-i}+u_{t}\right\} ,\label{eq:Lbnd}
\end{equation}
and take $\mathbf{L}=T^{1/2}\ell$ for some $\ell\in\reals$, to allow
the censoring point to be of the same order of magnitude as $\{y_{t}\}$.\footnote{
We emphasise that this dependence of $\mathbf{L}$ on $T$ should not be
interpreted literally as specifying a model in which the censor point
is a function of the sample size. Rather, the scaling is a mathematical
device that allows us to obtain an improved asymptotic approximation to the
finite sample distribution of $T^{-1/2}y_{\smlfloor{rT}}$, and hence of
the test statistics that depend upon it. (See e.g.\ Assumption~(A4) in
\citet{cavaliere2005}, where a similar device is used for this purpose.) If
$\mathbf{L}$ is in fact `small' relative to a given sample size $T$, then
this will be accommodated within our framework by $\ell$ being close to zero.}
Defining $\tilde{y}_{t}\defeq y_{t}-\mathbf{L}$ and subtracting $\mathbf{L}$
from both sides of \eqref{eq:Lbnd}, it may be verified that
\begin{align*}
\tilde{y}_{t} & =\left[\tilde{\alpha}+\beta\tilde{y}_{t-1}+\sum_{i=1}^{k-1}\phi_{i}\Delta\tilde{y}_{t-i}+u_{t}\right]_{+}
\end{align*}
where $\tilde{\alpha}\defeq\alpha+(\beta-1)\mathbf{L}$. Thus, $\{\tilde{y}_{t}\}$
follows a dynamic Tobit with censoring at zero, with drift
\[
T^{1/2}\tilde{\alpha}=T^{1/2}[\alpha+(\beta-1)T^{1/2}\ell]\goesto a+c\ell\eqdef\tilde{a}
\]
and initialisation
\[
\tilde{b}_{0}=\plim_{T\goesto\infty}T^{-1/2}\tilde{y}_{0}=\plim_{T\goesto\infty}T^{-1/2}(y_{0}-\mathbf{L})=b_{0}-\ell.
\]
All our results hold in the setting of \eqref{eq:Lbnd}, with appropriate modifications. For simplicity, we work with $\mathbf{L}=0$ throughout the rest of the paper, except where otherwise indicated.

\subsection{Alternative representation}

It will be occasionally useful to rewrite \eqref{eq:tobitark} in
a form that helps to clarify the connections between the dynamic Tobit
and the linear autoregressive model. We can do this by defining
\begin{equation}
y_{t}^{-}\defeq\left[\alpha+\beta y_{t-1}+\sum_{i=1}^{k-1}\phi_{i}\Delta y_{t-i}+u_{t}\right]_{-},\label{eq:ytminus}
\end{equation}
where $[x]_{-}\defeq\min\{x,0\}$. That is, when $y_{t}=0$, $y_{t}^{-}$
records the value that $y_{t}$ would have taken had it not been censored
at zero.

Since $[x]_{+}=x-[x]_{-}$, we may then rewrite \eqref{eq:tobitark}
as
\begin{equation}
y_{t}=\alpha+\beta y_{t-1}+\sum_{i=1}^{k-1}\phi_{i}\Delta y_{t-i}+u_{t}-y_{t}^{-},\label{eq:arkwithytminus}
\end{equation}
or equivalently, letting $L$ denote the lag operator, as
\begin{equation}
B(L)y_{t}=\alpha+u_{t}-y_{t}^{-}.\label{eq:ARklagpoly}
\end{equation}
Thus, if we view $y_{t}^{-}$ as an additional noise term, \eqref{eq:ARklagpoly}
takes the form of a linear autoregression. The main challenge is that
$y_{t}^{-}$ is itself a complicated nonlinear object, whose presence
fundamentally alters the dynamics of $\{y_t\}$, even in the long run.

\subsection{Connections with limited autoregressive processes}

\label{subsec:limited}

The representation \eqref{eq:ARklagpoly} allows us to draw out the connections
between our model and the limited autoregressive processes
developed by \citet{cav_bounded,cavaliere2005} and \citet{CX14JoE}.
To put their model -- for the special case of a process constrained
to lie in $[0,\infty)$ -- in a form comparable to ours, consider
a latent process $\{x_{t}^{\ast}\}$,
\begin{equation}
x_{t}^{\ast}=\rho_{T}x_{t-1}^{\ast}+\varepsilon_{t},\quad\rho_{T}=1+c/T,\label{eq:cav1}
\end{equation}
where $\{\varepsilon_{t}\}$ is stationary. Define an observed process
$\{x_{t}\}$, whose increments are related to those of $\{x_{t}^{\ast}\}$
via
\begin{equation}
\Delta x_{t}=\Delta x_{t}^{\ast}+\blw{\xi}_{t}\label{eq:cav2}
\end{equation}
where $\blw{\xi}_{t}>0$ if and only if $x_{t-1}+\Delta x_{t}^{\ast}<0$,
so as to ensure that $x_{t}\geq0$
for all $t$. In particular, if we set
\begin{equation}
\blw{\xi}_{t}=-x_{t}^{-}\defeq-[x_{t-1}+\Delta x_{t}^{\ast}]_{-}\label{eq:cav3}
\end{equation}
then $\{x_{t}\}$ will be censored at zero. When $c=0$, by combining
\eqref{eq:cav1}--\eqref{eq:cav3} we obtain
\begin{equation}
x_{t}=x_{t-1}+\varepsilon_{t}-x_{t}^{-} = [x_{t-1} + \varepsilon_{t}]_{+}\label{eq:cavcensdyn}
\end{equation}
as a valid representation of a limited autoregressive process censored
at zero.

Both \eqref{eq:ARklagpoly} and \eqref{eq:cavcensdyn} describe
censored processes, but have some subtle, and yet important
differences. Firstly, while $\{y_t\}$ in \eqref{eq:ARklagpoly} is
a Markov process (with state vector $(y_{t},\ldots,y_{t-k+1})$),
$\{x_t\}$ in \eqref{eq:cavcensdyn} will be Markov
only if $\{\varepsilon_t\}$ is i.i.d. The former is thus more
suited to forecasting in the presence of higher-order dynamics.

Secondly, at a technical level, the differences between the models
can be clearly illustrated by supposing that $\alpha=0$ and $\beta=1$.
Then $B(L)=(1-L)\phi(L)$ in (\ref{eq:ARklagpoly}), which simplifies
to
\begin{equation}
y_{t}=y_{t-1}+\phi(L)^{-1}(u_{t}-y_{t}^{-}).\label{eq:dyntob-comp}
\end{equation}
In the dynamic Tobit, higher-order dynamics are captured by the (stationary)
autoregressive polynomial $\phi(L)$, whereas in the limited autoregressive
model these enter via the weak dependence in $\{\varepsilon_{t}\}$.
To facilitate a comparison of the two models, suppose that $\{\varepsilon_{t}\}$
follows the autoregression $\phi(L)\varepsilon_{t}=u_{t}$. Then (\ref{eq:cavcensdyn})
becomes
\begin{equation}
x_{t}=x_{t-1}+\phi(L)^{-1}u_{t}-x_{t}^{-}.\label{eq:la-comp}
\end{equation}
We see immediately that if both models have only first-order dynamics
($k=1$), so that $\phi(L)=1$, then they exactly coincide (cf.\ Cavaliere,
2005, Remark~2.3). However, this no longer holds in the presence
of higher-order dynamics ($k\geq2$). Suppose e.g.\ that $\phi(L)=1-\phi_1 L$
for some $\smlabs{\phi_1}<1$. Then (\ref{eq:dyntob-comp}) becomes
\begin{equation}
y_{t}=y_{t-1}+\sum_{s=0}^{\infty}\phi^{s}_1u_{t-s}-\sum_{s=0}^{\infty}\phi^{s}_1y_{t-s}^{-}\label{eq:dyntob-ar2}
\end{equation}
whereas (\ref{eq:la-comp}) yields
\begin{equation}
x_{t}=x_{t-1}+\sum_{s=0}^{\infty}\phi^{s}_1u_{t-s}-x_{t}^{-}.\label{eq:la-ar2}
\end{equation}
Comparing (\ref{eq:dyntob-ar2}) with (\ref{eq:la-ar2}), we see that
the censoring affects the dynamics of $\{y_{t}\}$ and $\{x_{t}\}$
in different ways. Lagged values of $y_{t}^{-}$ have a direct effect
on future $y_{t}$ (via $\sum_{s=0}^{\infty}\phi^{s}_1 y_{t-s}^{-}$),
whereas lagged $x_{t}^{-}$ have no such effect on $x_{t}$.

Thirdly, the differences between the two models also manifests itself through
each giving rise to distinct classes of limiting processes.
Even when $k=1$, these coincide only in the special case of an exact
unit root ($c=0$): see the discussion following Theorem~\ref{thm:LUR_y} below.

\section{Asymptotic results}

\label{sec:asymptotics}

In this section we derive the weak limits of the standardised process $T^{-1/2}y_{\smlfloor{rT}}$
and the ordinary least squares (OLS) estimators of the parameters
of the dynamic Tobit. The latter provides the basis for a unit root
test for non-negative time series, a test that is both
is straightforward to compute, and which does not require any assumption to be made on the
distribution of the innovations, beyond the existence of sufficient moments.
We first provide a separate treatment of the first-order model
($k=1$), before progressing to the model with higher-order
dynamics ($k\geq 1$). This facilitates a simplified exposition of the former case,
which avoids the additional assumptions (\assref{MOM} and \assref{JSR}) and technical concepts --
notably the joint spectral radius -- that are required when $k\geq 2$.

\subsection{Limiting distribution of the regressor process}

\label{sec:yt_asymptotics}

Let $\theta\defeq(a,b_{0},c)$, define the process
\begin{equation}
K_{\theta}(r):=b_{0}+a\int_{0}^{r}\e^{-cs}\diff s+\sigma\int_{0}^{r}e^{-cs}\diff W(s),\label{eq:Kac}
\end{equation}
and denote its `regulated' counterpart by
\begin{equation}
J_{\theta}(r)\defeq\e^{cr}\left\{ K_{\theta}(r)+\sup_{r^{\prime}\leq r}[-K_{\theta}(r^{\prime})]_{+}\right\} .\label{eq:Jac}
\end{equation}
Here the supremum on the r.h.s.\ regulates $J_{\theta}(r)$ to ensure that it is
always non-negative: if $K_{\theta}(r)$ is negative,
$[-K_{\theta}(r)]_{+}=-K_{\theta}(r)>0$, so that $K_{\theta}(r)+\sup_{r^{\prime}\leq r}[-K_{\theta}(r^{\prime})]_{+}\geq0$.

We first provide a result for the case $k=1$, under which the model
\eqref{eq:tobitark} reduces to
\begin{equation}
y_{t}=[\alpha+\beta y_{t-1}+u_{t}]_{+}.\label{eq:dgp}
\end{equation}

\begin{thm}
\label{thm:LUR_y} Suppose Assumptions \ref{ass:INIT} and \ref{ass:DGP}
hold with $k=1$ in \eqref{eq:tobitark}. Then on $D[0,1]$,
\begin{equation}
T^{-1/2}y_{\smlfloor{rT}}\indist J_{\theta}(r).\label{eq:AR1wkc}
\end{equation}
\end{thm}
The preceding is a new result, which relates to some of the previous literature as follows.
\begin{enumerate}
\item The appearance of the supremum in the weak limit of \eqref{eq:AR1wkc} (see \eqref{eq:Jac} above)
is in line with the solution to the Skorokhod reflection problem (\citealp[p.~239]{Revuz_Yor}).

\item Suppose that $a=b_{0}=0$. Then $\e^{cr}K_{\theta}(r)=S_{c}(r)=\sigma\int_{0}^{r}\e^{c(r-s)}\diff W(s)$,
an Ornstein--Uhlenbeck process with autoregressive parameter $c$
(e.g.\ \citealp{CW87AS}; \citealp{Phil87Btka}), and \eqref{eq:AR1wkc}
specialises to
\[
T^{-1/2}y_{\smlfloor{rT}}\indist S_{c}(r)+\sup_{r^{\prime}\leq r}[-\e^{c(r-r^{\prime})}S_{c}(r^{\prime})]_{+}.
\]
on $D[0,1]$. Whereas, if we take $\varepsilon_{t}=u_{t}$ in \eqref{eq:cav1}, the limited autoregressive
process \eqref{eq:cav1}--\eqref{eq:cav3} satisfies
\[
T^{-1/2}x_{\smlfloor{rT}}\indist S_{c}(r)+\sup_{r^{\prime}\leq r}[-S_{c}(r^{\prime})]_{+}
\]
on $D[0,1]$. Comparing the two preceding limits, we observe a subtle
but crucial difference, due to the presence of the factor $\e^{c(r-r^{\prime})}$,
showing that the asymptotics of the dynamic Tobit and limited autoregressive
models are distinct even when $k=1$.
\item When $a=b_0=c=0$, $J_{\theta}(r)$ coincides with a Brownian motion
regulated from below at zero, which has the same distribution as a
Brownian motion reflected at the origin, $\smlabs{W(\cdot)}$, see e.g.\ \citet[p.~97]{karatzas}.
Another model that generates a process with this asymptotic distribution
(upon rescaling by $T^{-1/2}$) is a first-order threshold autoregression
with `unit root' and `stationary' regimes, as studied by \citet{LLS11Bern}
and \citet{GLY13JoE}. A special case of their model posits
\[
x_{t}=\beta(x_{t-1})x_{t-1}+u_{t},
\]
where $\beta(x)=1$ if $x\geq0$, and $\beta(x)=0$ otherwise. It
follows that $x_{t}=[x_{t-1}]_{+}+u_{t}$, and so $[x_{t}]_{+}=[[x_{t-1}]_{+}+u_{t}]_{+}$,
which corresponds to our setting \eqref{eq:tobitark} with $\alpha=0$,
$\beta=1$, $k=1$, and $y_{t}=[x_{t}]_{+}$. It is thus not surprising
that, in this case, our Theorem~\ref{thm:LUR_y} agrees exactly with the corresponding
Theorem~3.1 of \citet{LLS11Bern}.

\item
The proof of Theorem~\ref{thm:LUR_y} shows that the convergence
\eqref{eq:AR1wkc} relies ultimately on $T^{-1/2}\sum_{t=1}^{\smlfloor{rT}}u_{t}\indist\sigma W(r)$,
as follows by the functional central limit theorem when $\{u_{t}\}$
is i.i.d.\ with mean zero and variance $\sigma^{2}$, as per \assref{DGP}.
If more generally $\{u_{t}\}$ is weakly dependent with long-run variance
$\omega^{2}$, then under appropriate regularity conditions (see e.g.\ Theorem 1.1 in \citealp{PU05AP}) we would
instead have $T^{-1/2}\sum_{t=1}^{\smlfloor{rT}}u_{t}\indist\omega W(r)$,
with the consequence that \eqref{eq:AR1wkc} would continue to hold, albeit
with $\omega$ replacing $\sigma$ in \eqref{eq:Kac}. However, the results in Section \ref{sec:OLS_asymptotics} below cannot be so straightfowardly generalised, as here the presence of weak dependence in $\{u_t\}$ would entail
substantial changes to the limiting distributions (cf.\ the very different limits given in Theorems~2.1 and 2.3 of \citet{ito_convergence}, depending on whether the relevant innovation sequence is serially uncorrelated).
\end{enumerate}

When $k>1$, the other roots of the lag polynomial $B(L)$ affect
the behaviour of $\{y_{t}\}$, and we need a further condition to
ensure that the first differences $\{\Delta y_{t}\}$ are well behaved.
Let
\begin{equation}
F_{\delta}\defeq\begin{bmatrix}\phi_{1}\delta & \phi_{2} & \cdots & \phi_{k-2} & \phi_{k-1}\\
\delta & 0 & \cdots & 0 & 0\\
0 & 1\\
 &  & \ddots\\
 &  &  & 1 & 0
\end{bmatrix}.\label{eq:Fdef}
\end{equation}

Under an appropriate condition on the matrices $\{F_{\delta}\mid\delta\in[0,1]\}$,
we can ensure $\{\Delta y_{t}\}$ is stochastically bounded. To state
that, we need the following (cf.\ \citet{Jungers09}, Defn.\ 1.1):
\begin{defn*}
The \emph{joint spectral radius} (JSR) of a bounded collection $\mathcal{A}$
of square matrices is
\[
\lambda_{\jsr}(\mathcal{A})\defeq\limsup_{n\goesto\infty}\sup_{M\in\mathcal{A}^{n}}\lambda(M)^{1/n}
\]
where $\lambda(M)$ denotes the spectral radius of $M$, and $\mathcal{A}^{n}\defeq\{\prod_{i=1}^{n}A_{i}\mid A_{i}\in\mathcal{A}\}$.
\end{defn*}
Control over the JSR has been previously used to ensure the stationarity
of regime-switching autoregressive models (e.g.~\citet{Lieb05JTSA};
\citet{Saik08ET}), and we shall utilise it in a similar manner here.
\begin{assumption}
\label{ass:JSR}$\lambda_{\jsr}(\{F_{0},F_{1}\})<1$.
\end{assumption}

Approximate upper bounds for the JSR can be computed numerically,
to an arbitrarily high degree of accuracy, via semidefinite programming \citep{PJ08LAA},
making it reasonably straightforward to verify whether this condition is satisfied by given
parameter values. The following result, which is proved in Appendix \ref{sec:appendix_JSR}, provides a sufficient condition for \ref{ass:JSR}, which may be checked even more simply.

\begin{lem}\label{lem:JSR_holds}
If $\sum_{i=1}^{k-1}|\phi_i|<1$, then Assumption \ref{ass:JSR} is satisfied.
\end{lem}

\begin{rem}

(i) To give some intuition for why a condition on $\{F_{\delta}\}$ is
needed here, suppose that $\alpha=0$ and $\beta=1$. Then \eqref{eq:arkwithytminus}
entails
\[
\Delta y_{t}=\sum_{i=1}^{k-1}\phi_{i}\Delta y_{t-i}+u_{t}-y_{t}^{-},
\]
and hence in the absence of censoring $\{\Delta y_{t}\}$
would follow a linear autoregression, for which $F_{1}$ gives the associated
companion form. Assumption~\assref{JSR} implies that the eigenvalues of $F_{1}$
are below $1$ in modulus, and thus that the all roots of $\phi(z)$
lie strictly outside the unit circle (see e.g.\ \citealp[Prop.~1.1]{Hamilton94}).
(Since $\phi(0)=1$, it also follows that $\phi(1)>0$.)
In the presence of censoring, $\{\Delta y_{t}\}$ may be shown to instead (when $\alpha = 0$ and $\beta=1$) follow a time-varying autoregression, in which it evolves jointly with an auxiliary process $\{w_t\}$ as per
\begin{align*}
w_{t} & =\phi_{1}\delta_{t-1}w_{t-1}+\sum_{i=2}^{k-1}\phi_{i}\Delta y_{t-i}+u_{t},\\
\Delta y_{t-1} & =\delta_{t-1}w_{t-1},
\end{align*}
for some (stochastic) sequence $\delta_{t}\subset[0,1]$ (see the proof of Lemma~\ref{lem:qdnegl}
in Appendix~\ref{app:proofgeneral}).
$F_{\delta}$ thus corresponds to the companion form autoregressive
matrix for $(w_{t},\Delta y_{t-1},\ldots,\Delta y_{t-k+2})^{\trans}$
when $\delta_{t-1}=\delta$. Because $\lambda_{\jsr}(\{F_{\delta}\mid\delta\in[0,1]\})=\lambda_{\jsr}(\{F_{0},F_{1}\})$,
\assref{JSR} is sufficient to ensure that this time-varying autoregressive
system is stable, irrespective of the sequence $\{\delta_{t}\}$.

(ii) As illustrated by Example~\ref{ex:stat_not_suff} in Appendix~\ref{sec:appendix_JSR}, merely requiring that $\phi(z)$ have only stationary roots is not sufficient to guarantee the convergence (in distribution) of $T^{-1/2}y_{\smlfloor{rT}}$. Due to the nonlinearity in the model it may be possible to induce explosive trajectories for both $\Delta y_t$ and $y_t$, via a succession of periods in which $y_{t} > 0$ alternates with $y_{t}=0$ (see Figure \ref{fig:JSR_ex2_y}). Thus additional conditions, such as \ref{ass:JSR}, are needed to exclude such behaviour.

On the other hand, as illustrated by Example \ref{ex:JSR_not_nec} in Appendix~\ref{sec:appendix_JSR}, neither is Assumption~\ref{ass:JSR} necessary for the convergence of $T^{-1/2}y_{\smlfloor{rT}}$. Finding a necessary condition thus remains a challenging open question.
\end{rem}


\begin{thm}
\label{thm:ytARk}\label{ytARk} Suppose Assumptions \ref{ass:INIT}--\ref{ass:JSR}
hold. Then
\begin{equation}
T^{-1/2}y_{\smlfloor{rT}}\indist\phi(1)^{-1}J_{\theta_{\phi}}(r)\eqdef Y_{\theta_{\phi}}(r)\label{eq:regwklim}
\end{equation}
on $D[0,1]$, where $\theta_{\phi}\defeq[a,\phi(1)b_{0},\phi(1)^{-1}c]$.
\end{thm}
The principal difference between Theorems \ref{thm:LUR_y} and \ref{thm:ytARk}
is that when $k>1$, the stationary dynamics appear in the limit via
the factor $\phi(1)$. Notably, the local autoregressive parameter
$c$ is replaced by $\phi(1)^{-1}c$ -- exactly as it would be if
$\{y_{t}\}$ were generated by a linear autoregression with a root
local to unity (cf.\ \citealp{Hansen99REStat}, p.\ 599). Indeed,
$\phi(1)=1$ when $k=1$, so in this case the two results coincide.

\begin{cor}
For the Tobit model \eqref{eq:Lbnd} with censoring point $\mathbf{L}$, $T^{-1/2}\tilde{y}_{\smlfloor{rT}}\indist Y_{\tilde{\theta}_{\phi}}(r)$,
where $\tilde{\theta}_\phi\defeq[\tilde{a},\phi(1)\tilde{b}_{0},\phi(1)^{-1}c]$.
\end{cor}

\subsection{OLS estimates}

\label{sec:OLS_asymptotics}

We first consider the case where $k=1$, as in the model \eqref{eq:dgp}, to
develop intuition for our results.

When estimating \eqref{eq:dgp} by OLS, we need to decide which deterministic
terms should be included in the regression. In the absence of censoring,
i.e.~if the data generating process were simply a linear autoregression,
the inclusion of a constant and a linear trend would render the distribution
of the OLS estimator of $\beta$ free of any nuisance parameters related
to the deterministic components.\footnote{Strictly speaking, this is true only if the autoregressive model is
formulated in `unobserved components' form (see e.g.\ \citet[Section 2.1]{AC94JBES}
as
\begin{align*}
y_{t} & =\mu+\delta t+y_{t}^{\ast} & y_{t}^{\ast} & =\beta y_{t-1}^{\ast}+u_{t}
\end{align*}
so that the presence (or absence) of a linear drift in $y_{t}$
is independent of the value of $\beta$, and so can always be removed
by deterministic detrending. By contrast, if the model is formulated
`directly' as
\[
y_{t}=\alpha+\beta y_{t-1}+u_{t},
\]
then the linear trend that is present when $\beta=1$ becomes an exponential
trend when $\beta$ is local to unity. In the present (censored) setting, we may
note that \eqref{eq:dgp} is \emph{not} equivalent to
\begin{align*}
y_{t} & =[\mu+\delta t+y_{t}^{\ast}]_{+} & y_{t}^{\ast} & =\beta y_{t-1}^{\ast}+u_{t}.
\end{align*}
(This
model is in fact the \emph{latent} dynamic Tobit referred to in Section~\ref{sec:intro}.)} Unfortunately, the nonlinearity introduced by the censoring entails
that $\alpha$ -- or rather, the local parameter $a$ -- will show
up in the limiting distribution of $\hat{\beta}_{T}$, \emph{irrespective}
of which deterministics are included in the regression. To permit
inferences to also be drawn on $a$, if required, we consider the OLS regression
of $y_{t}$ on a constant and $y_{t-1}$, i.e.
\begin{equation}
\begin{bmatrix}\hat{\alpha}_{T}\\
\hat{\beta}_{T}
\end{bmatrix}\defeq\left(\sum_{t=1}^{T}\begin{bmatrix}1 & y_{t-1}\\
y_{t-1} & y_{t-1}^{2}
\end{bmatrix}\right)^{-1}\sum_{t=1}^{T}\begin{bmatrix}1\\
y_{t-1}
\end{bmatrix}y_{t}\eqdef\mathcal{M}_{T}^{-1}m_{T}.\label{eq:ols_formula}
\end{equation}

In the stationary dynamic Tobit model, OLS is inconsistent (see e.g.~\citealp[Supplementary Material, Lemma B.1]{bykh_JBES}). However, as the following shows,
when $\beta$ is local to unity, consistency is restored. The reason
is that observations in the vicinity of zero accumulate only at rate
$T^{1/2}$, so that a vanishingly small fraction of the sample is
affected by the censoring.
\begin{thm}
\label{thm:OLS_LUR} Suppose Assumptions \ref{ass:INIT} and \ref{ass:DGP}
hold, with $k=1$ in \eqref{eq:tobitark}. Then
\begin{align}
\begin{bmatrix}T^{1/2}(\hat{\alpha}_{T}-\alpha)\\
T(\hat{\beta}_{T}-\beta)
\end{bmatrix} & \indist\begin{bmatrix}1 & \int_{0}^{1}J_{\theta}(r)\diff r\\
\int_{0}^{1}J_{\theta}(r)\diff r & \int_{0}^{1}J_{\theta}^{2}(r)\diff r
\end{bmatrix}^{-1}\begin{bmatrix}v(1)-c\int_{0}^{1}J_{\theta}(r)\diff r-b_{0}-a\\
\sigma\int_{0}^{1}J_{\theta}(r)\diff W(r)
\end{bmatrix}\nonumber \\
 & \eqdef\mathcal{J}_{\theta}^{-1}\mathcal{U}_{\theta}\eqdef\begin{bmatrix}\mathfrak{a}_{\theta}\\
\mathfrak{b}_{\theta}
\end{bmatrix}.\label{eq:OLS_asy}
\end{align}
\end{thm}
\begin{rem}
Letting $J_{\theta}^{\mu}(r)\defeq J_{\theta}(r)-\int_{0}^{1}J_{\theta}(s)\diff s$,
an alternative expression for the limiting distribution of $\hat{\beta}_{T}$
is given by
\begin{equation}
\mathfrak{b}_{\theta}=\frac{J_{\theta}^{\mu}(1)^{2}-J_{\theta}^{\mu}(0)^{2}-\sigma^{2}}{2\int(J_{\theta}^{\mu}(r))^{2}\diff r}-c.\label{eq:altexpression}
\end{equation}
This agrees with the limiting distribution
that would be obtained in the linear autoregressive model, except with
$J_{\theta}(\cdot)$ taking the place of the usual Ornstein--Uhlenbeck
process. (See Appendix~\ref{proof_altexpression} for details.)
\end{rem}
For the case of general $k\geq1$, let $\vec{\phi}\defeq(\phi_{1},\ldots,\phi_{k-1})^{\trans}$,
and
\begin{equation}
(\hat{\alpha}_{T},\hat{\beta}_{T},\hat{\phi}_{1,T},\ldots,\hat{\phi}_{k-1,T})\defeq\argmin_{(a,b,f_{1},\ldots,f_{k-1})}\sum_{t=1}^{T}\left(y_{t}-a-by_{t-1}-\sum_{i=1}^{k-1}f_{i}\Delta y_{t-i}\right)^{2}\label{eq:olsARk}
\end{equation}
denote the OLS estimators of the parameters of \eqref{eq:tobitark}.
Since, as the next results shows, the limiting distributions of $(\hat{\alpha}_{T},\hat{\beta}_{T})$
depend on $\phi(1)$, a consistent estimate of that quantity is needed
to compute valid critical values for test statistics based on these
estimators. The following also guarantees the consistency of $\hat{\phi}(1)\defeq1-\sum_{i=1}^{k-1}\hat{\phi}_{i,T}$.
\begin{thm}\label{thm:olsARk}
Suppose Assumptions \ref{ass:INIT}--\ref{ass:JSR}
hold. Then $\hat{\vec{\phi}}_{T}\inprob\vec{\phi}$, and
\begin{equation}
\begin{bmatrix}T^{1/2}(\hat{\alpha}_{T}-\alpha)\\
T(\hat{\beta}_{T}-\beta)
\end{bmatrix}\indist\begin{bmatrix}1 & \int Y_{\theta_{\phi}}(r)\diff r\\
\int Y_{\theta_{\phi}}(r)\diff r & \int Y_{\theta_{\phi}}^{2}(r)\diff r
\end{bmatrix}^{-1}\begin{bmatrix}\phi(1)[Y_{\theta_{\phi}}(1)-b_{0}-c_{\phi}\int Y_{\theta_{\phi}}(r)\diff r]-a\\
\sigma\int Y_{\theta_{\phi}}(r)\diff W(r)
\end{bmatrix}\label{eq:olslimdist}
\end{equation}
\end{thm}
\begin{rem}
The parameters of the Tobit model \eqref{eq:Lbnd} with censoring point $\mathbf{L}$ can be estimated as in \eqref{eq:olsARk}, with
$\tilde{y}_{t}$ in place of $y_{t}$ and with $\tilde{\theta}_\phi$ replacing $\theta_\phi$ in \eqref{eq:olslimdist}.
\end{rem}

The theorem shows that, depite the nonlinearity of the Tobit model, OLS is consistent for all parameters in the presence of a near unit root. The associated limit distribution theory for $\hat{\beta}_T$ provides a basis for the unit root tests developed in the next section, yielding a test that is both easy to compute, and semiparametric with respect to the distribution of the innovations. (For examples in which the erroneous imposition of normality can have undesirable consequences, in the context of a dynamic Tobit model, see the empirical illustrations in \citet{stat_paper} and \citet{bykh_JBES}.)

\subsection{Unit root tests}

The preceding results allow us to conduct asymptotically valid hypothesis
tests on key parameters of the dynamic Tobit: in particular, to test
the hypothesis of a unit root in this setting. This may be of interest for several reasons.
For example, whether the variance of the errors made in forecasting $y_t$
remains bounded, or grows without bound at progressively longer forecast horizons,
depends crucially on the presence of a unit root. In a setting with multiple series, one
or more of which are non-negative, the presence of unit roots may also lead to
spurious regressions or, more constructively, allow long-run
equilibrium relationships to be identified from the (nonlinear) cointegrating
relationships between the series (see \citealp{DMW22}).

In a linear autoregressive model, the presence of a unit root -- equivalently, the sum of
the autoregressive coefficients being unity ($\beta=1$) -- necessarily imparts a stochastic trend to $\{y_t\}$.
However, in the dynamic Tobit the value of the intercept also matters. In particular,
a fixed negative intercept ($\alpha<0$ and not drifting toward zero) would continually push the process back towards the
censoring point, thereby rendering it stationary (for $k=1$, see \citealp[Theorem 3]{bykh_JBES}).
Thus to the extent that the purpose of a test for a unit root is to
test for the presence of a stochastic trend in $\{y_t\}$, rather
than to detect a unit root per se, it may be
considered more appropriate to test the null that $\alpha=0$
\emph{and} $\beta=1$, as opposed to merely the restriction that $\beta=1$,
with it being desirable to reject this null in favour of a stationary
alternative, when either $\beta<1$ (exactly as in a linear model),
or when $\beta=1$ but $\alpha<0$.\footnote{When $k > 1$, the
above must be qualified somewhat, because of the possibility that
the higher-order nonlinear dynamics of the system may generate explosive
trajectories even when $\beta < 1$. For stationary alternatives that are
\emph{local} to unity in the sense that $\beta = \exp(c/T)$ for some
$c<0$, this is excluded by \ref{ass:JSR}. For non-local alternatives,
some further condition (i.e.\ in addition to $\beta < 1$) on the autoregressive
system is needed to ensure stationarity: see e.g.\ \citet{stat_paper} or \citet[Sec.~3]{DMW22}.}

To construct our test statistics, we need an estimate of the error
variance $\sigma^{2}$. We use $\hat{\sigma}_{T}^{2}\defeq\frac{1}{T}\sum_{t=1}^{T}\hat{u}_{t}^{2}$,
where
\begin{equation}
\hat{u}_{t}\defeq y_{t}-\hat{\alpha}_{T}-\hat{\beta}_{T}y_{t-1}-\sum_{i=1}^{k-1}\hat{\phi}_{i,T}\Delta y_{t-i}.\label{eq:residuals}
\end{equation}
That is, $\{\hat{u}_{t}\}$ are the OLS residuals, computed as if
$y_{t}$ were not subject to censoring. Let $\mathcal{M}_{T}\defeq\sum_{t=1}^{T}\vec x_{t}\vec x_{t}^{\trans}$,
where $\vec x_{t}\defeq(1,y_{t-1},\Delta y_{t-1},\ldots,\Delta y_{t-k+1})^{\trans}$.
\begin{cor}
\label{thm:tstatARk} Suppose Assumptions~\ref{ass:INIT} and \ref{ass:DGP}
hold. If either: $k=1$ in \eqref{eq:tobitark}; or $k>1$, and \ref{ass:MOM}
and \ref{ass:JSR} hold, then $\hat{\sigma}_{T}^{2}\inprob\sigma^{2}$
and
\begin{align}
t_{\alpha,T} & \defeq\frac{\hat{\alpha}_{T}-\alpha}{\hat{\sigma}_T\sqrt{\mathcal{M}_{T}^{-1}(1,1)}}\indist\frac{\mathfrak{a}_{\theta_{\phi}}}{\sigma\sqrt{\mathcal{J}_{\theta_{\phi}}^{-1}(1,1)}} & t_{\beta,T}\defeq & \frac{\hat{\beta}_{T}-\beta}{\hat{\sigma}_T\sqrt{\mathcal{M}_{T}^{-1}(2,2)}}\indist\frac{\mathfrak{b}_{\theta_{\phi}}}{\sigma\sqrt{\mathcal{J}_{\theta_{\phi}}^{-1}(2,2)}},\label{eq:tstats}
\end{align}
where $\mathcal{M}^{-1}_{T}(i,j)$ denotes the $(i,j)$ element of $\mathcal{M}_{T}^{-1}$.
\end{cor}
This result allows us to conduct a one-sided test of a unit root versus
a stationary alternative, which rejects when $t_{\beta,T}\leq c$,
where $c$ is drawn from an appropriate quantile of the asymptotic
distribution of $t_{\beta,T}$. Under the null of a unit root $a=c=0$, the limiting distribution of $t_{\beta,T}$ in \eqref{eq:tstats} depends (continuously) on the model parameters only through $b_0 \phi(1)/\sigma$. This follows from the fact that for $\theta_{\phi}=(0,\phi(1)b_0,0)$,
\begin{align*}
\sigma^{-1} J_{\theta_{\phi}}(r) &= \sigma^{-1}\left\{\phi(1)b_0+\sigma W(r)+\sup_{r^{\prime}\leq r}[-\phi(1)b_0-\sigma W(r^{\prime})]_{+}\right\}\\
&=\tfrac{\phi(1)b_0}{\sigma}\left\{1+\tfrac{\sigma}{\phi(1)b_0}W(r)+\sup_{r^{\prime}\leq r}\left[-1-\tfrac{\sigma}{\phi(1)b_0}W(r^{\prime})\right]_{+}\right\}
\end{align*}
and thus $\sigma\sqrt{\mathcal{J}_{\theta_{\phi}}^{-1}(2,2)}$ and $\mathfrak{b}_{\theta_{\phi}}$ (as defined in \eqref{eq:altexpression} above) do not change so long as $b_0 \phi(1)/\sigma$ remains fixed. Table~\ref{test_quantiles} tabulates the critical values corresponding to the relevant quantiles of the asymptotic distribution of $t_{\beta,T}$, as a function of $b_0 \phi(1)/\sigma \in [0,2.5]$. One can see that the smaller the ratio of the parameters, the larger the corresponding critical values, for every significance level. Further, as the final two lines of the table and the further discussion in Section \ref{sec:b0_effect} indicate, for values of $b_0 \phi(1)/\sigma$ in excess of $2.5$, the critical values coincide (to within two decimal places) with those of a conventional ADF $t$ test (for a linear autoregression).

With the aid of the tabulated critical values, a test of $H_0 : a=c=0$ may thus be carried out as follows. (For a general lower bound $\mathbf{L} \neq 0$, first subtract it from the data, as described in Section~\ref{sec:nonzero_bound}.)

\begin{urtest}~
\begin{enumerate}[label=\arabic*., ref=\arabic*]
  \item Regress $y_t$ on $(1,y_{t-1},\Delta y_{t-1},\ldots,\Delta y_{t-k+1})^{\trans}$ using OLS.
  \item Calculate the $t_{\beta,T}$ statistic \eqref{eq:tstats} with $\beta=1$.
  \item \label{enu:crit} Let $\hat{\phi}(1) \defeq 1-\sum_{i=1}^{k-1}\hat{\phi}_{i,T}$ and $\hat{b}_0 \defeq T^{-1/2}y_1 $, where $y_1$ is the first observation in the sample. Compare $t_{\beta,T}$ with the critical values in Table~\ref{test_quantiles}, for the row corresponding to the value nearest to $\hat{b}_0 \hat{\phi}(1)/\hat{\sigma}_{T}$ (or use the conventional ADF critical values if this value exceeds $2.5$).\footnote{Alternatively one could use a parametric bootstrap procedure to estimate the quantiles of the null distribution of the test statistic, following the approach of \citet[Theorem~2]{CX14JoE}, who prove the validity of a bootstrap procedure in a related setting with $k=1$. Their proof could be directly transposed to our setting to justify a procedure based on $y_1^{(r)}=\sqrt{T'}\hat{\phi}(1) y_1/\sqrt{T}$ and $y_t^{(r)}=[y_{t-1}^{(r)}+u_t^r]_+$, $u_t^r\thicksim i.i.d.~\mathcal{N}(0,\hat{\sigma}_{T}^{2}),\, t>1$, $r=1,\ldots,R$, where $T'\geq T$ is the length of the simulated series and $R$ is the number of simulated series.}
\end{enumerate}
\end{urtest}


\begin{table}[t]
\begin{tabular}{c||c|c|c}
\hline
\diagbox[width=2.5cm, height=1.3cm]{$b_0\phi(1)/\sigma$}{Size} & \centering$1\%$ & \centering$5\%$ & \centering$10\%$\tabularnewline
\hline
\hline
    $0.0$ & $-4.69$ & $-3.77$ & $-3.34$ \\
    $0.1$ & $-4.58$ & $-3.66$ & $-3.22$ \\
    $0.2$ & $-4.49$ & $-3.57$ & $-3.14$ \\
    $0.3$ & $-4.38$ & $-3.49$ & $-3.07$ \\
    $0.4$ & $-4.25$ & $-3.40$ & $-3.00$ \\
    $0.5$ & $-4.11$ & $-3.31$ & $-2.93$ \\
    $0.6$ & $-3.97$ & $-3.22$ & $-2.87$ \\
    $0.7$ & $-3.85$ & $-3.15$ & $-2.81$ \\
    $0.8$ & $-3.75$ & $-3.08$ & $-2.75$ \\
    $0.9$ & $-3.67$ & $-3.03$ & $-2.71$ \\
    $1.0$ & $-3.60$ & $-2.99$ & $-2.68$ \\
    $1.1$ & $-3.56$ & $-2.96$ & $-2.65$ \\
    $1.2$ & $-3.52$ & $-2.94$ & $-2.63$ \\
    $1.3$ & $-3.50$ & $-2.92$ & $-2.62$ \\
    $1.4$ & $-3.48$ & $-2.90$ & $-2.61$ \\
    $1.5$ & $-3.47$ & $-2.89$ & $-2.60$ \\
    $1.6$ & $-3.46$ & $-2.89$ & $-2.59$ \\
    $1.7$ & $-3.45$ & $-2.88$ & $-2.58$ \\
    $1.8$ & $-3.45$ & $-2.87$ & $-2.58$ \\
    $1.9$ & $-3.44$ & $-2.87$ & $-2.58$ \\
    $2.0$ & $-3.44$ & $-2.87$ & $-2.57$ \\
    $2.5$ & $-3.43$ & $-2.86$ & $-2.57$ \\
\hline
    ADF & $-3.43$ & $-2.86$ & $-2.57$ \\
\hline
\end{tabular}\caption{Critical values for the Tobit ADF test for different values of $b_0\phi(1)/\sigma$ (based on $10^{7}$ Monte Carlo simulations of $y_{t}=\left[y_{t-1}+u_{t}\right]_{+}$,
$y_{0}=b_{0}\sqrt{T}$, $u_{t}\thicksim\text{i.i.d.}~\mathcal{N}(0,1)$, $T=10^5$). The final line reports the critical values appropriate to  an ADF $t$ test in a linear autoregression.}
\label{test_quantiles}
\end{table}

As the simulations in the following section illustrate, our test indeed has the desirable properties outlined
above, in the sense of tending to reject both when either $\beta<1$, or when
$(\beta=1 \sep \alpha<0)$, i.e.\ it has power to reject the null
whenever $\{y_t\}$ is stationary. In the event of a rejection, the reason for that rejection -- i.e.\ whether
this is due to $\alpha<0$ or $\beta<1$ -- may be further investigated
with the aid of LAD estimates of the model parameters, which by \citet[Theorem 6]{bykh_JBES}
are consistent and asymptotically normal in the stationary region.
(By contrast, due to the relatively greater frequency of censoring,
OLS is not consistent in the stationary region, as discussed in Section~4.3
of that work.)

\subsection{Specification testing}

An implication of the arguments given in the proof of Corollary~\ref{thm:tstatARk}, which is in line with alternative Tobit representation \eqref{eq:arkwithytminus}, is that the OLS residuals $\{\hat{u}_t\}$
are consistent for $\{u_t-y_t^{-}\}$. However, because $y_t=0$ occurs relatively infrequently, $\{u_t-y_t^{-}\}$ and $\{u_t\}$
coincide sufficiently closely that e.g.\ $T^{-1}\sum_{t=1}^{T} (u_t-y_t^{-})^2 = T^{-1}\sum_{t=1}^{T} u_t^2 + o_p(1)$. In view of this,
there appears to be some justification for continuing to use standard residual-based specification tests, such as for
heteroskedasticity, serial correlation, or non-normality, when the dynamic Tobit is estimated by OLS (for an
overview, see \citealp{KL17book}, Sec.~2.6--2.7). The proper second order asymptotics for such tests is outside of the focus of the present paper and is left for future research.

\section{Simulations}

\label{sec:montecarlo}

We now illustrate how the values of the initial condition $b_{0}$
and the localising parameters $a$ and $c$ affect the distribution of
$t_{\beta}$, and compare the performance of a test based on critical
values derived from Corollary~\ref{thm:tstatARk} with one based
on the conventional ADF critical values (\citealp{DF_test}), when the
data is subject to censoring.

\subsection{Effect of $b_{0}$ and the connection to the conventional ADF test}\label{sec:b0_effect}

Figure~\ref{t_ratio_effect_b0} depicts how a change in $b_{0}$
shifts the probability density (PDF) and the cumulative distribution (CDF) of $t_{\beta}$. As $b_{0}$ moves further above
zero, the density shifts progressively to the right, as the probability
that any trajectory of $K_{\theta}$ (initialised at $b_{0}$) will
reach zero, and so be subject to censoring, correspondingly declines.
Indeed, once $b_{0}$ is sufficiently large to make this probability
negligible, the density becomes visually indistinguishable from that generated
by a linear model (the solid green line in Figure~\ref{t_ratio_effect_b0}), which is invariant to $b_{0}$. (Under the parametrisation
used in the figure, this occurs when $b_{0}=2$; in general, this
will depend on the magnitude of $\phi(1)b_{0}/\sigma$, in accordance
with Theorem~\ref{thm:ytARk}.) The rightward shift of these distributions,
as $b_0$ grows, is similarly manifest in the critical values given
in Table~\ref{test_quantiles} above.

Surprisingly, the CDFs exhibit stochastic ordering, in the sense that the distributions corresponding to higher values of $b_{0}$ first-order stochastically dominate those with lower values of $b_0$ (holding all other parameters constant). This is in line with the monotonically increasing (in $b_0\phi(1)/\sigma$) quantiles in Table \ref{test_quantiles}. In particular, the null distribution of the conventional ADF $t$ test -- i.e.\ that appropriate to a linear autoregression -- stochastically dominates the null distribution appropriate to a dynamic Tobit. Therefore, when data is generated by a dynamic Tobit with a unit root, the conventional ADF test will tend to over-reject: in the worst case, which occurs when $b_0=0$, we find that a nominal $5$ per cent test will in fact reject $20$ per cent of the time. The intuition is that the censoring causes the trajectories of $\{y_{t}\}$ to appear stationary, masking the presence of a unit root.

\begin{figure}[t]
\begin{subfigure}{.48\textwidth} \centering \includegraphics[width=1\linewidth]{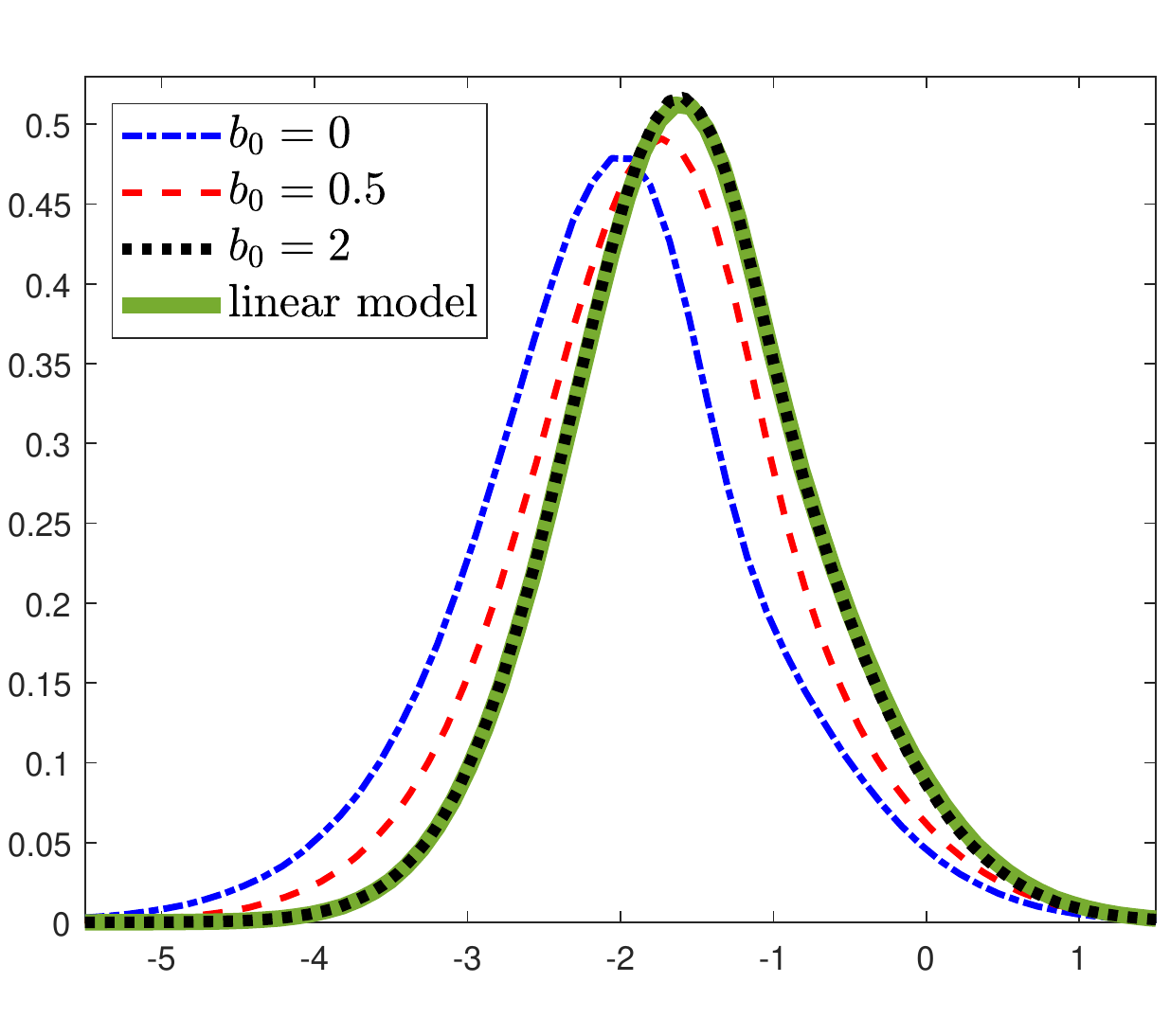}
\caption{Probability density functions.}
\label{linear_vs_tobit_pdf} \end{subfigure}
\begin{subfigure}{.48\textwidth}
\centering \includegraphics[width=1\linewidth]{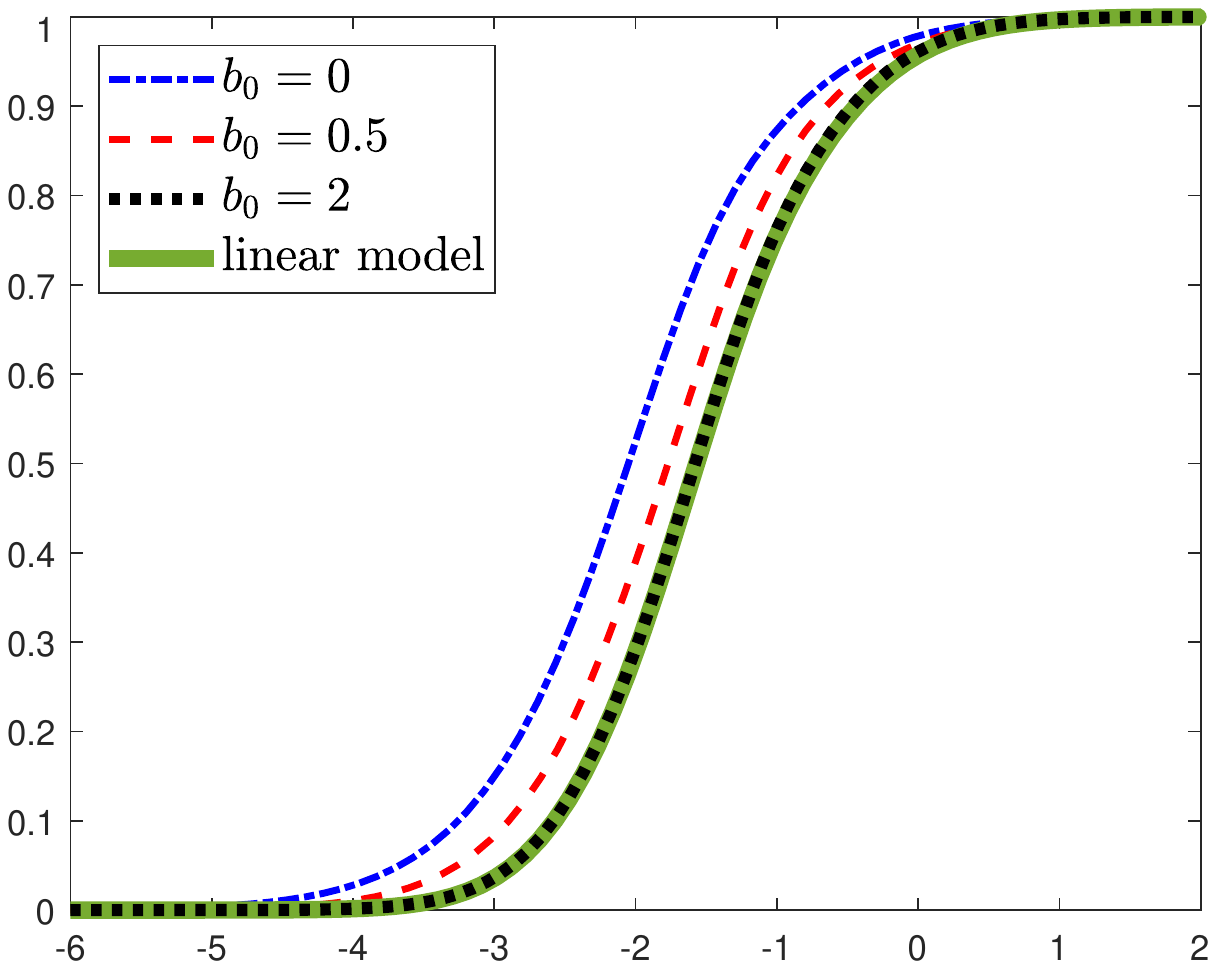}
\caption{Cumulative distribution functions.}
\label{linear_vs_tobit_cdf} \end{subfigure}
\caption{CDFs and PDFs of t-ratio $t_{\beta}$ under the Tobit
and linear models. Data generating process is $y_{t}=\left[y_{t-1}+u_{t}\right]_{+}$,
$y_{0}=b_{0}\sqrt{T}$ for Tobit model and $y_{t}^{\ell}=y_{t-1}^{\ell}+u_{t},\,y_{0}^{\ell}=0$
for a linear model, $u_{t}\thicksim\text{i.i.d.}~\mathcal{N}(0,1)$.
Data is obtained from $10^{7}$ samples of length $T=1000$.}
\label{t_ratio_effect_b0}
\end{figure}


For the remainder of this section, all simulations are conducted with
$y_{0}=b_{0}=0$.

\subsection{Effects of $a$ and $c$}

Figure~\ref{t_ratio_effect_a_c} shows how a change in local intercept
$a$ (left panel) and local slope coefficient $c$ (right panel) affects
the density of $t_{\beta}$. The means of these distributions across
a range of values for $a$ and $c$ are also reported in Table~\ref{mean_stat}.
We can see that as $a$ or $c$ fall further below zero, the distribution
of $t_{\beta}$ (both its mean and its entire probability mass) shifts
leftward -- with the opposite effect being observed when these parameters
are progressively raised above zero.

\begin{figure}[t]
\begin{subfigure}{.45\textwidth} \centering \includegraphics[width=1\linewidth]{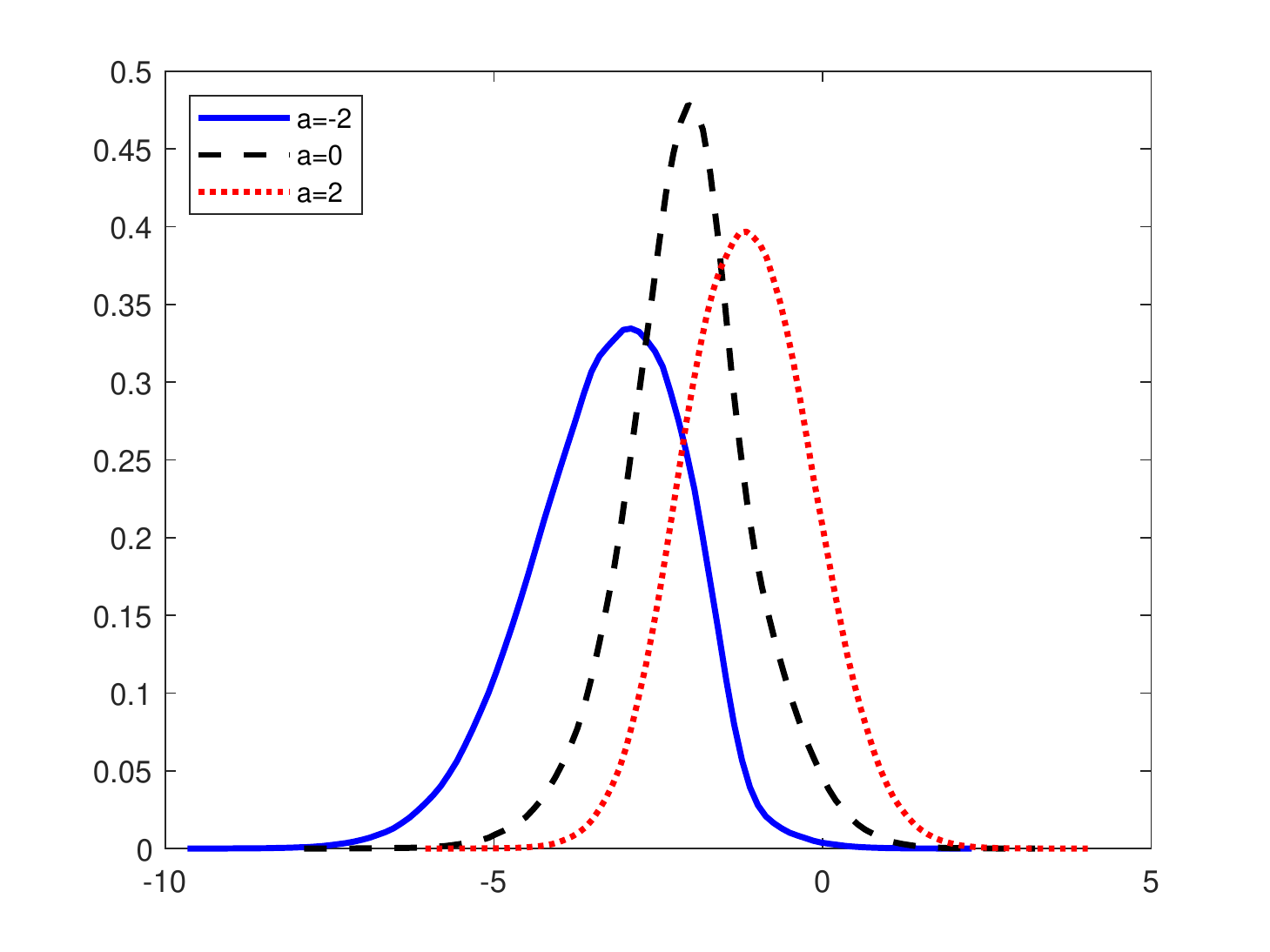}
\caption{Effect of a change in $a$ when $c=0$.}
\label{t_ratio_effect_a} \end{subfigure}\begin{subfigure}{.45\textwidth}
\centering \includegraphics[width=1\linewidth]{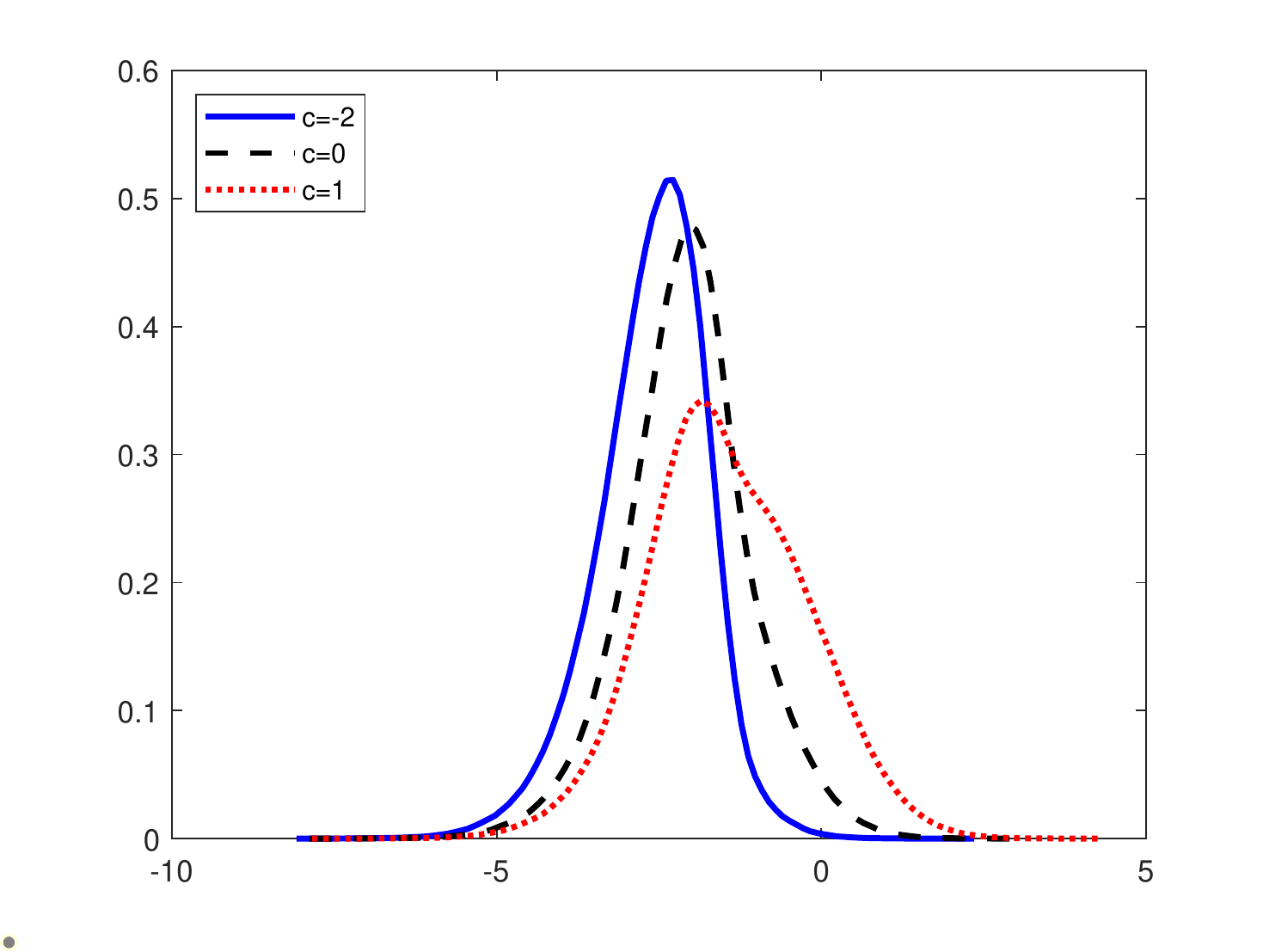}
\caption{Effect of a change in $c$ when $a=0$.}
\label{t_ratio_effect_c} \end{subfigure} \caption{Densities of t-ratio $t_{\beta}$ for various values of $a,c$. Data
generating process is $y_{t}=\left[\tfrac{a}{\sqrt{T}}+\left(1+\tfrac{c}{T}\right)y_{t-1}+u_{t}\right]_{+}$,
$y_{0}=0$, $u_{t}\thicksim\text{i.i.d.}~\mathcal{N}(0,1)$. Data
is obtained from $10^{6}$ samples of time series of length $T=1000$.}
\label{t_ratio_effect_a_c}
\end{figure}

\begin{table}[t]
\begin{tabular}{c||c|c|c|c|c|c|c}
\hline
\backslashbox{a}{c}  & -5 & -2  & -1  & 0  & 1  & 2  & 5\tabularnewline
\hline
\hline
-5  & -6.09  & -5.70  & -5.56  & -5.42  & -5.26  & -5.09  & -4.37\tabularnewline
-2  & -4.24  & -3.70  & -3.49  & -3.27  & -3.00  & -2.62  & 8.93\tabularnewline
-1  & -3.69  & -3.10  & -2.88  & -2.62  & -2.26  & -1.51  & 23.05\tabularnewline
0  & -3.20  & -2.60  & -2.37  & -2.06  & -1.46  & 0.03  & 43.23\tabularnewline
1  & -2.82  & -2.26  & -2.00  & -1.56  & -0.57  & 1.85  & 68.26\tabularnewline
2  & -2.58  & -2.06  & -1.76  & -1.13  & 0.28  & 3.68  & 96.73\tabularnewline
5  & -2.52  & -2.05  & -1.57  & -0.48  & 2.18  & 9.00  & 193.14\tabularnewline
\hline
\end{tabular}\caption{Mean of t-ratio $t_{\beta}$ for various values of $a,c$. Data
generating process is $y_{t}=\left[\tfrac{a}{\sqrt{T}}+\left(1+\tfrac{c}{T}\right)y_{t-1}+u_{t}\right]_{+}$,
$y_{0}=0$, $u_{t}\thicksim\text{i.i.d.}~\mathcal{N}(0,1)$. Data
is obtained from $10^{6}$ samples of time series of length $T=1000$.}
\label{mean_stat}
\end{table}

\subsection{Power}

\label{sec:power}

The preceding illustrates how changes in $a$ and/or $c$ may shift
the distribution of $t_{\beta}$ in either direction, and so will
affect the ability of the test to reject the null of a unit root (i.e.\ $H_{0}:\alpha=0,\beta=1$).
Power envelopes (rejection probabilities for a nominal 5 per cent,
one-sided test), are displayed in Figure~\ref{fig:power}. These
show that, on the one side, more negative values of $a$ and/or $c$
make it easier for the test to reject the null in favour of a stationary
alternative. The power eventually reaches $100$ per cent, indicating
the consistency of the test against fixed alternatives in this region.
This tendency to reject the null, as $a$ falls below zero, is in
fact a desirable property of the test in this setting, since having
$\alpha<0$ in the dynamic Tobit implies (for $k=1$) that $\{y_{t}\}$ is stationary,
even when $\beta=1$ \citep[Theorem 3]{bykh_JBES}.

On the other side,
positive values of $c$ move $\{y_{t}\}$ into the explosive region,
and our tendency to not reject in these cases is entirely consistent
with the use of a one-sided test, as it is in a linear model (Figure~\ref{power_c}).
Although we also fail to reject for sufficiently positive values of
$a$ (Figure~\ref{power_a}), in such cases an upward trend in $\{y_{t}\}$
would become discernable, and would carry the process away from the
censoring point. Since $\{y_{t}\}$ would then make few (if any) visits
to the censoring point, if one were interested in testing
the null of a unit root against a \emph{trend} stationary alternative,
in this case, a conventional ADF test with intercept and trend would be appropriate.

\begin{figure}[t]
\begin{subfigure}{.45\textwidth} \centering \includegraphics[width=1\linewidth]{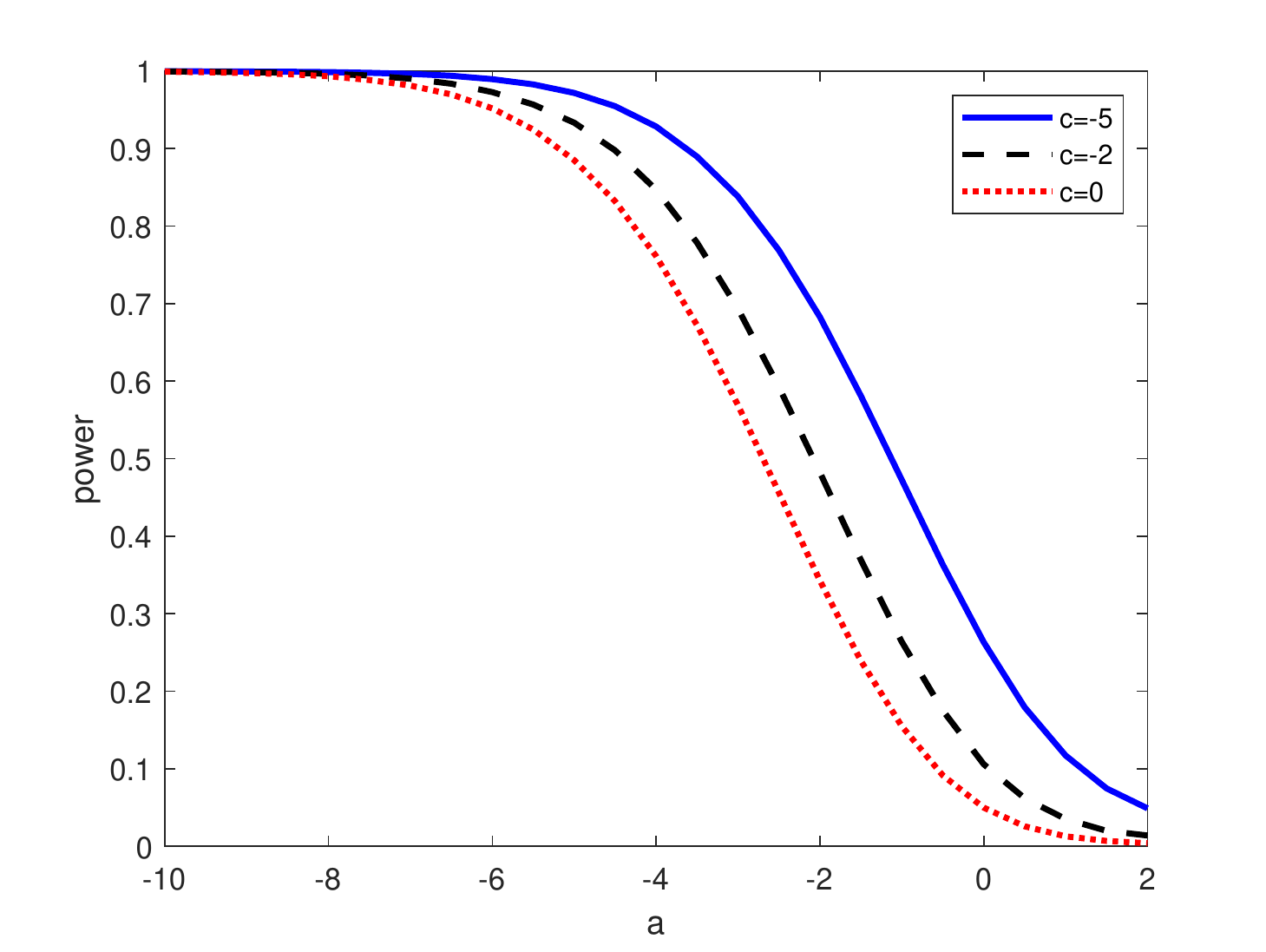}
\caption{Power envelopes with respect to $a$.}
\label{power_a} \end{subfigure}\begin{subfigure}{.45\textwidth}
\centering \includegraphics[width=1\linewidth]{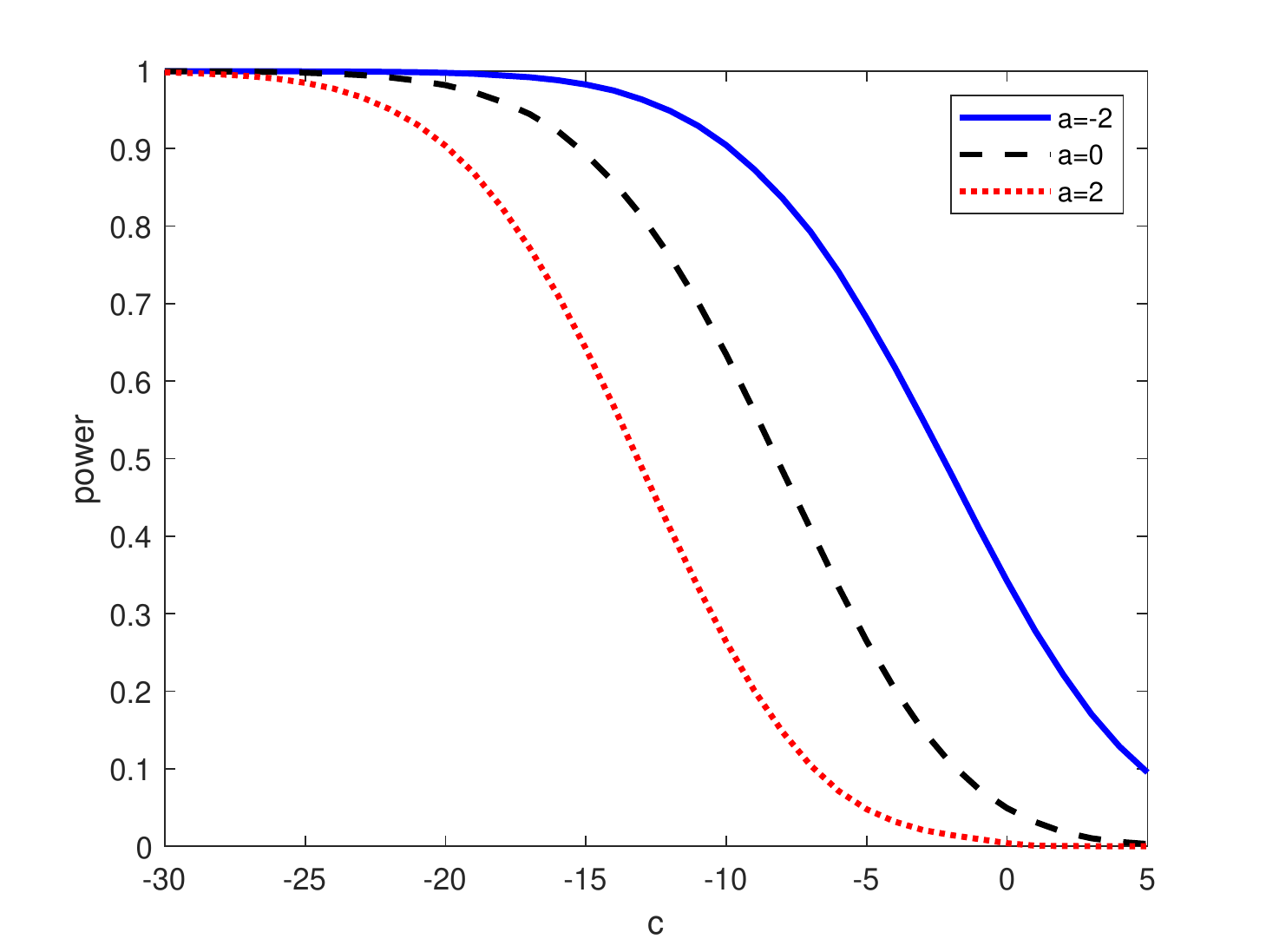} \caption{Power envelopes with respect to $c$.}
\label{power_c} \end{subfigure} \caption{Power envelopes with respect to $a$ and $c$. Data generating process
$y_{t}=\left[\tfrac{a}{\sqrt{T}}+\left(1+\tfrac{c}{T}\right)y_{t-1}+u_{t}\right]_{+}$,
$y_{0}=0$, $u_{t}\thicksim\text{i.i.d.}~\mathcal{N}(0,1)$. Data
is obtained from $10^{5}$ samples of time series of length $T=1000$.}
\label{fig:power}
\end{figure}

\section{Empirical illustration}

\label{sec:empirical}

In this section we illustrate the use of our methods through an application
to testing for unit roots in nominal exchange rates, when these are
subject to a one-sided bound. Unit roots are routinely detected in
these series by conventional tests in empirical work (see e.g.\ \citealp{BB89JFE};
\citealp{SV06JBF}, p.~3156; \citealp{HP10JBES}, p.~107). Their
presence is also manifested in the robust performance of exchange rate
forecasts based on random walks, which more elaborate models have
struggled to beat consistently (\citealp{Ros13JEL}). (For a discussion of
the theoretical basis for the presence of a unit root in exchange rates,
in the context of open economy New Keynesian models, we refer the reader
to Section~2.1 of \citealp{engel2014exchange}.) Here we examine
how censoring, as introduced by the deliberate action of a central
bank to keep exchange rates above (or below) a nominated threshold,
may alter our assessment of the evidence for or against a unit root,
depending on whether that censoring is accounted for (cf.~\citealp[Sec.~6.1]{cavaliere2005}).

\begin{figure}
\noindent \begin{centering}
\includegraphics[clip,width=0.9\textwidth,viewport=50bp 160bp 792bp 435bp]{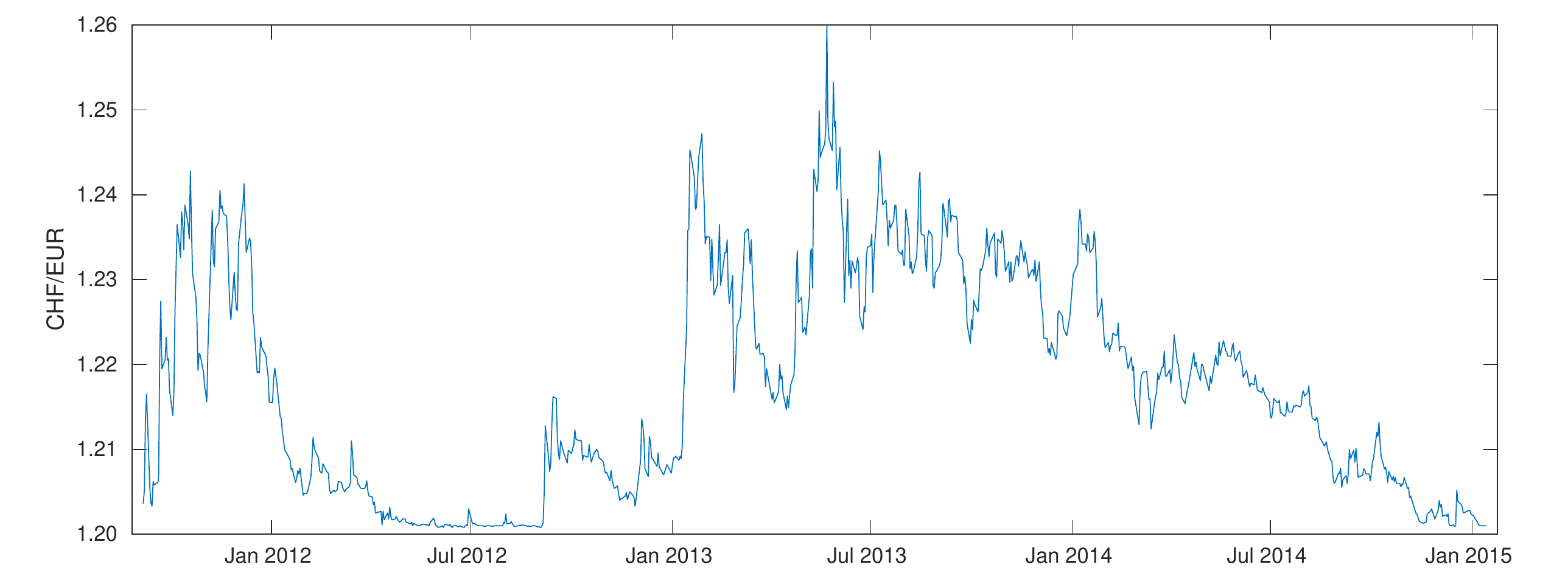}
\par\end{centering}
\caption{CHF/EUR exchange rate (Source: ECB).}
\label{fig:xrate}
\end{figure}

In September 2011, in response to the ongoing appreciation of the
Swiss franc, and with the policy rate effectively at the zero lower
bound, the Swiss National Bank (SNB) instituted a floor on the euro--Swiss
franc exchange rate of 1.20 francs per euro (\citealp{Jor16}; \citealp{Her22RIE}).
With the exchange rate well below the floor on the previous day, the
immediate intervention of the SNB was required to make the floor effective
upon its introduction on 6 September. The floor remained in effect
until the end of 15 January 2015. As can be seen from Figure~\ref{fig:xrate},
during that period, the exchange rate spent most of its time well
above the the floor (reaching a peak of 1.26 in May 2013), with two
notable exceptions: the periods of April to August 2012, and from
November 2014 until the end of the policy, both of which triggered
action by the SNB to prevent the floor from being violated (see Figure~9
in \citealp{Her22RIE}). (For a further discussion of the floor and
its aftermath, see also \citealp{vSTB21IRF}.) The observed trajectory
of the exchange rate, during these episodes, is thus more plausibly
consistent with a dynamic Tobit than it is with a linear autoregression.

In this setting, the presence of a unit root would entail that the
exchange rate tends to make extended sojourns away from the floor,
visits to which accumulate only at rate $T^{1/2}$, where $T$ denotes
the duration of the policy. By contrast, in the stationary dynamic
Tobit, visits to the floor would accumulate at rate $T$. Thus the
failure to reject the null of a unit root would signal to the central
bank that it may need to intervene relatively infrequently in the
foreign exchange market, in order to make the floor effective, something
that is in turn highly relevant to the future cost and viability of
maintaining that floor.

Our data is drawn from the European Central Bank's (ECB) daily reference
exchange rate series (code: EXR.D.CHF.EUR.SP00.A) from the period
during which the floor was operational (6 Sep 2011 to 15 Jan 2015),
transformed by taking logarithms. To select the lag order $k$ used
to compute $t_{\beta}$, we evaluated autoregressive models with $k\in\{1,\ldots,15\}$
using the Akaike and Bayesian information criteria; both selected
a model with only one lag. For this model, we obtained an ADF statistic
slightly above $-2.87$: so that if the censoring were ignored, and
this statistic referred to the conventional ADF critical values (Table~\ref{test_quantiles}),
there would be just sufficient evidence
to reject the null of a unit root at the five per cent level. On the
other hand, the unit root is not rejected at conventional significance
levels under the dynamic Tobit. By simulating the asymptotic distribution
of $t_{\beta}$, given in Corollary~\ref{thm:tstatARk}, we
compute the $p$-value appropriate to the dynamic Tobit as
either $0.2$ (with $b_{0}=0$ imposed) or $0.18$ (using $\hat{b}_{0}=T^{-1/2}(y_{0}-\mathbf{L})$
with $\mathbf{L}=\log(1.20)$); similar results also obtain when $k=2,3$.
This places $t_{\beta,T}$ much further from the critical region,
thereby lending support to the hypothesis that the data is consistent
with a dynamic Tobit model with a unit root.

\section{Conclusion}

\label{sec:concl}

This paper extends local to unity asymptotics to the setting of a dynamic Tobit
model. Censoring fundamentally changes the analysis and requires new
tools to derive the asymptotics. We obtain novel limit theorems
for convergence to regulated processes, that is, to processes constrained
to lie above a threshold. The effect of that censoring on the limiting
distribution of our test statistics varies according to
the proximity of the initialisation to the censoring point, with the
distributions associated with the linear model re-emerging
as that initialisation moves sufficiently far from the censoring
point.

Our results underpin the development of a unit root test appropriate
to censored data generated by a dynamic Tobit model.
In contrast to the setting of a linear model, here
the presence of a stochastic trend entails
restrictions on both $\alpha$ and $\beta$, since $\alpha<0$
(with $\beta=1$) can be consistent with stationarity (and is consistent with stationarity for $k=1$).
Nonetheless, a test of this null can still be effected by the usual
$t_{\beta}$ statistic, using adjusted critical values. We provide
an empirical illustration of our methods to testing for a unit root
in nominal exchange rates, when these are subject to a one-sided
bound.

The results of this paper could be developed further in a number of
directions. One possibility would be to extend our results beyond
the local to unity setting, by allowing for moderate deviations from
a unit root (\citealp{GP06JTSA}), with the aim of establishing (uniformly)
valid confidence intervals for $\beta$, as
per \citet{Mik07Ecta,mikusheva2012one} in the linear model. As we depart from the linear model, new types of asymptotics emerge,
which may involve $c$ in different ways (cf.~\citet{bykh_boundary} where $c$ is no longer a constant but varies with time).

The analysis of $c\to\pm\infty$ for the censored model may be undertaken in conjunction with the derivation of the asymptotics
of least absolute deviations regression or maximum likelihood
estimators of the model, which would enjoy consistency across a wider
domain than does OLS. For such extensions, the results of \prettyref{sec:yt_asymptotics},
regarding the asymptotics of $\{y_{t}\}$, are likely to be of fundamental importance.

For another direction in which our analysis could be extended, recall from Section~\ref{sec:intro}
that the dynamic Tobit provides
the kernel of more elaborate, multivariate models that allow for the
possibility of censoring and/or some other threshold-related nonlinearity.
The analysis of (exact) unit roots and cointegration, in the setting
of such a model, is developed by \citet{DMW22}, with the aid of our
results.

\appendix

\section{Proofs of results for $k=1$.}

\subsection{Limiting distribution of $T^{-1/2}y_{\protect\smlfloor{nr}}$.}
\begin{lem}
\label{lem:wkccens} Suppose that $x_{t}=[x_{t-1}+v_{t}]_{+}$, for
$t=1,\ldots,T$, and $T^{-1/2}(x_{0}+\sum_{s=1}^{\smlfloor{rT}}v_{s})\indist V(r)$
on $D[0,1]$. Then on $D[0,1]$,
\[
T^{-1/2}x_{\smlfloor{rT}}\indist V(r)+\sup_{r^{\prime}\leq r}[-V(r^{\prime})]_{+}.
\]
\end{lem}
\begin{proof}
By \citet[Supplementary Material, Lemma D.8]{bykh_JBES} and \citet[Lemma 1]{cav_bounded},
\[
x_{t}=[x_{t-1}+v_{t}]_{+}=x_{0}+\sum\limits _{s=1}^{t}v_{s}+\sup\limits _{t'\in\{0,\ldots,t\}}\left[-x_{0}-\sum\limits _{s=1}^{t'}v_{s}\right]_{+}.
\]
Defining $V_{T}(r)\defeq T^{-1/2}(x_{0}+\sum_{s=1}^{\smlfloor{rT}}v_{s})$,
we have that $V_{T}(r)\indist V(r)$ on $D[0,1]$ by the hypotheses
of the lemma. Since
\[
T^{-1/2}x_{\smlfloor{rT}}=V_{T}(r)+\sup_{r^{\prime}\in[0,r]}[-V_{T}(r^{\prime})]_{+}
\]
and the supremum on the r.h.s.\ is a continuous functional of $V_{T}(\cdot)$,
the result follows by the continuous mapping theorem (CMT).
\end{proof}
\begin{proof}[Proof of Theorem~\ref{thm:LUR_y}]
Multiplying \eqref{eq:dgp} by $\beta^{-t}$, and defining $x_{t}\defeq\beta^{-t}y_{t}$,
we have
\[
x_{t}=\beta^{-t}y_{t}=\beta^{-t}[\beta y_{t-1}+\alpha+u_{t}]_{+}=[\beta^{-(t-1)}y_{t-1}+\beta^{-t}(\alpha+u_{t})]_{+}=[x_{t-1}+v_{t}]_{+},
\]
where $v_{t}\defeq\beta^{-t}(\alpha+u_{t})$. Under \assref{INIT},
$T^{-1/2}x_{0}=T^{-1/2}y_{0}\goesto b_{0}$, and
so
\begin{align}
\frac{1}{T^{1/2}}\left(x_{0}+\sum_{s=1}^{\smlfloor{rT}}v_{s}\right) & =\frac{x_{0}}{T^{1/2}}+\frac{a}{T}\sum_{s=1}^{\smlfloor{rT}}e^{-cs/T}+\frac{1}{T^{1/2}}\sum_{s=1}^{\smlfloor{rT}}e^{-cs/T}u_{s}\nonumber \\
 & \indist b_{0}+a\int_{0}^{r}e^{-cs}\diff s+\sigma\int_{0}^{r}e^{-cs}\diff W(s)=K_{\theta}(r)\label{eq:Kcvg}
\end{align}
on $D[0,1]$, where $K_{\theta}$ is as defined in \eqref{eq:Kac},
$\theta=(a,b_{0},c)$, and the weak convergence follows by the martingale
central limit theorem. It follows by Lemma \ref{lem:wkccens} that
on $D[0,1]$,
\[
T^{-1/2}y_{\smlfloor{rT}}=\beta^{\smlfloor{rT}}T^{-1/2}x_{\smlfloor{rT}}=e^{c\smlfloor{rT}/T}T^{-1/2}x_{\smlfloor{rT}}\indist e^{cr}\left[K_{\theta}(r)+\sup_{r^{\prime}\leq r}[-K_{\theta}(r^{\prime})]\right].\qedhere
\]
\end{proof}

\subsection{OLS asymptotics}

\label{app:OLSar1}

Recall from \eqref{eq:ytminus} that, in the AR(1) model, $y_{t}^{-}=[\alpha+\beta y_{t-1}+u_{t}]_{-}$,
and that in this case \eqref{eq:arkwithytminus} specialises to
\begin{equation}
y_{t}=[\alpha+\beta y_{t-1}+u_{t}]_{+}=\alpha+\beta y_{t-1}+u_{t}-y_{t}^{-}.\label{eq:uncens}
\end{equation}

\begin{lem}
\label{lem:AR1aux} Suppose Assumptions \ref{ass:INIT} and \ref{ass:DGP}
hold with $k=1$. Then
\begin{enumerate}
\item \label{enu:AR1aux:wklim}$T^{-1/2}\sum_{t=1}^{T}(u_{t}-y_{t}^{-})\indist J_{\theta}(1)-b_{0}-c\int_{0}^{1}J_{\theta}(r)\diff r-a$;
\item \label{enu:AR1aux:oT}$\sum_{t=1}^{T}(y_{t}^{-})^{2}=o_{p}(T)$; and
\item \label{enu:AR1aux:diffsq}$T^{-1}\sum_{t=1}^{T}(\Delta y_{t})^{2}\inprob\sigma^{2}$.
\end{enumerate}
\end{lem}
\begin{proof}
\textbf{\enuref{AR1aux:wklim}.} We first note from \eqref{eq:uncens}
that
\begin{align*}
\sum_{t=1}^{T}(u_{t}-y_{t}^{-})=\sum_{t=1}^{T}(y_{t}-\alpha-\beta y_{t-1}) & =\sum_{t=1}^{T}(y_{t}-y_{t-1})-(\beta-1)\sum_{t=1}^{T}y_{t-1}-T\alpha\\
 & =y_{T}-y_{0}-(\beta-1)\sum_{t=1}^{T}y_{t-1}-T^{1/2}a.
\end{align*}
Since $T(\beta-1)=c+o(1)$, it follows from Theorem \ref{thm:LUR_y}
and the CMT that, under \assref{INIT},
\begin{align*}
\frac{1}{T^{1/2}}\sum_{t=1}^{T}(u_{t}-y_{t}^{-}) & =\frac{1}{T^{1/2}}(y_{T}-y_{0})-\frac{c+o(1)}{T^{3/2}}\sum_{t=1}^{T}y_{t-1}-a\\
 & \indist J_{\theta}(1)-b_{0}-c\int_{0}^{1}J_{\theta}(r)\diff r-a.
\end{align*}

\textbf{\enuref{AR1aux:oT}.} Since $\beta\geq0$ and $y_{t-1}\geq0$,
\[
0\geq y_{t}^{-}=(\alpha+\beta y_{t-1}+u_{t})\indic\{\alpha+\beta y_{t-1}+u_{t}\leq0\}\geq v_{t}\indic\{v_{t}\leq-\beta y_{t-1}\},
\]
where $v_{t}\defeq\alpha+u_{t}$. Hence
\[
\max_{1\leq t\leq T}\smlabs{y_{t}^{-}}\leq\max_{1\leq t\leq T}\smlabs{v_{t}}\leq\frac{\smlabs a}{T^{1/2}}+\max_{1\leq t\leq T}\smlabs{u_{t}}=o_{p}(T^{1/2})
\]
where the final equality holds since $\{u_{t}\}$ is i.i.d.\ with
finite variance, under Assumption \assref{DGP}\enuref{DGP:ut}.
Further, by the result of part~\enuref{AR1aux:wklim},
\[
\sum_{t=1}^{T}y_{t}^{-}=\sum_{t=1}^{T}u_{t}+O_{p}(T^{1/2})=O_{p}(T^{1/2}).
\]
Hence
\[
\sum_{t=1}^{T}(y_{t}^{-})^{2}\leq\max_{1\leq t\leq T}\smlabs{y_{t}^{-}}\sum_{t=1}^{T}\smlabs{y_{t}^{-}}=-\max_{1\leq t\leq T}|y_{t}^{-}|\sum_{t=1}^{T}y_{t}^{-}=o_{p}(T^{1/2})O_{p}(T^{1/2})=o_{p}(T).
\]

\textbf{\enuref{AR1aux:diffsq}.} Since
\begin{align*}
\sum_{t=1}^{T}\alpha^{2} & =a^{2}=O(1), & \sum_{t=1}^{T}\alpha(\beta-1)y_{t-1} & =O_{p}(1), & \sum_{t=1}^{T}[(\beta-1)y_{t-1}]^{2} & =O_{p}(1)
\end{align*}
by Theorem \ref{thm:LUR_y} and the CMT;
\begin{align*}
\sum_{t=1}^{T}\alpha(u_{t}-y_{t}^{-}) & =O_{p}(1),
\end{align*}
by Lemma \ref{lem:AR1aux}\enuref{AR1aux:wklim}; and
\begin{equation}
\frac{1}{T}\sum_{t=1}^{T}(u_{t}-y_{t}^{-})^{2}=\frac{1}{T}\sum_{t=1}^{T}u_{t}^{2}-\frac{2}{T}\sum_{t=1}^{T}y_{t}^{-}u_{t}+\frac{1}{T}\sum_{t=1}^{T}(y_{t}^{-})^{2}\inprob\sigma^{2},\label{eq:sigma2}
\end{equation}
by the law of large numbers, the Cauchy--Schwarz (CS) inequality, and the result of
part~\enuref{AR1aux:oT}; and
\begin{align*}
\left|\sum_{t=1}^{T}(\beta-1)y_{t-1}(u_{t}-y_{t}^{-})\right|\leq\sqrt{\sum_{t=1}^{T}[(\beta-1)y_{t-1}]^{2}\sum_{t=1}^{T}(u_{t}-y_{t}^{-})^{2}} & =O_{p}(\sqrt{T}),
\end{align*}
by the preceding; it follows that
\begin{align*}
\frac{1}{T}\sum_{t=1}^{T}(\Delta y_{t})^{2} & =\frac{1}{T}\sum_{t=1}^{T}(\alpha+(\beta-1)y_{t-1}+u_{t}-y_{t}^{-})^{2}=\frac{1}{T}\sum_{t=1}^{T}(u_{t}-y_{t}^{-})^{2}+o_{p}(1)\inprob\sigma^{2}.\qedhere
\end{align*}
\end{proof}
\begin{proof}[Proof of Theorem \ref{thm:OLS_LUR}]
We have from \eqref{eq:uncens} that
\begin{equation}
\begin{bmatrix}\hat{\alpha}_{T}-\alpha\\
\hat{\beta}_{T}-\beta
\end{bmatrix}=\left(\sum_{t=1}^{T}\begin{bmatrix}1 & y_{t-1}\\
y_{t-1} & y_{t-1}^{2}
\end{bmatrix}\right)^{-1}\sum_{t=1}^{T}\begin{bmatrix}1\\
y_{t-1}
\end{bmatrix}(u_{t}-y_{t}^{-})\label{eq:olsdecomp}
\end{equation}
where
\begin{align}
\frac{1}{T^{1/2}}\sum_{t=1}^{T}(u_{t}-y_{t}^{-}) & \indist J_{\theta}(1)-b_{0}-c\int_{0}^{1}J_{\theta}(r)\diff r-a\label{eq:ar1num1}
\end{align}
by Lemma \ref{lem:AR1aux}\enuref{AR1aux:wklim}. To obtain the
weak limit of $\sum_{t=1}^{T}y_{t-1}(u_{t}-y_{t}^{-})$, we note that
since only one of $y_{t}$ and $y_{t}^{-}$ can be nonzero, $y_{t}y_{t}^{-}=0$,
and hence $y_{t}^{-}y_{t-1}=-y_{t}^{-}\Delta y_{t}$. Thus by the
CS inequality and Lemma \ref{lem:AR1aux}\enuref{AR1aux:oT}--\enuref{AR1aux:diffsq},
\[
\abs{\sum_{t=1}^{T}y_{t}^{-}y_{t-1}}=\abs{\sum_{t=1}^{T}y_{t}^{-}\Delta y_{t}}\leq\left[\sum_{t=1}^{T}(y_{t}^{-})^{2}\sum_{t=1}^{T}(\Delta y_{t})^{2}\right]^{1/2}=o_{p}(T).
\]
It follows by the preceding and \citet[Theorem 2.1]{ito_convergence}
that
\begin{equation}
\frac{1}{T}\sum_{t=1}^{T}y_{t-1}(u_{t}-y_{t}^{-})=\frac{1}{T}\sum_{t=1}^{T}y_{t-1}u_{t}+o_{p}(1)\indist\sigma\int_{0}^{1}J_{\theta}(r)\diff W(r).\label{eq:ar1num2}
\end{equation}
In view of \eqref{eq:olsdecomp}--\eqref{eq:ar1num2}, a final appeal
to Theorem \ref{thm:LUR_y} and the CMT yields
\begin{align*}
\begin{bmatrix}T^{1/2}(\hat{\alpha}_{T}-\alpha)\\
T(\hat{\beta}_{T}-\beta)
\end{bmatrix} & =\left(\sum_{t=1}^{T}\begin{bmatrix}T^{-1} & T^{-3/2}y_{t-1}\\
T^{-3/2}y_{t-1} & T^{-2}y_{t-1}^{2}
\end{bmatrix}\right)^{-1}\sum_{t=1}^{T}\begin{bmatrix}T^{-1/2}\\
T^{-1}y_{t-1}
\end{bmatrix}(u_{t}-y_{t}^{-})\\
 & \indist\begin{bmatrix}1 & \int J_{\theta}(r)\diff r\\
\int J_{\theta}(r)\diff r & \int J_{\theta}^{2}(r)\diff r
\end{bmatrix}^{-1}\begin{bmatrix}J_{\theta}(1)-b_{0}-c\int J_{\theta}(r)\diff r-a\\
\sigma\int J_{\theta}(r)\diff W(r)
\end{bmatrix}.\qedhere
\end{align*}
\end{proof}
\begin{proof}[Verification of \eqref{eq:altexpression}]
\label{proof_altexpression} 
By the Frisch-Waugh-Lovell theorem, and using partial summation (as in the
proof of Theorem 3.1 in \citet{Phil87Ecta}) we have
\begin{equation}
\begin{split}\label{eq:proof_{a}ltexpression}\hat{\beta}_{T}-1 & =\frac{\sum_{t=1}^{T}y_{t-1}^{\mu}y_{t}^{\mu}}{\sum_{t=1}^{T}(y_{t-1}^{\mu})^{2}}-1=\frac{\sum_{t=1}^{T}y_{t-1}^{\mu}\Delta y_{t}^{\mu}}{\sum_{t=1}^{T}(y_{t-1}^{\mu})^{2}}\\
 & =\frac{\sum_{t=1}^{T}\left((y_{t-1}^{\mu}+\Delta y_{t}^{\mu})^{2}-(y_{t-1}^{\mu})^{2}-(\Delta y_{t}^{\mu})^{2}\right)}{2\sum_{t=1}^{T}(y_{t-1}^{\mu})^{2}}=\frac{(y_{T}^{\mu})^{2}-(y_{0}^{\mu})^{2}-\sum_{t=1}^{T}(\Delta y_{t}^{\mu})^{2}}{2\sum_{t=1}^{T}(y_{t-1}^{\mu})^{2}}
\end{split}
\end{equation}
where $y_{t}^{\mu}\defeq y_{t}-\frac{1}{T}\sum_{t=1}^{T}y_{t-1}$.
\eqref{eq:altexpression} then follows (noting the centring here is
around $1$, rather than $\beta$) by Theorem \ref{thm:LUR_y}, Lemma
\ref{lem:AR1aux}\ref{enu:AR1aux:diffsq}, and the CMT.
\end{proof}
\begin{proof}[Proof of Corollary \ref{thm:tstatARk} ($k=1$)]
Once we have shown that $\hat{\sigma}_{T}^{2}\inprob\sigma^{2}$,
the limiting distributions of the $t$ statistics will follow from
Theorem~\ref{thm:LUR_y} and the CMT. Noting from \eqref{eq:uncens}
that
\[
\hat{u}_{t}=y_{t}-\hat{\alpha}_{T}-\hat{\beta}_{T}y_{t-1}=(\alpha-\hat{\alpha}_{T})+(\beta-\hat{\beta}_{T})y_{t-1}+(u_{t}-y_{t}^{-}),
\]
and that $\sum_{t=1}^{T}\hat{u}_{t}=0$ and $\sum_{t=1}^{T}y_{t-1}\hat{u}_{t}=0$,
we have
\begin{align*}
\sum_{t=1}^{T}[\hat{u}_{t}^{2}-(u_{t}-y_{t}^{-})^{2}] & =\sum_{t=1}^{T}[\hat{u}_{t}-(u_{t}-y_{t}^{-})][\hat{u}_{t}+(u_{t}-y_{t}^{-})]\\
 & =\sum_{t=1}^{T}[(\alpha-\hat{\alpha}_{T})+(\beta-\hat{\beta}_{T})y_{t-1}][\hat{u}_{t}+(u_{t}-y_{t}^{-})]\\
 & =(\alpha-\hat{\alpha}_{T})\sum_{t=1}^{T}(u_{t}-y_{t}^{-})+(\beta-\hat{\beta}_{T})\sum_{t=1}^{T}y_{t-1}(u_{t}-y_{t}^{-})\\
 & =O_{p}(T^{-1/2})O_{p}(T^{1/2})+O_{P}(T^{-1})O_{p}(T)=O_{p}(1)
\end{align*}
where the orders of $\sum_{t=1}^{T}(u_{t}-y_{t}^{-})$ and $\sum_{t=1}^{T}y_{t-1}(u_{t}-y_{t}^{-})$
follow from \eqref{eq:ar1num1} and \eqref{eq:ar1num2} above, and
the rates of convergence of $\hat{\alpha}_{T}$ and $\hat{\beta}_{T}$
from Theorem~\ref{thm:OLS_LUR}. Hence, by \eqref{eq:sigma2},
\begin{align*}
\hat{\sigma}_{T}^{2}=\frac{1}{T}\sum_{t=1}^{T}\hat{u}_{t}^{2} & =\frac{1}{T}\sum_{t=1}^{T}(u_{t}-y_{t}^{-})^{2}+O_{p}(T^{-1})\inprob\sigma^{2}.\qedhere
\end{align*}
\end{proof}

\section{Proofs of results for general $k$}

\label{app:proofgeneral}

\subsection{Limiting distribution of $T^{-1/2}y_{\protect\smlfloor{nr}}$: AR($k$)
case}

Let $\rho$ denote the \emph{inverse} of the root of $B(z)$ closest
to real unity, which for $T$ sufficiently large must be real because \ref{ass:DGP}
permits $B(z)$ to have only one root local to unity. Thus $B(z)$ factorises as
\begin{equation}
B(z)=(1-\beta)z+\phi(z)(1-z)=\psi(z)(1-\rho z)\label{eq:psidef}
\end{equation}
for $z\in\complex$ where $\psi(z)=1-\sum_{i=1}^{k-1}\psi_{i}z^{i}$.
Under \ref{ass:DGP}\enuref{DGP:LU}, $\beta=\beta_{T}\goesto1$
and, thus, $\rho=\rho_{T}\goesto1$, from which it follows that $\psi_{i} \goesto \phi_{i}$
for $i\in\{1,\ldots,k-1\}$, as $T\goesto \infty$.\footnote{Formally, one should write $\psi_{i,T}$, since the coefficients
of the polynomial $\psi(z)$ depend on $\beta$, which in turn varies
with $T$. We omit this dependence on $T$ for ease of notation.} Thus for $T$ sufficiently large, $\rho$ is real and positive (as
we shall maintain throughout the following), and such a condition
as \assref{JSR} also holds when each $\phi_{i}$ in \eqref{eq:Fdef}
is replaced by $\psi_{i}$. Moreover, taking $z=\rho_{T}^{-1}$ in the
preceding, it follows that
\begin{equation}
\begin{split}\label{eq:rholim}0 & =B(\rho_{T}^{-1})=(1-\beta_{T})\rho_{T}^{-1}+\phi(\rho_{T}^{-1})(1-\rho_{T}^{-1})\\
 & \Leftrightarrow T(\rho_{T}-1)=\phi(\rho_{T}^{-1})^{-1}T(\beta_{T}-1)\goesto\phi(1)^{-1}c=:c_{\phi}.
\end{split}
\end{equation}
The factorisation \eqref{eq:psidef} also permits us to rewrite the
model \eqref{eq:ARklagpoly} for $\{y_{t}\}$ in terms of the quasi-differences
$\Delta_{\rho}y_{t}$,
\begin{equation}
\psi(L)\Delta_{\rho}y_{t}=\alpha+u_{t}-y_{t}^{-},\label{eq:psiDy}
\end{equation}
where $\Delta_{\rho}\defeq1-\rho L$. With the aid of this representation
we establish the following preliminary lemmas.
\begin{lem}
\label{lem:censrep} Suppose Assumption \ref{ass:DGP} holds, and
define
\begin{align}
x_{t} & \defeq\psi(\rho^{-1})\rho^{-t}y_{t} & \xi_{t}\defeq & \sum_{s=1}^{t}\rho^{-s}(\alpha+u_{s})-\gamma_{\rho}(L)[\rho^{-t}\Delta_{\rho}y_{t}-\Delta_{\rho}y_{0}]\label{eq:x-xidef}
\end{align}
for $t\in\{1,\ldots,T\}$, where $\gamma_{\rho}(L)$ is the $k-2$
order polynomial such that $\psi(L)={\psi(\rho^{-1})+\gamma_{\rho}(L)\Delta_{\rho}}$
(with $\psi(L)=1$ and $\gamma_{\rho}(L)=0$ when $k=1$). Then for
$t\in\{1,\ldots,T\}$,
\begin{equation}
x_{t}=[x_{t-1}+\Delta\xi_{t}]_{+}.\label{eq:xkcens}
\end{equation}
\end{lem}
\begin{proof}
Observe that for any series $\{\eta_{s}\}$,
\begin{equation}
\sum_{s=1}^{t}\rho^{t-s}\Delta_{\rho}\eta_{s}=\sum_{s=1}^{t}\rho^{t-s}(\eta_{s}-\rho\eta_{s-1})=\sum_{s=1}^{t}\rho^{t-s}\eta_{s}-\sum_{s=1}^{t}\rho^{t-(s-1)}\eta_{s-1}=\eta_{t}-\rho^{t}\eta_{0}.\label{eq:unravel}
\end{equation}
Applying this to $\psi(L)\Delta_{\rho}y_{s}=\alpha+u_{s}-y_{s}^{-}$
(from \eqref{eq:psiDy} above), we obtain
\[
\psi(L)y_{t}-\rho^{t}\psi(L)y_{0}=\sum_{s=1}^{t}\rho^{t-s}(\alpha+u_{s})-\sum_{s=1}^{t}\rho^{t-s}y_{s}^{-}.
\]
Using $\psi(L)=\psi(\rho^{-1})+\gamma_{\rho}(L)\Delta_{\rho}$, and
rearranging yields
\[
\psi(\rho^{-1})y_{t}+\sum_{s=1}^{t}\rho^{t-s}y_{s}^{-}=\sum_{s=1}^{t}\rho^{t-s}(\alpha+u_{s})-\gamma_{\rho}(L)(\Delta_{\rho}y_{t}-\rho^{t}\Delta_{\rho}y_{0})+\rho^{t}\psi(\rho^{-1})y_{0}.
\]
Finally, multiplying by $\rho^{-t}$ and recalling the definitions
of $(x_{t},\xi_{t})$ from \eqref{eq:x-xidef}, we have
\begin{equation}
x_{t}+\sum_{s=1}^{t}\rho^{-s}y_{s}^{-}=x_{0}+\xi_{t}.\label{eq:x-xi-int}
\end{equation}

To proceed from \eqref{eq:x-xi-int} to show that $(x_{t},\xi_{t})$
satisfy \eqref{eq:xkcens}, we note that for $t\geq1$,
\begin{align*}
\xi_{t+1}=x_{t+1}-x_{0}+\sum_{s=1}^{t+1}\rho^{-s}y_{s}^{-} & =x_{t+1}-x_{0}+\rho^{-(t+1)}y_{t+1}^{-}+\sum_{s=1}^{t}\rho^{-s}y_{s}^{-}\\
 & =x_{t+1}+\rho^{-(t+1)}y_{t+1}^{-}+\xi_{t}-x_{t}
\end{align*}
where the first and last equalities follow from \eqref{eq:x-xi-int}.
Hence
\[
x_{t+1}+\rho^{-(t+1)}y_{t+1}^{-}=x_{t}+\Delta\xi_{t+1}.
\]
From \eqref{eq:tobitark} and \eqref{eq:ytminus}, at most one of
$y_{t+1}$ and $y_{t+1}^{-}$ can be nonzero, and must have opposite
signs. Since $\psi(\rho^{-1}) \goesto \phi(1) > 0$ (due to stationarity), we must have
$\psi(\rho^{-1})\rho^{-t}>0$ for all $T$ sufficiently
large. The same must also be true for $x_{t+1}=\psi(\rho^{-1})\rho^{-(t+1)}y_{t+1}$
and $\rho^{-(t+1)}y_{t+1}^{-}$. Hence,
\[
x_{t+1}=[x_{t}+\Delta\xi_{t+1}]_{+}
\]
for $t\geq1$. Plugging $\xi_{0}=0$ into \eqref{eq:x-xi-int} when $t=1$,
we have
\[
x_{1}+\rho^{-1}y_{1}^{-}=x_{0}+\xi_{1}=x_{0}+\Delta\xi_{1}
\]
and thus $x_{1}=[x_{0}+\Delta\xi_{1}]_{+}$, by the same argument.
\end{proof}
\begin{lem}
\label{lem:qdnegl} Suppose Assumptions \ref{ass:INIT}--\ref{ass:JSR} hold. Then there exists a $C<\infty$
such that
\[
\max_{-k+2\leq t\leq T}(\smlnorm{\Delta_{\rho}y_{t}}_{2+\delta_{u}}+\smlnorm{\Delta y_{t}}_{2+\delta_{u}})<C.
\]
\end{lem}
\begin{proof}
We have from \eqref{eq:psiDy} that
\begin{equation}
\Delta_{\rho}y_{t}+y_{t}^{-}=\sum_{i=1}^{k-1}\psi_{i}\Delta_{\rho}y_{t-i}+\alpha+u_{t}\eqdef w_{t},\label{eq:wdef}
\end{equation}
for $t\in\{1,\ldots,T\}$. Since, as noted in the proof of Lemma \ref{lem:censrep},
only one of $y_{t}$ and $y_{t}^{-}$ can be nonzero, and have opposite
signs,
\[
\Delta_{\rho}y_{t}>0\implies y_{t}>\rho y_{t-1}\geq0\implies y_{t}^{-}=0.
\]
It follows that either $w_{t}>0$, in which case $\Delta_{\rho}y_{t}=w_{t}$;
or $w_{t}\leq0$, in which case $\Delta_{\rho}y_{t}\in[w_{t},0]$.
Hence there exists a $\delta_{t}\in[0,1]$ such that $\Delta_{\rho}y_{t}=\delta_{t}w_{t}$
for $t\in\{1,\ldots,T\}$. Taking $w_{0}=\Delta_{\rho}y_{0}$ and
$\delta_{0}=1$, \eqref{eq:wdef} is equivalent to a dynamical system
defined by
\begin{align}
w_{t} & =\psi_{1}\delta_{t-1}w_{t-1}+\sum_{i=2}^{k-1}\psi_{i}\Delta_{\rho}y_{t-i}+\alpha+u_{t},\label{eq:w1}\\
\Delta_{\rho}y_{t-1} & =\delta_{t-1}w_{t-1}\label{eq:w2}
\end{align}
for $t\in\{1,\ldots,T\}$, for an appropriate sequence $\{\delta_{t}\}\subset[0,1]$.
Defining
\begin{align*}
\vec w_{t} & \defeq\begin{bmatrix}w_{t}\\
\Delta_{\rho}y_{t-1}\\
\vdots\\
\Delta_{\rho}y_{t-k+2}
\end{bmatrix} & \vec v_{t} & \defeq\begin{bmatrix}\alpha+u_{t}\\
0\\
\vdots\\
0
\end{bmatrix} & F_{\delta}(\vec{\psi}) & \defeq\begin{bmatrix}\psi_{1}\delta & \psi_{2} & \cdots & \psi_{k-2} & \psi_{k-1}\\
\delta & 0 & \cdots & 0 & 0\\
0 & 1\\
 &  & \ddots\\
 &  &  & 1 & 0
\end{bmatrix}
\end{align*}
where $\vec{\psi}\defeq(\psi_{1},\ldots,\psi_{k-1})$, we can write
the companion form of \eqref{eq:w1}--\eqref{eq:w2} as
\[
\vec w_{t}=F_{\delta_{t-1}}(\vec{\psi})\vec w_{t-1}+\vec v_{t}
\]
for $t\in\{1,\ldots,T\}$, with the initialisation $\vec w_{0}\defeq(\Delta_{\rho}y_{0},\Delta_{\rho}y_{-1},\ldots,\Delta_{\rho}y_{-k+2})^{\trans}$.

Since $\vec{\psi} \goesto \vec{\phi}=(\phi_{1},\ldots,\phi_{k-1})^{\trans}$
as $T\goesto\infty$, $F_{\delta}(\vec{\psi}) \goesto F_{\delta}(\vec{\phi})=F_{\delta}$,
where $F_{\delta}$ is as defined in \eqref{eq:Fdef}. By Proposition~1.8
in \citet{Jungers09} and the continuity of the JSR,
\[
\lambda_{\jsr}(\{F_{\delta}(\vec{\psi})\mid\delta\in[0,1]\})=\lambda_{\jsr}(\{F_{0}(\vec{\psi}),F_{1}(\vec{\psi})\})\goesto\lambda_{\jsr}(\{F_{0},F_{1}\})<1
\]
under \assref{JSR}. It follows that there exists a norm $\smlnorm{\cdot}_{\ast}$
and a $\gamma\in[0,1)$ such that (for all $T$ sufficiently large),
\[
\smlnorm{\vec w_{t}}_{\ast}\leq\smlnorm{F_{\delta_{t-1}}}_{\ast}\smlnorm{\vec w_{t-1}}_{\ast}+\smlnorm{\vec v_{t}}_{\ast}\leq\gamma\smlnorm{\vec w_{t-1}}_{\ast}+\smlnorm{\vec v_{t}}_{\ast}
\]
whence by backward substitution,
\begin{align*}
\smlnorm{\vec w_{t}}_{\ast} & \leq\sum_{s=0}^{t-1}\gamma^{s}\smlnorm{\vec v_{t-s}}_{\ast}+\gamma^{t}\smlnorm{\vec w_{0}}_{\ast}.
\end{align*}
By the equivalence of norms on finite-dimensional spaces, there exists
a $C<\infty$ such that
\[
\smlabs{\Delta_{\rho}y_{t-1}}\leq C\left[\sum_{s=0}^{t-1}\gamma^{s}(\smlabs{\alpha}+\smlabs{u_{t-s}})+\gamma^{t}\sum_{i=-k+2}^{0}\smlabs{\Delta_{\rho}y_{i}}\right].
\]
Deduce that for any $p\geq1$,
\begin{align}
\smlnorm{\Delta_{\rho}y_{t-1}}_{p} & \leq\frac{C\smlabs{\alpha}}{1-\gamma}+C\sum_{s=0}^{t-1}\gamma^{s}\smlnorm{u_{t-s}}_{p}+C\sum_{i=-k+2}^{0}\smlnorm{\Delta_{\rho}y_{i}}_{p}\nonumber \\
 & \leq\frac{C\smlabs{\alpha}}{1-\gamma}+\frac{C}{1-\gamma}\max_{1\leq s\leq t}\smlnorm{u_{s}}_{p}+C\sum_{i=-k+2}^{0}\smlnorm{\Delta_{\rho}y_{i}}_{p}.\label{eq:delrhobnd}
\end{align}
Now for each $i\in\{-k+2,\ldots,0\}$, $\Delta_{\rho}y_{i}=(1-\rho)[y_{0}-\sum_{j=i+1}^{0}\Delta y_{j}]+\rho\Delta y_{i}$,
and hence there exists a $C^{\prime}<\infty$ such that
\[
\smlnorm{\Delta_{\rho}y_{i}}_{2+\delta_{u}}\leq C^{\prime}\left[\frac{1}{T^{1/2}}\smlnorm{T^{-1/2}y_{0}}_{2+\delta_{u}}+\frac{1}{T}\sum_{j=i+1}^{0}\smlnorm{\Delta y_{j}}_{2+\delta_{u}}+\smlnorm{\Delta y_{i}}_{2+\delta_{u}}\right],
\]
where we have used that $1-\rho=O(T^{-1})$.
Taking $p=2+\delta_{u}$ in \eqref{eq:delrhobnd}, it follows under
\assref{MOM} that $\max_{-k+2\leq t\leq T}\smlnorm{\Delta_{\rho}y_{t}}_{2+\delta_{u}}$
is bounded uniformly in $T$. To obtain the corresponding result for
$\Delta y_{t}$, we use \eqref{eq:unravel} with $\eta_t=y_t$ to write
\[
\Delta y_{t}=\Delta_{\rho}y_{t}+(\rho-1)y_{t-1}=\Delta_{\rho}y_t+(\rho-1)\left[\sum_{s=1}^{t-1}\rho^{t-1-s}\Delta_{\rho}y_{s}+\rho^{t-1}y_{0}\right]
\]
whence there exists a $C^{\prime\prime}<\infty$ such that
\[
\smlnorm{\Delta y_{t}}_{2+\delta_{u}}\leq\smlnorm{\Delta_{\rho}y_{t}}_{2+\delta_{u}}+C^{\prime\prime}\left[T^{-1}\max_{1\leq s\leq t}\smlnorm{\Delta_{\rho}y_{s}}_{2+\delta_{u}}+T^{-1/2}\smlnorm{T^{-1/2}y_{0}}_{2+\delta_{u}}\right],
\]
from which, under \assref{MOM}, the result follows.
\end{proof}
\begin{proof}[Proof of Theorem \ref{thm:ytARk}]
When $k=1$, the result follows by Theorem \ref{thm:LUR_y}; we therefore
suppose $k\geq2$. We first note that by \eqref{eq:rholim} above,
$\rho^{\smlfloor{rT}}\goesto e^{c_{\phi}r}$ uniformly in $r\in[0,1]$.
Hence for $(x_{t},\xi_{t})$ as in Lemma \ref{lem:censrep},

\begin{align}
T^{-1/2}\left[x_{0}+\sum_{s=1}^{\smlfloor{rT}}\Delta\xi_{s}\right] & =T^{-1/2}x_{0}+T^{-1/2}\xi_{\smlfloor{rT}}\nonumber \\
 & =_{(1)}\psi(\rho^{-1})T^{-1/2}y_{0}+T^{-1/2}\sum_{s=1}^{\smlfloor{rT}}\rho^{-s}(\alpha+u_{s})+o_{p}(1)\indist_{(2)}K_{\theta_{\phi}}(r),\label{eq:proclim}
\end{align}
on $D[0,1]$, where $\theta_{\phi}=(a,\phi(1)b_{0},\phi(1)^{-1}c)$,
$=_{(1)}$ holds since Lemma \ref{lem:qdnegl} implies that $\max_{0\leq t\leq T}\smlabs{\Delta_{\rho}y_{t}}=o_{p}(T^{1/2})$,
and $\indist_{(2)}$ holds by the same arguments as which yielded \eqref{eq:Kcvg}
above and by recalling that $\psi(\rho^{-1})\goesto\phi(1)$. Hence by Lemma \ref{lem:censrep}, $x_{t}$ and $v_{t}\defeq\Delta\xi_{t}$
satisfy the requirements of Lemma \ref{lem:wkccens}. We thus have
\begin{align*}
T^{-1/2}y_{\smlfloor{rT}} & =\psi(\rho^{-1})^{-1}\rho^{\smlfloor{rT}}T^{-1/2}x_{\smlfloor{rT}}\\
 & \indist\phi(1)^{-1}\e^{c_{\phi}r}\left\{ K_{\theta_{\phi}}(r)+\sup_{r^{\prime}\leq r}[-K_{\theta_{\phi}}(r^{\prime})]\right\} =\phi(1)^{-1}J_{\theta_{\phi}}(r)
\end{align*}
on $D[0,1]$.
\end{proof}

\subsection{OLS asymptotics: AR($k$) case.}

Since, by the implication of Assumption \ref{ass:JSR}, all the roots of $\phi(z)=1-\sum_{i=1}^{k-1}\phi_{i}z^{i}$
lie strictly outside the unit circle, there exists a sequence $\{\varphi_{i}\}_{i=0}^{\infty}$
with $\varphi_{0}=1$ and $\sum_{0=1}^{\infty}\smlabs{\varphi_{i}}<\infty$
such that $\phi^{-1}(z)\defeq\sum_{i=0}^{\infty}\varphi_{i}z^{i}$
satisfies $\phi^{-1}(z)\phi(z)=1$ for all $\smlabs z\leq1$. Moreover,
there exists a $C<\infty$ and a $\gamma_{\phi}\in(0,1)$ such that
$\smlabs{\varphi_{i}}<C\gamma_{\phi}^{i}$ for all $i\geq0$. (See
e.g.~\citet[Section 3.3]{Brockwell91}.)
\begin{lem}
\label{lem:partialinverse}Let $\phi_{m}^{-1}(z)\defeq\sum_{i=0}^{m}\varphi_{i}z^{i}$,
where $m\geq1$. Then there exists a $C<\infty$, independent of $m$,
and $\{d_{m,i}\}_{i=1}^{k-1}$ such that
\[
\phi_{m}^{-1}(z)\phi(z)=1-z^{m}\sum_{i=1}^{k-1}d_{m,i}z^{i}\eqdef1-d_{m}(z)z^{m}
\]
for all $\smlabs z\leq1$ and $m\in\naturals$, and where $\smlabs{d_{m,i}}\leq C\gamma_{\phi}^{m}$.
\end{lem}
\begin{proof}
Since
\[
1=\phi^{-1}(z)\phi(z)=\phi(z)\left[\phi_{m}^{-1}(z)+\sum_{i=m+1}^{\infty}\varphi_{i}z^{i}\right]
\]
for $\smlabs z\leq1$, it follows that
\[
1-\phi(z)\sum_{i=m+1}^{\infty}\varphi_{i}z^{i}=\phi(z)\phi_{m}^{-1}(z)=\left(1-\sum_{i=1}^{k-1}\phi_{i}z^{i}\right)\sum_{i=0}^{m}\varphi_{i}z^{i}\eqdef1-\sum_{i=1}^{m+k-1}\vartheta_{m,i}z^{i},
\]
using that $\varphi_{0}=1$. Matching coefficients on the left and
right hand sides, we obtain $\vartheta_{m,i}=0$ for all $i\in\{1,\ldots,m\}$;
while for $i\in\{m+1,\ldots,m+k-1\}$, $\vartheta_{m,i}=\sum_{j=0}^{i-(m+1)}\phi_{j}\varphi_{i-j}$.
Taking $d_{m,i}\defeq\vartheta_{m,m+i}=\sum_{j=0}^{i-1}\phi_{j}\varphi_{m+i-j}$
for $i\in\{1,\ldots,k-1\}$, and noting that
\[
\smlabs{d_{m,i}}\leq\sum_{j=i}^{k-1}\smlabs{\phi_{j}}\smlabs{\varphi_{m+i-j}}\leq C\gamma_{\phi}^{m}\sum_{j=1}^{k-1}\smlabs{\phi_{j}}
\]
yields the result.
\end{proof}
\begin{lem}
\label{lem:ARkterms}Suppose Assumptions \ref{ass:INIT}--\ref{ass:JSR}
hold. Then for each $s\in\{0,\ldots,k-1\}$,
\begin{enumerate}
\item \label{enu:numerator} $T^{-1/2}\sum_{t=1}^{T}(u_{t}-y_{t}^{-})\indist\phi(1)[Y_{\theta_{\phi}}(1)-c_{\phi}\int Y_{\theta_{\phi}}(r)-b_{0}]-a$;
\item \label{enu:ytminsq}$\sum_{t=1}^{T}(y_{t}^{-})^{2}=o_{p}(T)$;
\item \label{enu:crossdelta}$T^{-1}\sum_{t=1}^{T}\Delta y_{t}\Delta y_{t-s}\inprob\sigma^{2}\sum_{n=0}^{\infty}\varphi_{n}\varphi_{n+s}$;
and
\item \label{enu:leveldelta}$\sum_{t=1}^{T}y_{t}\Delta y_{t-s}=O_{p}(T)$.
\end{enumerate}
\end{lem}
\begin{proof}
\textbf{\enuref{numerator}.} Applying the factorisation \eqref{eq:psidef}
to \eqref{eq:ARklagpoly}, we get
\begin{equation}
u_{t}-y_{t}^{-}=-\alpha-(\beta-1)y_{t-1}+\phi(L)\Delta y_{t}\label{eq:deltayform}
\end{equation}
and hence, recalling that $\phi(1)=1-\sum_{i=1}^{k-1}\phi_{i}$,
$\alpha=T^{1/2}a$, and $T(\beta-1)\goesto c$ under \ref{ass:DGP},
\begin{align*}
\frac{1}{T^{1/2}}\sum_{t=1}^{T}(u_{t}-y_{t}^{-}) & =-a-\frac{T(\beta-1)}{T^{3/2}}\sum_{t=1}^{T}y_{t-1}+\frac{1}{T^{1/2}}\left(\sum_{t=1}^{T}\Delta y_{t}-\sum_{i=1}^{k-1}\phi_{i}\sum_{t=1}^{T}\Delta y_{t-i}\right)\\
 & =-a-\frac{T(\beta-1)}{T^{3/2}}\sum_{t=1}^{T}y_{t-1}+\frac{y_{T}-y_{0}}{T^{1/2}}-\sum_{i=1}^{k-1}\phi_{i}\frac{y_{T-i}-y_{0-i}}{T^{1/2}}\\
 & \indist-a-c\int_{0}^{1}Y_{\theta_{\phi}}(r)\diff r+\phi(1)[Y_{\theta_{\phi}}(1)-b_{0}]\\
 & =\phi(1)\left[Y_{\theta_{\phi}}(1)-c_{\phi}\int_{0}^{1}Y_{\theta_{\phi}}(r)\diff r-b_{0}\right]-a,
\end{align*}
where convergence holds by Assumption \ref{ass:INIT}, Theorem \ref{thm:ytARk}
and the CMT, recalling that $c_{\phi}=\phi(1)^{-1}c$.

\textbf{\enuref{ytminsq}.} The argument is analogous to that used
to prove Lemma \ref{lem:AR1aux}\enuref{AR1aux:oT}. Rewriting the
factorisation \eqref{eq:psidef} as
\[
\beta z+(1-z)\sum_{i=1}^{k-1}\phi_{i}z^{i}=1-(1-\rho z)\left[1-\sum_{i=1}^{k-1}\psi_{i}z^{i}\right]=\rho z+\sum_{i=1}^{k-1}\psi_{i}z^{i}(1-\rho z)
\]
we have from \eqref{eq:ARklagpoly} that
\[
y_{t}^{-}=\left[\beta y_{t-1}+\sum_{i=1}^{k-1}\phi_{i}\Delta y_{t-i}+\alpha+u_{t}\right]_{-}=\left[\rho y_{t-1}+\sum_{i=1}^{k-1}\psi_{i}\Delta_{\rho}y_{t-i}+\alpha+u_{t}\right]_{-}=[\rho y_{t-1}+v_{t}]_{-},
\]
where $v_{t}\defeq\sum_{i=1}^{k-1}\psi_{i}\Delta_{\rho}y_{t-i}+\alpha+u_{t}$.
Since $\rho y_{t-1}\geq0$, it follows that
\begin{align*}
0\geq y_{t}^{-}=(\rho y_{t-1}+v_{t})\indic\{\rho y_{t-1}+v_{t}\leq0\} & \geq v_{t}\indic\{\rho y_{t-1}+v_{t}\leq0\}.
\end{align*}
By Lemma \ref{lem:qdnegl}, $\smlnorm{\Delta_{\rho}y_{t}}_{2+\delta_{u}}$
is bounded uniformly in $t$. Hence, under \assref{MOM}, so too
is
\begin{align*}
\smlnorm{v_{t}}_{2+\delta_{u}} & \leq\sum_{i=1}^{k-1}\smlabs{\psi_{i}}\smlnorm{\Delta_{\rho}y_{t-i}}_{2+\delta_{u}}+\smlabs{\alpha}+\smlnorm{u_{t}}_{2+\delta_{u}}.
\end{align*}
whence it follows that $\max_{1\leq t\leq T}\smlabs{y_{t}^{-}}\leq\max_{1\leq t\leq T}\smlabs{v_{t}}=o_{p}(T^{1/2})$.
Hence, using the result of part~\enuref{numerator},
\begin{align*}
\sum_{t=1}^{T}(y_{t}^{-})^{2} & \leq\max_{1\leq t\leq T}\smlabs{y_{t}^{-}}\sum_{t=1}^{T}\smlabs{y_{t}^{-}}=-\max_{1\leq t\leq T}\smlabs{y_{t}^{-}}\sum_{t=1}^{T}y_{t}^{-}\\
 & =-\max_{1\leq t\leq T}\smlabs{y_{t}^{-}}\left(\sum_{t=1}^{T}u_{t}+O_{p}(T^{1/2})\right)=o_{p}(T).
\end{align*}

\textbf{\enuref{crossdelta}.} Let $s\in\{0,\ldots,k-1\}$. By Lemma
\ref{lem:partialinverse}, applying $\phi_{t-1}^{-1}(L)=\sum_{i=0}^{t-1}\varphi_{i}L^{i}$
to both sides of \eqref{eq:deltayform} yields
\begin{align*}
\Delta y_{t}-d_{t-1}(L)\Delta y_{1}=\phi_{t-1}^{-1}(L)\phi(L)\Delta y_{t} & =\sum_{i=0}^{t-1}\varphi_{i}L^{i}[(\beta-1)y_{t-1}+\alpha+u_{t}-y_{t}^{-}]
\end{align*}
and hence for $t\in\{1,\ldots,T\}$,
\begin{align}
\Delta y_{t} & =\sum_{i=0}^{t-1}\varphi_{i}(\alpha+u_{t-i})+(\beta-1)\sum_{i=0}^{t-1}\varphi_{i}y_{t-1-i}-\sum_{i=0}^{t-1}\varphi_{i}y_{t-i}^{-}+d_{t-1}(L)\Delta y_{1}\nonumber \\
 & \eqdef w_{t}+r_{1,t}+r_{2,t}+r_{3,t}.\label{eq:delydecomp}
\end{align}
Decompose
\[
w_{t}=\sum_{i=0}^{t-1}\varphi_{i}(\alpha+u_{t-i})=\sum_{i=0}^{\infty}\varphi_{i}u_{t-i}-\sum_{i=t}^{\infty}\varphi_{i}u_{t-i}+\frac{a}{T^{1/2}}\sum_{i=0}^{t-1}\varphi_{i}\eqdef\eta_{t}+r_{4,t}+r_{5,t}.
\]
The sequence $\{\eta_{t}\}$ is a stationary linear process whose coefficients
decay exponentially; hence by \citet[Thm.~3.7 and Rem.~3.9]{PS92AS},
\[
\frac{1}{T}\sum_{t=s+1}^{T}\eta_{t}\eta_{t-s}\inprob\sigma^{2}\sum_{n=0}^{\infty}\varphi_{n}\varphi_{n+s}.
\]
Since for all $i$, $|\varphi_{i}|<C\gamma_{\phi}^{i}$, $\gamma_{\phi}\in(0,1)$,
we have
\begin{align*}
\expect\left[\frac{1}{T}\sum_{t=1}^{T}r_{4,t}^{2}\right] & =\frac{\sigma^{2}}{T}\sum_{t=1}^{T}\sum_{i=t}^{\infty}\varphi_{i}^{2}=O(T^{-1}), & \frac{1}{T}\sum_{t=1}^{T}r_{5,t}^{2} & =\frac{a^{2}}{T^{2}}\sum_{t=1}^{T}\left(\sum_{i=0}^{t-1}\varphi_{i}\right)^{2}=O(T^{-1}).
\end{align*}
It follows by the CS inequality that $\frac{1}{T}\sum_{t=s+1}^{T}(\eta_{t}r_{\ell,t-s}+\eta_{t-s}r_{\ell,t})=o_{p}(1)$
for $\ell\in\{4,5\}$ and $\frac{1}{T}\sum_{t=s+1}^{T}r_{\ell,t}r_{\ell',t-s}=o_{p}(1)$
for $\ell,\ell'\in\{4,5\}$, and hence
\begin{equation}
\frac{1}{T}\sum_{t=s+1}^{T}w_{t}w_{t-s}=\frac{1}{T}\sum_{t=s+1}^{T}\eta_{t}\eta_{t-s}+o_{p}(1)\inprob\sigma^{2}\sum_{n=0}^{\infty}\varphi_{n}\varphi_{n+s}\label{eq:xprodlim}
\end{equation}
for each $s\in\{0,\ldots,k-1\}$. Thus, if we can show that
\begin{align}
\sum_{t=1}^{T}r_{1,t}^{2} & =O_{p}(1) & \sum_{t=1}^{T}r_{2,t}^{2} & =o_{p}(T) & \sum_{t=1}^{T}r_{3,t}^{2} & =O_{p}(1).\label{eq:rems1}
\end{align}
it will follow that
\begin{align*}
\frac{1}{T}\sum_{t=1}^{T}\Delta y_{t}\Delta y_{t-s} & =_{(1)}\frac{1}{T}\sum_{t=s+1}^{T}\Delta y_{t}\Delta y_{t-s}+o_{p}(1)\\
 & =_{(2)}\frac{1}{T}\sum_{t=s+1}^{T}w_{t}w_{t-s}+o_{p}(1)\inprob_{(3)}\sigma^{2}\sum_{n=0}^{\infty}\varphi_{n}\varphi_{n+s}.
\end{align*}
where $=_{(1)}$ holds by Lemma \ref{lem:qdnegl}, $=_{(2)}$ by \eqref{eq:delydecomp}--\eqref{eq:rems1}
and the CS inequality, and $\inprob_{(3)}$ by \eqref{eq:xprodlim}.

It remains to prove \eqref{eq:rems1}. Since $\beta-1=O(T^{-1})$,
there exists a $C<\infty$ such that
\[
\smlabs{r_{1,t}}\leq\frac{C}{T}\sum_{i=0}^{t-1}\smlabs{\varphi_{i}}\smlabs{y_{t-1-i}}\leq\left(\sum_{i=0}^{\infty}\smlabs{\varphi_{i}}\right)\frac{C}{\sqrt{T}}\max_{0\leq t\leq T-1}\left|y_{t}/\sqrt{T}\right|=O_{p}(T^{-1/2})
\]
uniformly in $t$ by Theorem \ref{thm:ytARk}; the first part of \eqref{eq:rems1}
follows immediately. Next,
\begin{align*}
\sum_{t=1}^{T}r_{2,t}^{2} & =\sum_{t=1}^{T}\left(\sum_{i=0}^{t-1}\varphi_{i}y_{t-i}^{-}\right)^{2}=\sum_{t=1}^{T}\sum_{i=0}^{t-1}\sum_{j=0}^{t-1}\varphi_{i}\varphi_{j}y_{t-i}^{-}y_{t-j}^{-}\\
 & =\sum_{i=0}^{T}\sum_{j=0}^{T}\varphi_{i}\varphi_{j}\sum_{t=\max\{i,j\}+1}^{T}y_{t-i}^{-}y_{t-j}^{-}=o_{p}(T)
\end{align*}
by $\sum_{i=0}^{\infty}\smlabs{\varphi_{i}}<\infty$, the result of
part~\enuref{ytminsq}, and the CS inequality. Finally, by Lemma
\ref{lem:partialinverse},
\[
\sum_{t=1}^{T}r_{3,t}^{2}=\sum_{t=1}^{T}\left(\sum_{i=1}^{k-1}d_{t-1,i}\Delta y_{1-i}\right)^{2}\leq C^{2}\left(\sum_{i=1}^{k-1}\smlabs{\Delta y_{1-i}}\right)^{2}\sum_{t=1}^{T}\gamma_{\phi}^{2(t-1)}=O_{p}(1).
\]

\textbf{\enuref{leveldelta}.} We first note that since $ab = \frac{1}{2}((a+b)^2 - a^2 - b^2)$,
\[
\sum_{t=1}^{T}y_{t-1}\Delta y_{t}=\frac{1}{2}\sum_{t=1}^{T}(y_t^2 - y_{t-1}^2 - (\Delta y_t)^2) =\frac{1}{2}\left\{ y_{T}^{2}-y_{0}^{2}-\sum_{t=1}^{T}(\Delta y_{t})^{2}\right\}
\]
which is $O_{p}(T)$ by Theorem \ref{thm:ytARk} and part~\enuref{crossdelta}
of this lemma. This gives the result when $s=1$. To obtain the result
for general $s\in\{0,\ldots,k-1\}$, we simply note that
\[
\sum_{t=1}^{T}y_{t}\Delta y_{t}=\sum_{t=1}^{T}(\Delta y_{t})^{2}+\sum_{t=1}^{T}y_{t-1}\Delta y_{t}
\]
and
\[
\sum_{t=1}^{T}y_{t-s}\Delta y_{t}=\sum_{t=1}^{T}\left(y_{t}-\sum_{r=0}^{s-1}\Delta y_{t-r}\right)\Delta y_{t}=\sum_{t=1}^{T}y_{t}\Delta y_{t}-\sum_{r=0}^{s-1}\sum_{t=1}^{T}\Delta y_{t}\Delta y_{t-r}
\]
when $s\geq1$; the result then follows by a further appeal to part~\enuref{crossdelta}
of this lemma.
\end{proof}
\begin{proof}[Proof of Theorem \ref{thm:olsARk}]
Let $\vec y_{t}\defeq(y_{t},y_{t-1},\ldots,y_{t-k+2})^{\trans}$ and
define
\begin{align*}
\mathcal{M}_{T} & \defeq\sum_{t=1}^{T}\begin{bmatrix}1 & y_{t-1} & \Delta\vec y_{t-1}^{\trans}\\
y_{t-1} & y_{t-1}^{2} & y_{t-1}\Delta\vec y_{t-1}^{\trans}\\
\Delta\vec y_{t-1} & y_{t-1}\Delta\vec y_{t-1} & \Delta\vec y_{t-1}\Delta\vec y_{t-1}^{\trans}
\end{bmatrix} & m_{T} & \defeq\sum_{t=1}^{T}\begin{bmatrix}1\\
y_{t-1}\\
\Delta\vec y_{t-1}
\end{bmatrix}(u_{t}-y_{t}^{-}).
\end{align*}
Then by rewriting \eqref{eq:arkwithytminus} as
\begin{align*}
y_{t} &= \alpha+\beta y_{t-1}+ \vec{\phi}^{\trans} \Delta \vec{y}_{t-1} + u_t - y_{t}^{-},
\end{align*}
the centred and rescaled OLS estimators $\hat{\mu}_{T}^{\trans}\defeq(\hat{\alpha}_{T},\hat{\beta}_{T},\hat{\vec{\phi}}_{T}^{\trans})$
of $\mu^{\trans}=(\alpha,\beta,\vec{\phi}^{\trans})$ are equal to
\begin{align}
\begin{bmatrix}T^{1/2}(\hat{\alpha}_{T}-\alpha)\\
T(\hat{\beta}_{T}-\beta)\\
\hat{\vec{\phi}}_{T}-\vec{\phi}
\end{bmatrix} & =D_{2,T}(\hat{\mu}_{T}-\mu)=(D_{1,T}^{-1}\mathcal{M}_{T}D_{2,T}^{-1})^{-1}D_{1,T}^{-1}m_{T}\label{eq:OLS}
\end{align}
where $D_{1,T}\defeq\diag\{T^{1/2},T,I_{k-1}T\}$, $D_{2,T}\defeq\diag\{T^{1/2},T,I_{k-1}\}$,
\begin{align}
D_{1,T}^{-1}\mathcal{M}_{T}D_{2,T}^{-1} & =\begin{bmatrix}1 & T^{-3/2}\sum_{t=1}^{T}y_{t-1} & T^{-1/2}\sum_{t=1}^{T}\Delta\vec y_{t-1}^{\trans}\\
T^{-3/2}\sum_{t=1}^{T}y_{t-1} & T^{-2}\sum_{t=1}^{T}y_{t-1}^{2} & T^{-1}\sum_{t=1}^{T}y_{t-1}\Delta\vec y_{t-1}^{\trans}\\
T^{-3/2}\sum_{t=1}^{T}\Delta\vec y_{t-1} & T^{-2}\sum_{t=1}^{T}y_{t-1}\Delta\vec y_{t-1} & T^{-1}\sum_{t=1}^{T}\Delta\vec y_{t-1}\Delta\vec y_{t-1}^{\trans}
\end{bmatrix}\label{eq:DMD}
\end{align}
and
\begin{equation}
D_{1,T}^{-1}m_{T}=\begin{bmatrix}T^{-1/2}\sum_{t=1}^{T}(u_{t}-y_{t}^{-})\\
T^{-1}\sum_{t=1}^{T}y_{t-1}(u_{t}-y_{t}^{-})\\
T^{-1}\sum_{t=1}^{T}\Delta\vec y_{t-1}(u_{t}-y_{t}^{-})
\end{bmatrix}.\label{eq:Dm}
\end{equation}

It remains to determine the weak limits of the elements of \eqref{eq:DMD}
and \eqref{eq:Dm}. We consider \eqref{eq:DMD} first. By Theorem
\ref{thm:ytARk} and the CMT,
\begin{equation}
\mathcal{Y}_{T}\defeq\begin{bmatrix}1 & T^{-3/2}\sum_{t=1}^{T}y_{t-1}\\
T^{-3/2}\sum_{t=1}^{T}y_{t-1} & T^{-2}\sum_{t=1}^{T}y_{t-1}^{2}
\end{bmatrix}\indist\begin{bmatrix}1 & \int Y_{\theta_{\phi}}(r)\diff r\\
\int Y_{\theta_{\phi}}(r)\diff r & \int Y_{\theta_{\phi}}^{2}(r)\diff r
\end{bmatrix}\eqdef\mathcal{Y}_{\theta_{\phi}}.\label{eq:YT}
\end{equation}
By Theorem \ref{thm:ytARk} and Lemma \ref{lem:ARkterms}\enuref{leveldelta},
\begin{equation}
\Xi_{T}\defeq\begin{bmatrix}T^{-1/2}\sum_{t=1}^{T}\Delta\vec y_{t-1}^\trans\\
T^{-1}\sum_{t=1}^{T}y_{t-1}\Delta\vec y_{t-1}^{\trans}
\end{bmatrix}=\begin{bmatrix}T^{-1/2}(\vec y_{T-1}-\vec y_{-1})^\trans\\
T^{-1}\sum_{t=1}^{T}y_{t-1}\Delta\vec y_{t-1}^{\trans}
\end{bmatrix}=O_{p}(1).\label{eq:XIT}
\end{equation}
By Lemma \ref{lem:ARkterms}\enuref{crossdelta},
\begin{equation}
\Omega_{T}\defeq T^{-1}\sum_{t=1}^{T}\Delta\vec y_{t-1}\Delta\vec y_{t-1}^{\trans}\inprob\Omega\label{eq:OMEGAT}
\end{equation}
where $\Omega_{ij}\defeq\sigma^{2}\sum_{n=0}^{\infty}\varphi_{n}\varphi_{n+\smlabs{i-j}}$.
It follows by the partitioned matrix inversion formula and the continuity
of matrix inversion that
\begin{equation}
(D_{1,T}^{-1}\mathcal{M}_{T}D_{2,T}^{-1})^{-1}=\begin{bmatrix}\mathcal{Y}_{T} & \Xi_{T}\\
o_{p}(1) & \Omega_{T}
\end{bmatrix}^{-1}=\begin{bmatrix}\mathcal{Y}_{T}^{-1} & -\mathcal{Y}_{T}^{-1}\Xi_{T}\Omega_{T}^{-1}\\
0 & \Omega_{T}^{-1}
\end{bmatrix}+o_{p}(1).\label{eq:MTstd}
\end{equation}

We turn next to \eqref{eq:Dm}. By Lemma \ref{lem:ARkterms}\enuref{numerator},
\begin{align}
\mathcal{Z}_{T}^{(\alpha)}\defeq\frac{1}{T^{1/2}}\sum_{t=1}^{T}(u_{t}-y_{t}^{-}) & \to\phi(1)\left[Y_{\theta_{\phi}}(1)-c_{\phi}\int_{0}^{1}Y_{\theta_{\phi}}(r)\diff r-b_{0}\right]-a\defeq\mathcal{Z}_{\theta_{\phi}}^{(\alpha)}.\label{eq:denom1}
\end{align}
Next, since only one of $y_{t}$ and $y_{t}^{-}$ can be nonzero,
$y_{t-1}y_{t}^{-}=-\Delta y_{t}y_{t}^{-}$, and hence
\begin{align}
\mathcal{Z}_{T}^{(\beta)}\defeq\frac{1}{T}\sum_{t=1}^{T}y_{t-1}(u_{t}-y_{t}^{-}) & =\frac{1}{T}\sum_{t=1}^{T}y_{t-1}u_{t}+\frac{1}{T}\sum_{t=1}^{T}\Delta y_{t}y_{t}^{-}\nonumber \\
 & =_{(1)}\frac{1}{T}\sum_{t=1}^{T}y_{t-1}u_{t}+o_{p}(1)\indist_{(2)}\sigma\int_{0}^{1}Y_{\theta_{\phi}}(r)\diff W(r)\defeq\mathcal{Z}_{\theta_{\phi}}^{(\beta)}\label{eq:denom2}
\end{align}
where $\indist_{(2)}$ holds by Theorem \ref{thm:ytARk} and \citet[Theorem 2.1]{ito_convergence},
and $=_{(1)}$ since
\begin{equation}
\abs{\frac{1}{T}\sum_{t=1}^{T}\Delta y_{t}y_{t}^{-}}^{2}\leq\frac{1}{T}\sum_{t=1}^{T}(\Delta y_{t})^{2}\frac{1}{T}\sum_{t=1}^{T}(y_{t}^{-})^{2}=o_{p}(1)\label{eq:dely-yminus}
\end{equation}
by the CS inequality and Lemma \ref{lem:ARkterms}\enuref{ytminsq}--\enuref{crossdelta}.
Finally, for $s\in\{1,\ldots,k-1\}$,
\begin{align*}
\frac{1}{T}\sum_{t=1}^{T}\Delta y_{t-s}(u_{t}-y_{t}^{-}) & =\frac{1}{T}\sum_{t=1}^{T}\Delta y_{t-s}u_{t}-\frac{1}{T}\sum_{t=1}^{T}\Delta y_{t-s}y_{t}^{-}=\frac{1}{T}\sum_{t=1}^{T}\Delta y_{t-s}u_{t}+o_{p}(1),
\end{align*}
by the same argument as which yielded \eqref{eq:dely-yminus}. Further,
by Lemma \ref{lem:qdnegl},
\[
\expect\left(\frac{1}{T}\sum_{t=1}^{T}\Delta y_{t-s}u_{t}\right)^{2}=\frac{\sigma^{2}}{T^{2}}\sum_{t=1}^{T}\expect(\Delta y_{t-s})^{2}=O(T^{-1}).
\]
Hence
\begin{equation}
\frac{1}{T}\sum_{t=1}^{T}\Delta y_{t-s}(u_{t}-y_{t}^{-})=O_{p}(T^{-1/2})+o_{p}(1)=o_{p}(1).\label{eq:denom3}
\end{equation}

Letting $\mathcal{Z}_{T}\defeq(\mathcal{Z}_{T}^{(\alpha)},\mathcal{Z}_{T}^{(\beta)})^{\trans}$
and $\mathcal{Z}_{\theta_{\phi}}\defeq(\mathcal{Z}_{\theta_{\phi}}^{(\alpha)},\mathcal{Z}_{\theta_{\phi}}^{(\beta)})^{\trans}$,
it follows from \eqref{eq:Dm}, \eqref{eq:denom1}, \eqref{eq:denom2}
and \eqref{eq:denom3} that
\[
D_{1,T}^{-1}m_{T}=\begin{bmatrix}\mathcal{Z}_{T}\\
o_{p}(1)
\end{bmatrix}\indist\begin{bmatrix}\mathcal{Z}_{\theta_{\phi}}\\
0
\end{bmatrix}.
\]
Therefore, recalling \eqref{eq:OLS}, and using \eqref{eq:YT}--\eqref{eq:MTstd},
we obtain
\begin{align*}
D_{2,T}(\hat{\mu}_{T}-\mu) & =\left(\begin{bmatrix}\mathcal{Y}_{T}^{-1} & -\mathcal{Y}_{T}^{-1}\Xi_{T}\Omega_{T}^{-1}\\
0 & \Omega_{T}^{-1}
\end{bmatrix}+o_{p}(1)\right)\begin{bmatrix}\mathcal{Z}_{T}\\
o_{p}(1)
\end{bmatrix}\\
 & =\begin{bmatrix}\mathcal{Y}_{T}^{-1}\mathcal{Z}_{T}\\
0
\end{bmatrix}+o_{p}(1)\indist\begin{bmatrix}\mathcal{Y}_{\theta_{\phi}}^{-1}\mathcal{Z}_{\theta_{\phi}}\\
0
\end{bmatrix}.\qedhere
\end{align*}
\end{proof}
\begin{proof}[Proof of Corollary \ref{thm:tstatARk} ($k\geq 2$)]
We first show that $\hat{\sigma}_{T}^{2}\inprob\sigma^{2}$. Adapting
the argument from the $k=1$ case, we have
\begin{align*}
\sum_{t=1}^{T}[\hat{u}_{t}^{2}-(u_{t}-y_{t}^{-})^{2}] & =\sum_{t=1}^{T}[(\alpha-\hat{\alpha}_{T})+(\beta-\hat{\beta}_{T})y_{t-1}+(\vec{\phi}-\hat{\vec{\phi}}_{T})^{\trans}\Delta\vec y_{t-1}][\hat{u}_{t}+(u_{t}-y_{t}^{-})]\\
 & =(\alpha-\hat{\alpha}_{T})\sum_{t=1}^{T}(u_{t}-y_{t}^{-})+(\beta-\hat{\beta}_{T})\sum_{t=1}^{T}y_{t-1}(u_{t}-y_{t}^{-})\\
 & \qquad\qquad+(\vec{\phi}-\hat{\vec{\phi}}_{T})^{\trans}\sum_{t=1}^{T}\Delta\vec y_{t-1}(u_{t}-y_{t}^{-})\\
 & =O_{p}(T^{-1/2})O_{p}(T^{1/2})+O_{P}(T^{-1})O_{p}(T)+o_{p}(1)o_{p}(T)=o_{p}(T)
\end{align*}
where the the orders of the sums follow from Lemma \ref{lem:ARkterms}\enuref{numerator},
\eqref{eq:denom2} and \eqref{eq:denom3}, and the rates of convergence
of the OLS estimators from Theorem \ref{thm:olsARk}. It follows that
\[
\frac{1}{T}\sum_{t=1}^{T}\hat{u}_{t}^{2}=\frac{1}{T}\sum_{t=1}^{T}(u_{t}-y_{t}^{-})^{2}+o_{p}(1)=\frac{1}{T}\sum_{t=1}^{T}u_{t}^{2}+o_{p}(1)\inprob\sigma^{2}
\]
by Lemma \ref{lem:ARkterms}\enuref{ytminsq}, the LLN and the CS
inequality.

To complete the proof, we note that since $\phi(1)^{-1}J_{\theta_{\phi}}(r)=Y_{\theta_{\phi}}(r)$,
\[
\begin{split}\mathcal{Y}_{\theta_{\phi}} & =\begin{bmatrix}1 & \int_{0}^{1}Y_{\theta_{\phi}}(r)\diff r\\
\int_{0}^{1}Y_{\theta_{\phi}}(r)\diff r & \int_{0}^{1}Y_{\theta_{\phi}}^{2}(r)\diff r
\end{bmatrix}\\
 & =\begin{bmatrix}1 & 0\\
0 & \phi(1)^{-1}
\end{bmatrix}\begin{bmatrix}1 & \int_{0}^{1}J_{\theta_{\phi}}(r)\diff r\\
\int_{0}^{1}J_{\theta_{\phi}}(r)\diff r & \int_{0}^{1}J_{\theta_{\phi}}^{2}(r)\diff r
\end{bmatrix}\begin{bmatrix}1 & 0\\
0 & \phi(1)^{-1}
\end{bmatrix}.
\end{split}
\]
It follows that the result of Theorem \ref{thm:olsARk} can be rewritten
as
\[
\begin{bmatrix}T^{1/2}(\hat{\alpha}_{T}-\alpha)\\
T(\hat{\beta}_{T}-\beta)
\end{bmatrix}\indist\begin{bmatrix}\mathfrak{a}_{\theta_{\phi}}\\
\phi(1)\mathfrak{b}_{\theta_{\phi}}
\end{bmatrix}.
\]
By \eqref{eq:YT} and \eqref{eq:MTstd}, the upper left $2\times2$
block of $D_{1,T}^{-1}\mathcal{M}_{T}D_{2,T}^{-1}$ converges to $\mathcal{Y}_{\theta_{\phi}}$.
Thus, $D_{2,T}\mathcal{M}_{T}^{-1}D_{1,T}$ converges to $\mathcal{Y}_{\theta_{\phi}}^{-1}$,
and so $T\mathcal{M}_{T}^{-1}(1,1)\indist\mathcal{J}_{\theta_{\phi}}^{-1}(1,1)$
and $T^{2}\mathcal{M}_{T}^{-1}(2,2)\indist\phi(1)^{2}\mathcal{J}_{\theta_{\phi}}^{-1}(2,2)$.
Hence, the result follows by the CMT.
\end{proof}

\section{Role of the joint spectral radius}\label{sec:appendix_JSR}

\begin{example}[Stationary roots of $\phi(z)$ are not sufficient for Theorem \ref{ytARk}]\label{ex:stat_not_suff}

\emph{Let $\beta=1$ and $$B(z)=(1-z)\phi(z),\quad\phi(z)=(1-z+0.9z^2)(1-1.3z+0.9z^2).$$ The roots of $\phi(z)$ are larger than $1$ in absolute value (roots are approximately $0.56\pm0.9 i$ and $0.7\pm 0.77i$). However, simulations in Figure \ref{fig:JSR_ex2} indicate that $\{\Delta y_{t}\}$ is not stochastically bounded and $\{ y_{t}\}$ grows exponentially, preventing the result of Theorem~\ref{ytARk} from holding.}
\end{example}
\begin{figure}[h]
\begin{subfigure}{.45\textwidth} \centering \includegraphics[width=1\linewidth]{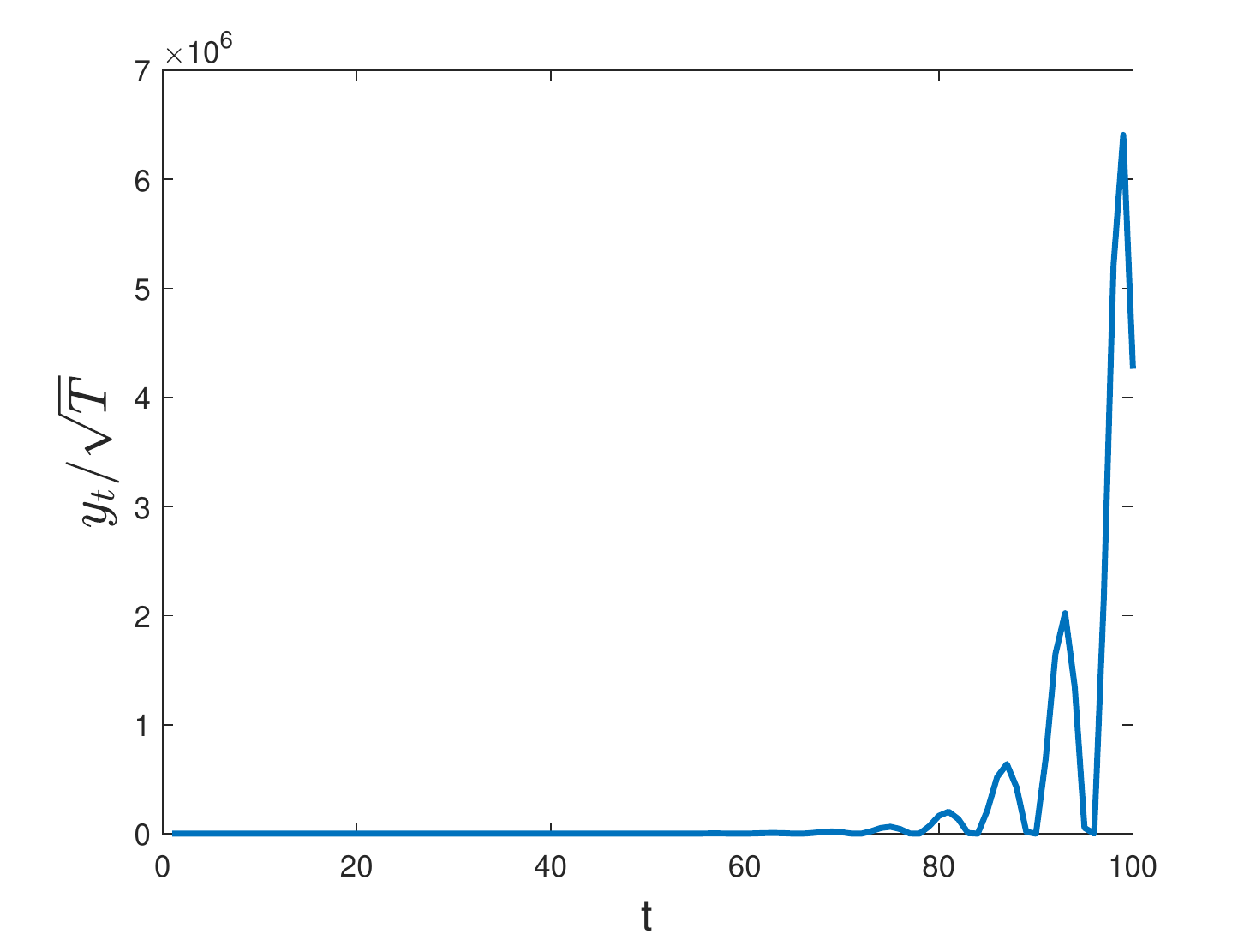}
\caption{The rescaled by $\sqrt{T}$ process explodes.}
\label{fig:JSR_ex2_y}
\end{subfigure}
\begin{subfigure}{.45\textwidth} \centering \includegraphics[width=1\linewidth]{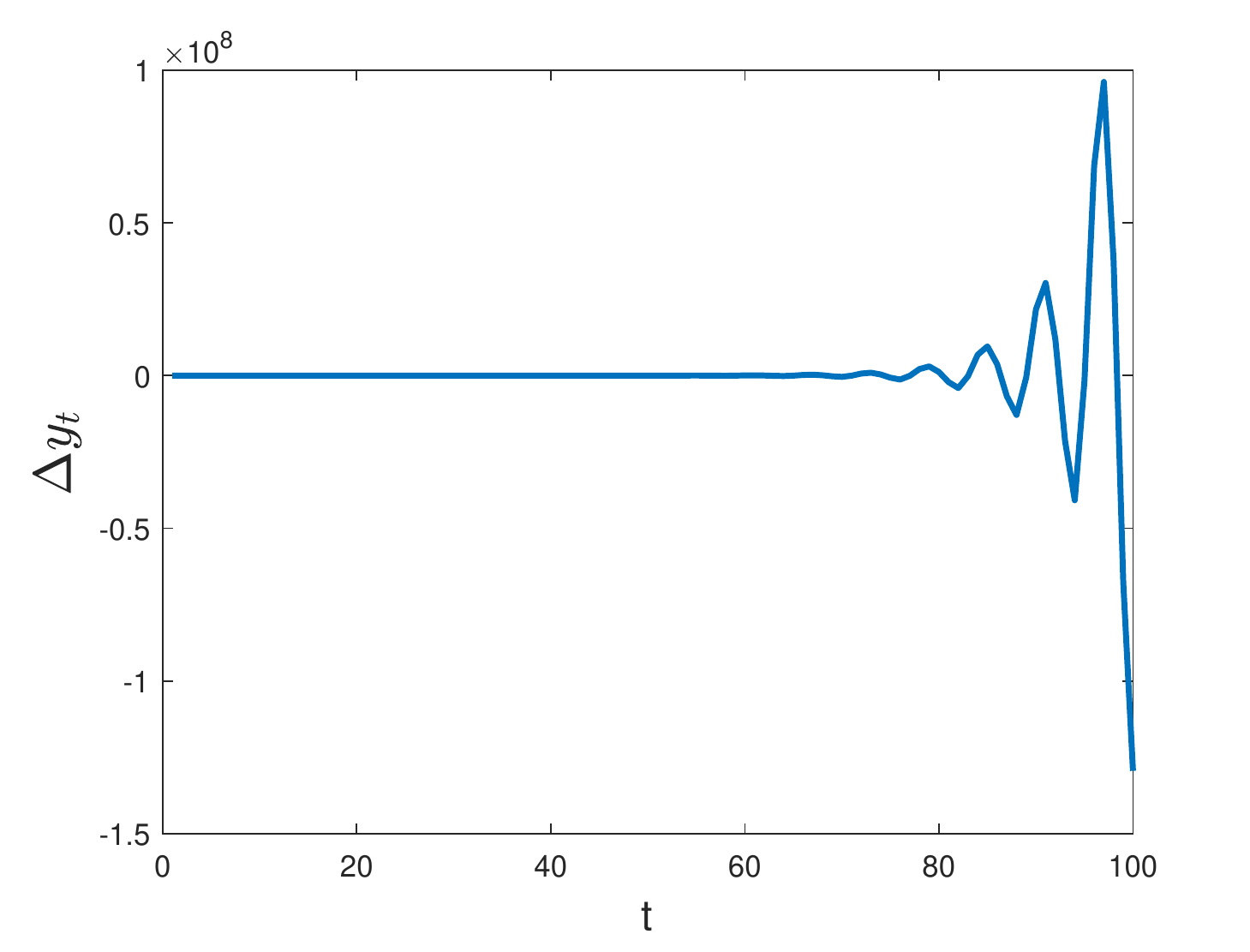}
\caption{First differences are not bounded.}
\label{fig:JSR_ex2_dy}
\end{subfigure}
\caption{Stationary roots of $\phi(z)$ are not sufficient for Theorem \ref{ytARk}. Data generating process corresponds to Example \ref{ex:stat_not_suff} with $T=1000$, $y_{0}=0$, $u_{t}\thicksim\text{i.i.d.}~\mathcal{N}(0,1)$.}
\label{fig:JSR_ex2}
\end{figure}

\begin{example}[Assumption \ref{ass:JSR} is not necessary for Theorem \ref{ytARk}]\label{ex:JSR_not_nec}

\emph{Let $$B(z)=(1-z)\phi(z),\quad\phi(z)=1-1.3z+0.8z^2.$$ That is, $k=3$, $\beta=1$, $\phi_{1}=1.3$ and $\phi_{2}=-0.8$. The roots of $\phi(z)$ are larger than $1$ in absolute value (roots are approximately $0.8\pm0.77i$). However, the largest (in modulus) eigenvalue of $F_{1}F_{1}F_{0}$ is $-1.04$, and so $\lambda_{\jsr}(\{F_{0},F_{1}\})\geq\smlabs{-1.04}^{1/3}>1$.
On the other hand, simulations in Figure \ref{fig:JSR_ex1} indicate that $\{\Delta y_{t}\}$ is stochastically bounded. Thus, Assumption \ref{ass:JSR} is not necessary for Theorem \ref{ytARk}.}
\end{example}
\begin{figure}[t]
\begin{subfigure}{.45\textwidth} \centering \includegraphics[width=1\linewidth]{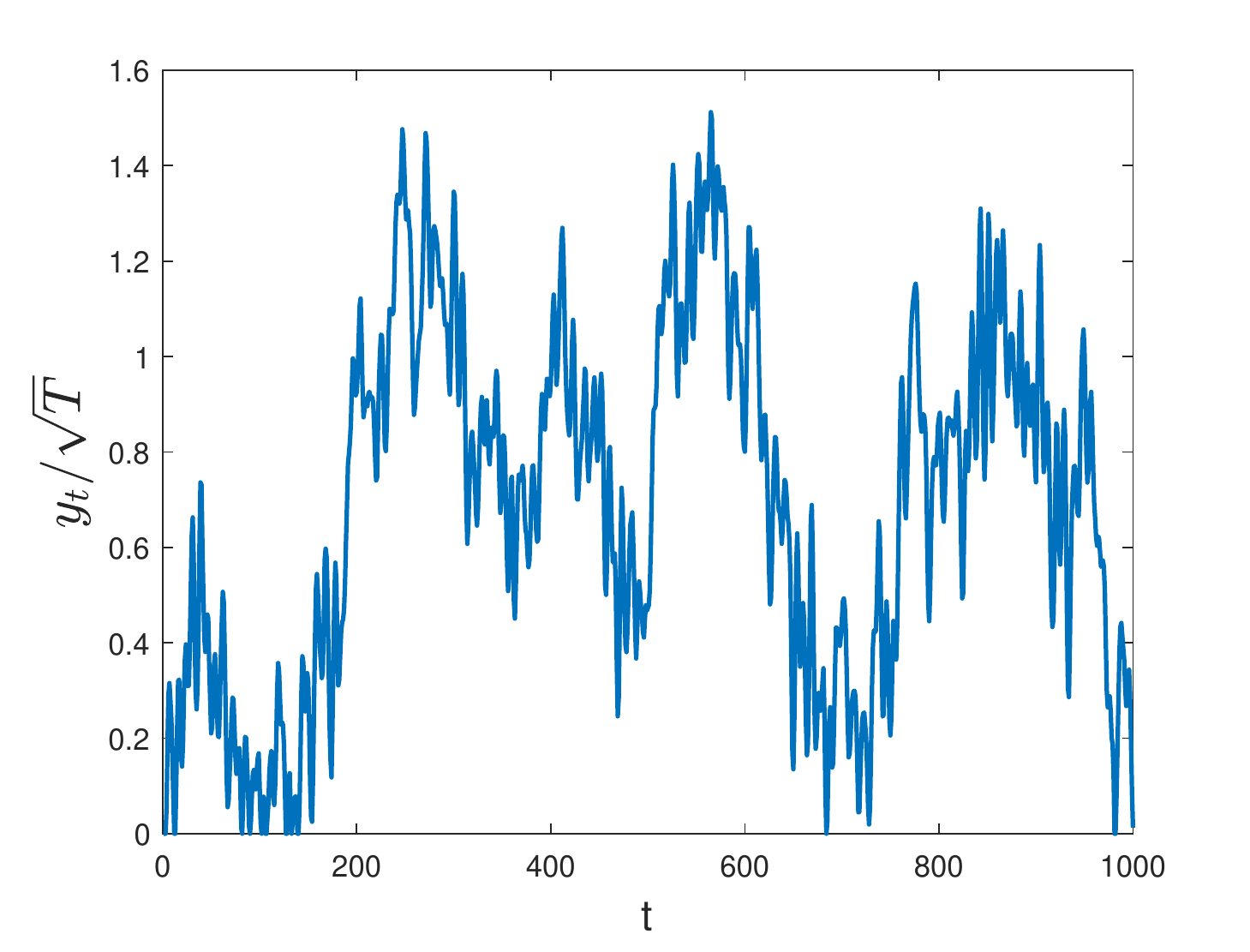}
\caption{The rescaled by $\sqrt{T}$ process.}
\label{fig:JSR_ex1_y}
\end{subfigure}
\begin{subfigure}{.45\textwidth} \centering \includegraphics[width=1\linewidth]{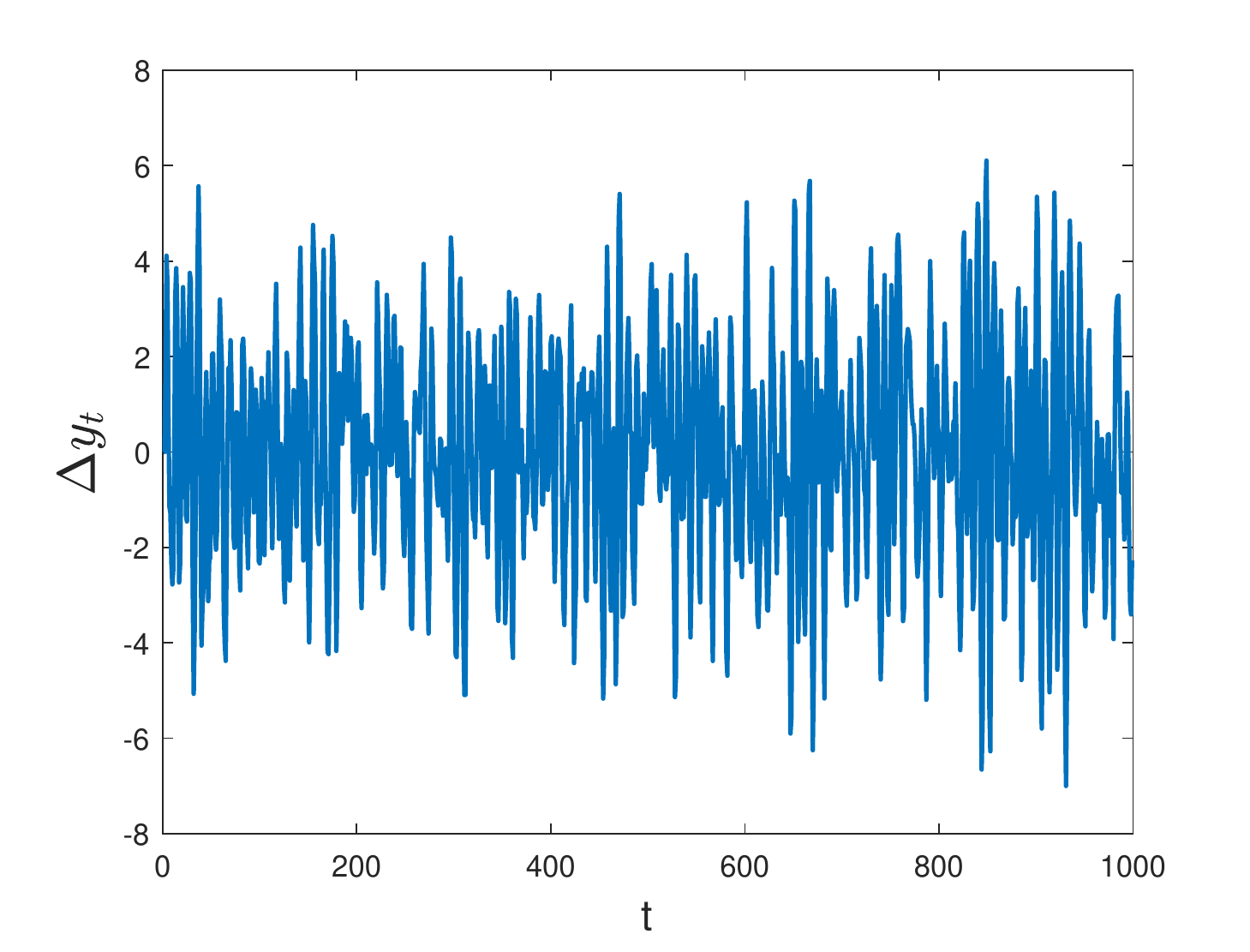}
\caption{First differences are bounded.}
\label{fig:JSR_ex1_dy}
\end{subfigure}
\caption{Assumption \ref{ass:JSR} is not necessary for Theorem \ref{ytARk}. Data generating process corresponds to Example \ref{ex:JSR_not_nec} with $T=1000$, $y_{0}=0$, $u_{t}\thicksim\text{i.i.d.}~\mathcal{N}(0,1)$.}
\label{fig:JSR_ex1}
\end{figure}

\begin{proof}[Proof of Lemma \ref{lem:JSR_holds}]
Let $\sum_{i=1}^{k-1}|\phi_i|=\Phi<1$ and $M\in\mathcal{A}^{n}=\{\prod_{j=1}^{n}A_{i}\mid A_{j}\in\{F_0,F_1\}\}$. Let us show that the spectral radius of $M$, $\lambda(M)$, is at most $\Phi^{\lfloor n/(k-1) \rfloor}$, where $\lfloor\cdot\rfloor$ is a floor function. Then,
$$\lambda_{\jsr}(\{F_0,F_1\})=\limsup_{n\goesto\infty}\sup_{M\in\mathcal{A}^{n}}\lambda(M)^{1/n}
\leq\limsup_{n\goesto\infty}\left(\Phi^{\left\lfloor\tfrac{n}{k-1}\right\rfloor}\right)^{1/n}=\Phi^{\tfrac1{k-1}}<1.$$

First, suppose $n=k-1$, so that $M=F_{\delta_{k-1}}\cdot\ldots\cdot F_{\delta_1},$ where $\delta_j\in\{0,1\},\,j=1,\ldots,k-1$. Let us show that for any vector $x=(x_1,\ldots,x_{k-1})^{\trans}\in\mathbb{R}^{k-1}$, $\|Mx\|_{\infty}\leq\Phi\|x\|_{\infty}$, where $\|x\|_{\infty}:=\max\limits_{i=1,\ldots,k-1}\{|x_i|\}$, i.e.~the maximum norm $\ell_{\infty}$. Let $x^s:=F_{\delta_s}\cdot\ldots\cdot F_{\delta_1}x$, $s=1,\ldots,k-1$.

We prove by induction that for any $s=1,\ldots,k-1$, $|x^s_i|\leq\Phi\|x\|_{\infty}$ for $i=1,\ldots,s$ and $|x^s_i|\leq\|x\|_{\infty}$ for $i>s$. To verify induction base, note that  since
\begin{equation}\label{eq:Fdelta_x}
F_{\delta}x=\left(\delta\phi_1x_1+\sum_{i=2}^{k-1}\phi_i x_i,\,\delta x_1,\,x_2,\ldots,x_{k-2}\right)^{\trans},
\end{equation}
$|x^1_1|\leq\Phi\|x\|_{\infty}$ and $|x^1_i|\leq|x_{i-1}|\leq\|x\|_{\infty}$ for $i>1$. To verify induction step, suppose that the claim holds for $s$ and let us prove it for $s+1$. Using $x^{s+1}=F_{\delta_{s+1}}x^s$ and \eqref{eq:Fdelta_x}, we get $|x^{s+1}_1|\leq\Phi\|x^s\|_{\infty}\leq\Phi\|x\|_{\infty}$, where the last inequality follows from induction hypothesis. Next, $|x^{s+1}_2|=|\delta x^{s}_1|\leq|x^{s}_1|\leq\Phi\|x\|_{\infty}$ and for $i>2$ we have $|x^{s+1}_i|=|x^{s}_{i-1}|$. Thus, for $i>s+1$ we have $|x^{s+1}_i|=|x^{s}_{i-1}|\leq\|x\|_{\infty}$, while for $2<i\leq s+1$, $|x^{s+1}_i|=|x^{s}_{i-1}|\leq\Phi\|x\|_{\infty}$.

It follows from the preceding that since $Mx=x^{k-1}$, all coordinates of $Mx$ are bounded by $\Phi\|x\|_{\infty}$ and $\|Mx\|_{\infty}\leq\Phi\|x\|_{\infty}$.

Now consider general $n$ and $M\in\mathcal{A}^{n}$, $M=F_{\delta_n}\cdot\ldots\cdot F_{\delta_1}$. We know that for any $x\in\mathbb{R}^{k-1}$,
$\|F_{\delta_{k-1}}\cdot\ldots\cdot F_{\delta_1}x\|_{\infty}\leq\Phi\|x\|_{\infty}$. Thus, $\|F_{\delta_{2(k-1)}}\cdot\ldots\cdot F_{\delta_1}x\|_{\infty}\leq\Phi\|F_{\delta_{k-1}}\cdot\ldots\cdot F_{\delta_1}x\|_{\infty}\leq\Phi^2\|x\|_{\infty}$ and letting $\tilde{n}=\left\lfloor\tfrac{n}{k-1}\right\rfloor$ by iterative back-substitution we get $\|F_{\delta_{\tilde{n}(k-1)}}\cdot\ldots\cdot F_{\delta_1}x\|_{\infty}\leq\Phi^{\tilde{n}}\|x\|_{\infty}$. Finally note from \eqref{eq:Fdelta_x} that applying $F_\delta$ cannot increase the maximum norm, so that
$$\|Mx\|_{\infty}=\|F_{\delta_n}\cdot\ldots\cdot F_{\delta_1}x\|_{\infty}\leq\|F_{\delta_{\tilde{n}(k-1)}}\cdot\ldots\cdot F_{\delta_1}x\|_{\infty}\leq\Phi^{\tilde{n}}\|x\|_{\infty}.$$

If $x$ is an eigenvector of $M$ with the corresponding eigenvalue $\lambda$, we must have $Mx=\lambda x$ and $\|Mx\|_{\infty}=|\lambda|\cdot\| x\|_{\infty}$. Since $\|Mx\|_{\infty}\leq\Phi^{\tilde{n}}\|x\|_{\infty}$, we must also have $|\lambda|\leq\Phi^{\tilde{n}}$. Therefore, the spectral radius of $M$, $\lambda(M)$, is at most $\Phi^{\left\lfloor\tfrac{n}{k-1}\right\rfloor}$, which completes the proof.
\end{proof}

\bibliographystyle{ecca}
\bibliography{LUR_tobit_biblio}

\end{document}